\documentclass[a4paper,11pt]{article}
\usepackage[top=1in,bottom=1in,left=1in,right=1in]{geometry}
\usepackage{authblk}

\usepackage[utf8x]{inputenc}
\usepackage[section]{algorithm}
\usepackage{algpseudocode}
\usepackage{enumitem}
\usepackage{amsfonts}
\usepackage{amsmath}
\usepackage{amssymb}
\usepackage{amsthm}
\usepackage{verbatim}
\usepackage{graphicx}
\usepackage{float}
\usepackage{caption}
\usepackage{subcaption}

\newtheorem{theorem}{Theorem}
\newtheorem{lemma}[theorem]{Lemma}
\newtheorem{corollary}[theorem]{Corollary}

\title{Constant Approximation Algorithms for Guarding Simple Polygons using Vertex Guards} 
\author[1]{Pritam Bhattacharya\thanks{Email: pritam.bhattacharya@cse.iitkgp.ernet.in, lord.pritomose@gmail.com}} 
\author[2]{Subir Kumar Ghosh \thanks{Email: subir.ghosh@rkmvu.ac.in, profsubirghosh@gmail.com}}
\author[1]{Sudebkumar Prasant Pal \thanks{Email: spp@cse.iitkgp.ernet.in, sudebkumar@gmail.com}}
\affil[1]{Department of Computer Science and Engineering, Indian Institute of Technology, Kharagpur, West Bengal - 721302, India.}
\affil[2]{Department of Computer Science, RKM Vivekananda Educational and Research Institute, Belur, West Bengal - 711202, India.}
\date{}

\begin{document}

\maketitle

\vspace*{-2.2em}
\begin{abstract}
\vspace*{-0.33em}
The art gallery problem enquires about the least number of guards sufficient to ensure that an art gallery, represented by a simple polygon $P$, is fully guarded.
Most standard versions of this problem are known to be NP-hard. 
In 1987, Ghosh provided a deterministic $\mathcal{O}(\log n)$-approximation algorithm for the case of vertex guards and edge guards in simple polygons. 
In the same paper, Ghosh also conjectured the existence of constant ratio approximation algorithms for these problems.
We present here three polynomial-time algorithms with a constant approximation ratio for guarding an $n$-sided simple polygon $P$ using vertex guards.
{(i)} The first algorithm, that has an approximation ratio of 18, guards all vertices of $P$ in $\mathcal{O}(n^4)$ time.
{(ii)} The second algorithm, that has the same approximation ratio of 18, guards the entire boundary of $P$ in $\mathcal{O}(n^5)$ time.
{(iii)} The third algorithm, that has an approximation ratio of 27, guards all interior and boundary points of $P$ in $\mathcal{O}(n^5)$ time.
Further, these algorithms can be modified to obtain similar approximation ratios while using edge guards. 
The significance of our results lies in the fact that these results settle the conjecture by Ghosh 
regarding the existence of constant-factor approximation algorithms for this problem, 
which has been open since 1987 despite several attempts by researchers. 
Our approximation algorithms exploit several deep visibility structures of simple polygons which are interesting in their own right.
\end{abstract}


\vspace*{-0.9em}
\section{Introduction}
\label{intro}

\vspace*{-0.55em}
\subsection{The art gallery problem and its variants}
\label{agp}

\vspace*{-0.44em}
The art gallery problem enquires about the least number of guards sufficient to ensure that an art gallery 
(represented by a polygon $P$) is fully guarded, assuming that a guard’s field of view covers 360\textdegree\:as well 
as unbounded distance. This problem was first posed by Victor Klee in a conference in 1973, and has become a well investigated problem in computational geometry. \\

\vspace*{-0.66em}
A \emph{polygon} $P$ is defined to be a closed region in the plane bounded by 
a finite set of line segments, called edges of $P$, such that between any two points of $P$, 
there exists a path which does not intersect any edge of $P$. 
If the boundary of a polygon $P$ consists of two or more cycles, then $P$ is called a \emph{polygon with holes} 
(see Figure \ref{sp_a}).
Otherwise, $P$ is called a \emph{simple polygon} or a \emph{polygon without holes} (see Figure \ref{sp_b}). \\

\vspace*{-0.66em}
An art gallery can be viewed as an $n$-sided polygon $P$ (with or without holes) and guards as points inside $P$.
Any point $z \in P$ is said to be \emph{visible} from a guard $g$ if the line segment ${z g}$ does not intersect 
the exterior of $P$. 
In general, guards may be placed anywhere inside $P$. 
If the guards are allowed to be placed only on vertices of $P$, they are called \emph{vertex guards}. If there is 
no such restriction, then they are called \emph{point guards}. The point guards that are constrained to lie on 
the boundary of $P$, but not necessarily at the vertices, are referred to as \emph{perimeter guards}. 
Point, vertex and perimeter guards together are also referred to as \emph{stationary guards}. 
If guards are allowed to patrol along a line segment inside $P$, they are called \emph{mobile guards}. 
If they are allowed to patrol only along the edges of $P$, they are called \emph{edge guards} 
\cite{Ghosh_2007,O'Rourke_1987}. 

\begin{figure}[H]
\begin{minipage}{.37\textwidth}
\centerline{\includegraphics[width=\textwidth]{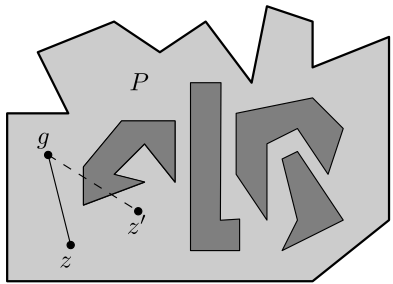}}
\caption{Polygon with holes}
\label{sp_a}
\end{minipage}
\hspace*{.19\textwidth}
\begin{minipage}{.37\textwidth}
\centerline{\includegraphics[width=\textwidth]{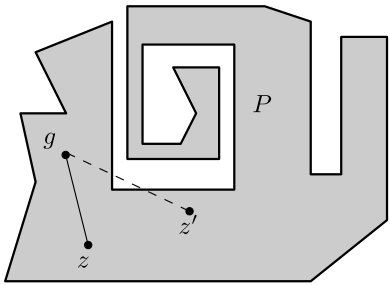}}
\caption{Polygon without holes}
\label{sp_b}
\end{minipage}
\end{figure}

\vspace{-0.66em}
In 1975, Chvátal \cite{Chvatal_1975} showed that $\lfloor\frac{n}{3}\rfloor$ stationary guards 
are sufficient and sometimes necessary (see Figure \ref{sg_a})
for guarding a simple polygon. 
In 1978, Fisk \cite{Fisk_1978} presented a simpler and more elegant proof of this result. 
For a simple orthogonal polygon, whose edges are either horizontal or vertical, Kahn et al. \cite{KKK_1983} 
and also O’Rourke \cite{O'Rourke_1983} showed that $\lfloor\frac{n}{4}\rfloor$ stationary guards 
are sufficient and sometimes necessary (see Figure \ref{sg_b}). 

\begin{figure}[H]
\begin{minipage}{.33\textwidth}
\centerline{\includegraphics[width=\textwidth]{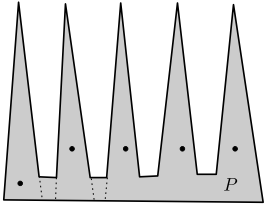}}
\caption{A polygon where $\lfloor\frac{n}{3}\rfloor$ stationary guards are necessary.}
\label{sg_a}
\end{minipage}
\hspace*{.22\textwidth}
\begin{minipage}{.38\textwidth}
\centerline{\includegraphics[width=\textwidth]{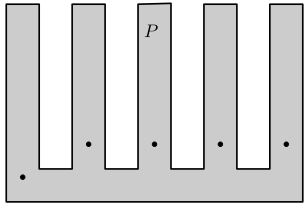}}
\caption{A polygon where $\lfloor\frac{n}{4}\rfloor$ stationary guards  are necessary.}
\label{sg_b}
\end{minipage}
\end{figure}

\vspace*{-0.44em}
\subsection{Related hardness and approximation results}
\label{rhar}
\vspace*{-0.33em}
The decision version of the art gallery problem is to determine, given a polygon $P$ and a number $k$ as input, whether 
the polygon $P$ can be guarded with $k$ or fewer point guards. 
This problem was first shown to be NP-hard for polygons with holes by O’Rourke and Supowit \cite{RS_1983}. 
This problem was also shown to be NP-hard for simple polygons for guarding using only vertex guards by Lee and Lin \cite{LL_1986}. Their proof was generalized to work for point guards by Aggarwal \cite{Aggarwal_1984}.
The problem was shown to be NP-hard even for simple orthogonal polygons by Katz and Roisman \cite{KR_2008} and Schuchardt and Hecker \cite{SH_1995}.
Abrahamsen, Adamaszek and Miltzow \cite{ETR-complete} have recently shown that the art gallery problem for point guards is ETR-complete. \\ 

In 1987, Ghosh \cite{Ghosh_1987,Ghosh_2010} provided a deterministic $\mathcal{O}(\log n)$-approximation algorithm for the case of vertex and edge guards by discretizing the input polygon $P$ and treating it as an instance of the Set Cover problem. 
As pointed out by King and Kirkpatrick \cite{KK_2011}, newer methods for improving the approximation ratio of the Set Cover problem itself have been developed in the time after Ghosh’s algorithm was published.
By applying these methods, the approximation ratio of Ghosh's algorithm becomes $\mathcal{O}(\log OPT)$ for guarding simple polygons
and $\mathcal{O}(\log h \log OPT)$ for guarding a polygon with $h$ holes, where $OPT$ denotes the size of the smallest guard 
set for $P$. Deshpande et al. \cite{DKDS_2007} obtained an approximation factor of $\mathcal{O}(\log OPT)$ for point guards or 
perimeter guards by developing a sophisticated discretization method that runs in pseudo-polynomial time. Efrat and Har-Peled \cite{EH_2006} provided a randomized algorithm with the same approximation ratio that runs in fully polynomial expected time. 
Bonnet and Miltzow \cite{BM_2017} obtained an approximation factor of $\mathcal{O}(\log OPT)$ for the point guard problem assuming integer coordinates and a specific general position.   
For guarding simple polygons using perimeter guards, King and Kirkpatrick \cite{KK_2011} designed a deterministic $\mathcal{O}(\log\log OPT)$-approximation 
algorithm in 2011. The analysis of this result was simplified by Kirkpatrick \cite{K_15}. \\ 

\vspace*{-0.5em}
In 1998, Eidenbenz, Stamm and Widmayer \cite{ESW_1998,ESW_2001} proved that the problem is APX-complete, implying that an approximation 
ratio better than a fixed constant cannot be achieved unless NP~=~P. 
They also proved that if the input polygon is allowed to 
contain holes, then there cannot exist a polynomial time algorithm for the problem with an approximation ratio better than 
$((1−\epsilon)/12)\ln n$ for any $\epsilon>0$, unless NP $\subseteq$ TIME($n^{\mathcal{O}(\log\log n)}$). Extending their method,
Bhattacharya, Ghosh and Roy \cite{VGinWVP} proved that, even for the special subclass of polygons with holes that are 
weakly visible from an edge, there cannot exist a polynomial time algorithm for the problem with an approximation ratio better 
than $((1−\epsilon)/12)\ln n$ for any $\epsilon>0$, unless NP~=~P. These inapproximability results establish that the approximation
ratio of $\mathcal{O}(\log n)$ obtained by Ghosh in 1987 is in fact the best possible for the case of polygons with holes. However, for simple 
polygons, the existence of a constant factor approximation algorithm for vertex and edge guards was conjectured 
by Ghosh \cite{Ghosh_1987,GhoshW_2010} in 1987. \\

\vspace*{-0.5em}
Ghosh's conjecture has been shown to be true for vertex guarding in two special sub-classes of simple polygons, viz. monotone polygons and polygons weakly visible from an edge. In 2012, Krohn and Nilsson \cite{KrohnNilsson_2013} presented an approximation algorithm that computes in polynomial time a guard set for a monotone polygon $P$, such that the size of the guard set is at most $30\cdot OPT$. 
Bhattacharya, Ghosh and Roy \cite{VGinWVP,AGWVP} presented a  
6-approximation algorithm that runs in $\mathcal{O}(n^2)$ time for vertex guarding simple polygons that are weakly visible from an edge.
For vertex guarding this subclass of simple polygons that are weakly visible from an edge, a PTAS has recently been proposed by Katz \cite{Katz_PTAS}.

\subsection{Our contributions}
\label{contributions}

\vspace*{-.4em}
In this paper, we present three polynomial-time algorithms with a constant approximation ratio for guarding an $n$-sided simple polygon $P$ using vertex guards. 
The first algorithm, that has an approximation ratio of 18, guards all vertices of $P$ in $\mathcal{O}(n^4)$ time. 
The second algorithm, that has the same approximation ratio of 18, 
guards the entire boundary of $P$ in $\mathcal{O}(n^5)$ time.
The third algorithm, that has an approximation ratio of 27, guards all interior and boundary points of $P$ 
in $\mathcal{O}(n^5)$ time.
As an extension we show, using similar techniques, 
constant-factor approximation can also be achieved for guarding $P$
we also present identical algorithms, maintaining both the approximation bounds as well as the running times, 
can be obtained using edge guards. 
In particular, we show that the same approximation ratios of 18, 18 and 27 hold
for guarding all vertices, the entire boundary, and the interior of $P$, 
with time complexities $\mathcal{O}(n^4)$, 
$\mathcal{O}(n^5)$ and $\mathcal{O}(n^5)$ respectively.
The significance of our results lies in the fact that these results settle the \emph{long-standing conjecture by Ghosh} \cite{Ghosh_1987} regarding the existence of constant-factor approximation algorithms for these problem, which has been open since 1987 despite several attempts by researchers. \\

\vspace*{-.5em}
In each of our algorithms, $P$ is first partitioned into a hierarchy of \emph{weak visibility polygons} according to the \emph{link distance} from a starting vertex (see Figure \ref{windows}). 
This partitioning is very similar to the \emph{window 
partitioning} given by Suri \cite{Suri_1986,Suri_1987} in the context of computing minimum link paths. Then, starting with the farthest 
level in the hierarchy (i.e. the set of weak visibility polygons that are at the maximum link distance from the starting vertex), the entire hierarchy is traversed backward level by level, and at each level, 
vertex guards (of two types, viz. \emph{inside} and \emph{outside}) 
are placed for guarding every weak visibility polygon at that level of $P$. 
At every level, a novel procedure is used that has been developed for placing guards in (i) a simple polygon that is weakly visible from an internal chord, 
or (ii) a union of overlapping polygons that are weakly visible from multiple disjoint internal chords.
Note that these chords are actually the constructed edges introduced during the hierarchical partitioning of $P$. \\

\vspace*{-.5em}
Due to partitioning according to link distances, guards can only see points within the adjacent weak visibility polygons in the hierarchical partitioning of $P$. 
This property locally restricts the visibility of the chosen guards, and thereby ensures that the approximation bound on the number of vertex guards placed by our algorithms at any level 
leads directly to overall approximation bounds for guarding $P$. 
Thus, a constant factor approximation bound on the overall number of guards placed by our algorithms is a direct consequence of choosing vertex guards in a judicious manner for guarding each collection of overlapping weak visibility polygons obtained from the hierarchical partitioning of $P$. 
Our algorithms exploit several deep visibility structures of simple polygons which are interesting in their own right.

\vspace*{-.5em}
\subsection{Organization of the paper}
\label{org}

\vspace*{-.2em}
In Section \ref{prelims}, we introduce some preliminary definitions and notations that are used throughout the rest of the paper. 
In Section \ref{partitioning_algo}, we present the hierarchical partitioning of a simple polygon $P$ into weak visibility polygons.
Next, in Section \ref{traverse_hierarchy}, we describe how the algorithm traverses the hierarchy of visibility polygons, starting from the farthest level, and uses the procedures from Section \ref{vertex_algo} at each level as a sub-routine for guarding $P$. 
In Section \ref{vertex_algo}, we present a novel procedure for 
placing vertex guards necessary for guarding a simple polygon $Q$ that is weakly visible from a single internal chord $uv$ or from multiple disjoint chords. 
In Section \ref{final}, we establish the 
overall approximation ratios for the three approximation algorithms. 
In Section \ref{edge_peri}, we show how these algorithms can be modified 
to obtain similar approximation bounds while using edge guards. 
Finally, in Section \ref{conclude}, we conclude the paper with a few remarks. 

\vspace*{-.33em}
\section{Preliminary definitions and notations}
\label{prelims}

\vspace*{-.66em}
Let $P$ be a simple polygon. Assume that the vertices of $P$ are labelled $v_1, v_2,\dots, v_n$ in clockwise order. 
Let $\mathcal{V}(P)$ denote the set of all vertices.
Let $bd_c(p,q)$ (or $bd_{cc}(p,q)$) denote the clockwise (respectively, counterclockwise) boundary of $P$ from a vertex $p$ 
to another vertex $q$. Note that by definition, $bd_c(p,q) = bd_{cc}(q,p)$. Also, we denote the entire boundary of $P$ by 
$bd(P)$. So, $bd(P) = bd_c(p,p) = bd_{cc}(p,p)$ for any chosen vertex $p$ belonging to $P$. \\

\vspace*{-.66em}
The \emph{visibility polygon} of $P$ from a point $z$, denoted as $\mathcal{VP}(z)$, is defined to be the set of all points of $P$ 
that are visible from $z$. In other words, $\mathcal{VP}(z) = \{ q \in P : q \mbox{\: is visible from \:} z \}$.  
Observe that the boundary of $\mathcal{VP}(z)$ consists of polygonal edges and non-polygonal edges. 
We refer to the non-polygonal edges as \emph{constructed edges}. Note that one point of a constructed edge is a vertex 
(say, $v_i$) of $P$, while the other point (say, $u_i$) lies on $bd(P)$. Moreover, the points $z$, $v_i$ and $u_i$ are collinear (see Figure \ref{vp}). \\

\vspace*{-.66em}
Let $bc$ be an internal chord or an edge of $P$.
A point $q$ of $P$ is said to be \emph{weakly visible} from $bc$ if there exists a point $z \in bc$ such that $q$ is visible from $z$.
The set of all such points of $P$ is said to be the \emph{weak visibility polygon} of $P$ from $bc$, and denoted as $\mathcal{VP}(bc)$. 
If $\mathcal{VP}(bc) = P$, 
then $P$ is said to be \emph{weakly visible from $bc$}.
Like $\mathcal{VP}(z)$, the boundary of $\mathcal{VP}(bc)$ also consists of polygonal edges and constructed edges $v_iu_i$ (see Figure \ref{vp}). 
If $z$ (or $bc$) does not belong to $bd_{c}(v_iu_i)$, then $v_iu_i$ is called a \emph{left constructed edge} of $\mathcal{VP}(z)$ (respectively, $\mathcal{VP}(bc)$). Otherwise, $v_iu_i$ is called a \emph{right constructed edge}. 
For any constructed edge $v_iu_i$ of $\mathcal{VP}(bc)$ (or $\mathcal{VP}(z)$), 
observe that $v_iu_i$ divides $P$ into two subpolygons. 
One of the subpolygons is bounded by $bd_c(v_i,u_i)$ and $v_iu_i$, 
whereas the other one is bounded by $bd_{cc}(v_i,u_i)$ and $v_iu_i$. 
Out of these two, the subpolygon that does not contain $bc$ (respectively, $z$) is referred to as 
the \emph{pocket} of $v_iu_i$, and is denoted by $P(v_iu_i)$ (see Figure \ref{vp}). 
If $v_iu_i$ is a left (or right) constructed edge, then $P(v_iu_i)$ is called 
a \emph{left pocket} (or \emph{right pocket}). 

\vspace*{-0.22em}
\begin{figure}[H]
\begin{minipage}{.44\textwidth}
\centerline{\includegraphics[width=1.13\textwidth]{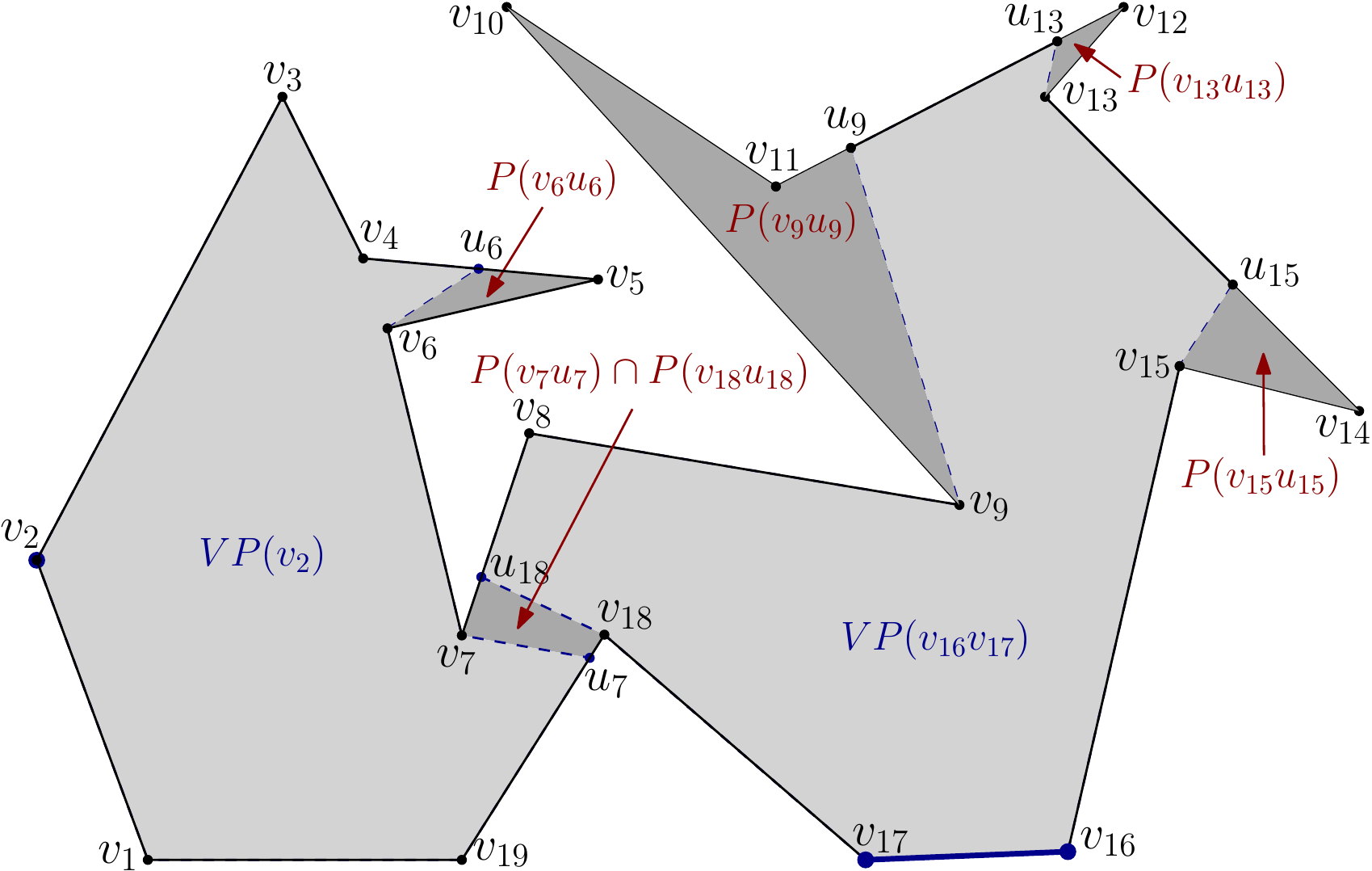}}
\caption{Figure showing visibility polygon $\mathcal{VP}(v_2)$ and weak visibility polygon $\mathcal{VP}(v_{16}v_{17})$, along with several pockets created by constructed edges belonging to both.}
\label{vp}
\end{minipage}
\hspace*{.01\textwidth}
\begin{minipage}{.55\textwidth}
\centerline{\includegraphics[width=1.13\textwidth]{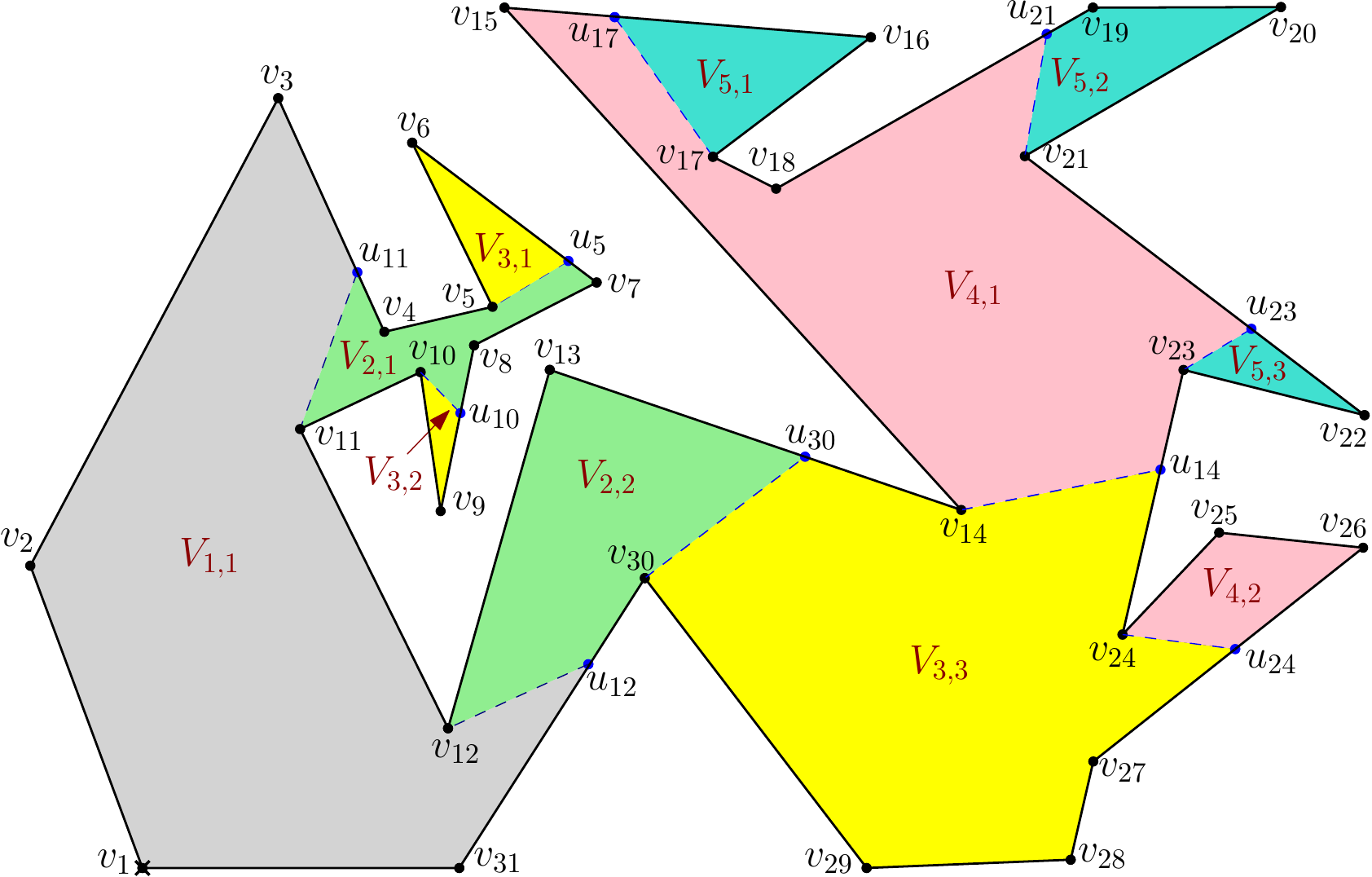}}
\caption{Figure showing the partitioning of a simple polygon into visibility windows.}
\label{windows}
\end{minipage}
\end{figure}

\vspace*{-0.55em}
Let $SP(s,t)$ define the Euclidean shortest path from a point $s$ to another point $t$ within $P$. 
The \emph{shortest path tree} of $P$ rooted at any point $s$ of $P$, denoted by $SPT(s)$, is the union of Euclidean shortest paths 
from $s$ to all vertices of $P$ (see Figure \ref{spt}). 
This union of paths is a planar tree, rooted at $s$, which has $n$ nodes, namely the vertices of $P$. 
For every vertex $x$ of $P$, let $p(s,x)$ denote the parent of $x$ in $SPT(s)$. 

\vspace*{-0.28em}
\begin{figure}[H]
  \centerline{\includegraphics[width=0.48\textwidth]{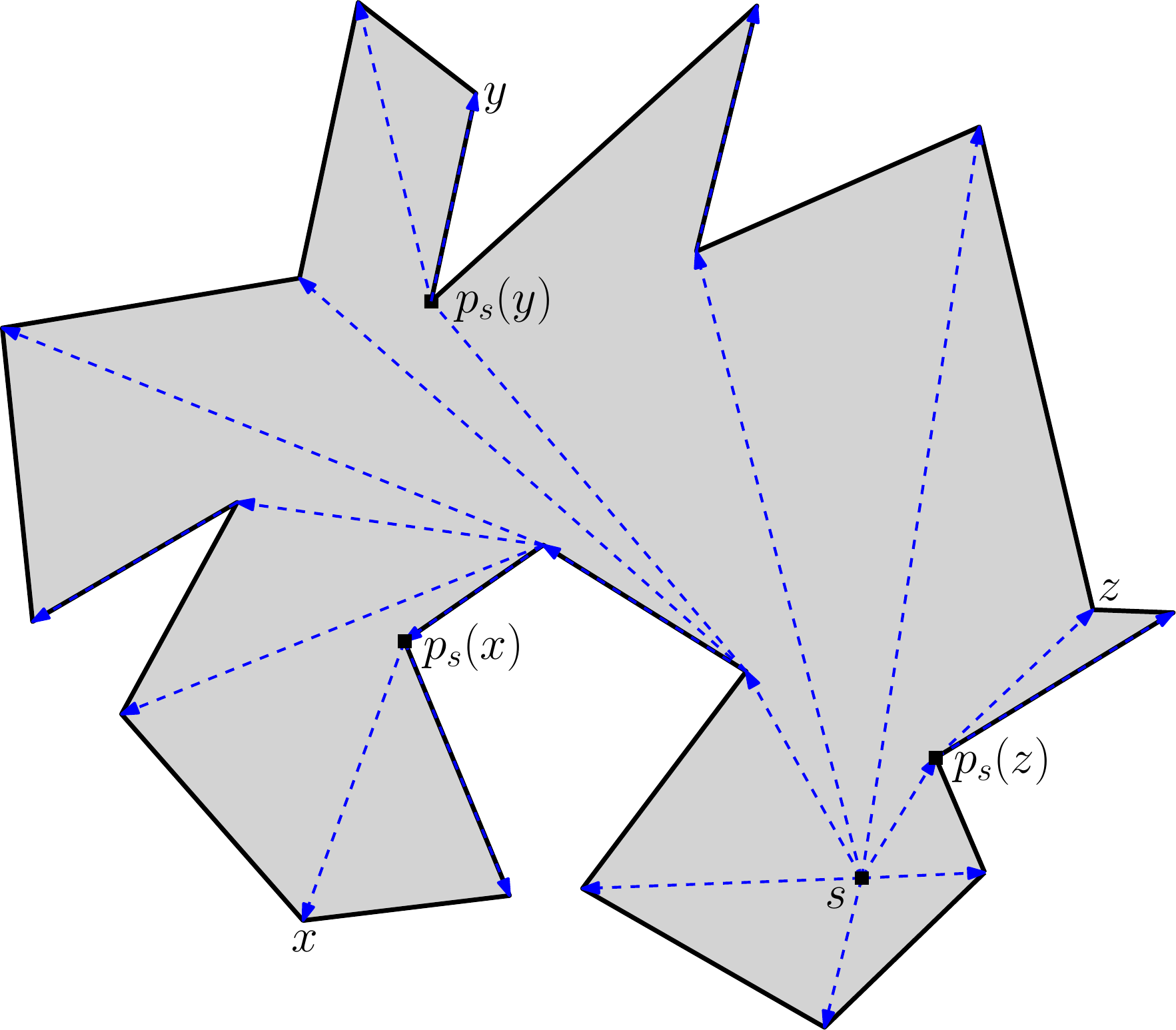}}
  \caption{Euclidean shortest path tree rooted at $s$. 
           The parents of vertices $x$, $y$ and $z$ in $SPT(s)$ are marked as $p_s(x)$, $p_s(y)$ and $p_s(z)$ respectively.}
  \label{spt}
\end{figure}

\vspace*{-.5em}
A \emph{link path} between two points $s$ and $t$ in $P$ is a path inside $P$ that connects $s$ and $t$ by a chain of line segments 
called \emph{links}. A \emph{minimum link path} between $s$ and $t$ is a link path connecting $s$ and $t$ that has the minimum number 
of links. Observe that there may be several different minimum link paths between $s$ and $t$. The \emph{link distance} between any two 
points of $P$ is defined to be the number of links in a minimum link path between them. 

\section{Partitioning a simple polygon into weak visibility polygons}
\label{partitioning_algo}

\vspace*{-0.7em}
The partitioning algorithm partitions $P$ into regions according to their link distance from $v_1$. 
The algorithm starts by computing $\mathcal{VP}(v_1)$, 
which is the set of all points of $P$ whose link distance from $v_1$ is 1. 
Let us denote $\mathcal{VP}(v_1)$ as $V_{1,1}$. 
Then the algorithm computes the weak visibility polygons from every constructed edge of $V_{1,1}$. Let $v_{k(1)}u_{k(1)},
v_{k(2)}u_{k(2)},\dots,v_{k(c)}u_{k(c)}$ denote the constructed edges of $V_{1,1}$ along $bd(P)$ in clockwise order from $v_1$, 
where $c$ is the number of constructed edges in $V_{1,1}$. 
Then the algorithm removes $V_{1,1}$ from $P$. It can be seen that the remaining polygon $P\setminus V_{1,1}$ consists of $c$ disjoint polygons $P(v_{k(1)}u_{k(1)}),P(v_{k(2)}u_{k(2)}),\dots,P(v_{k(c)}u_{k(c)})$. 
For each $j \in \{1,2,\dots,c\}$, the weak visibility polygon $\mathcal{VP}(v_{k(j)}u_{k(j)})$ is computed inside the pocket $P(v_{k(j)}u_{k(j)})$, 
and it is denoted as $V_{2,j}$, i.e. $V_{2,j} = \mathcal{VP}(v_{k(j)}u_{k(j)})\cap P(v_{k(j)}u_{k(j)})$. 
Let $W_1 = \{V_{1,1}\}$ and $W_2 = \bigcup_{j=1}^{c}\{V_{2,j}\}$. Observe that $W_2$ is the set 
of all the disjoint regions of $P$, such that every point of each disjoint region in $W_2$ is at link distance two from $v_1$. \\

\vspace*{-.5em}
Repeating the same process, the algorithm computes $W_3,W_4,\dots,W_d$, where $d$ denotes the maximum link distance of any point of $P$ from $v_1$. 
Note that it is not possible for any visibility polygon belonging to $W_d = \bigcup_{j=1}^{c}V_{d,j}$ to have any constructed edge. 
Therefore, no further visibility polygon is computed. Hence, $P = W_1 \cup W_2 \cup \ldots W_d = V_{1,1} \cup V_{2,1} \cup V_{2,2} 
\cup \ldots \cup V_{d,1} \cup V_{d,2} \cup \ldots$. Thus, the algorithm returns the set $W = \bigcup_{i=1}^{d}W_{i}$, which is a 
partition of $P$. We present the pseudocode for the entire partitioning algorithm below as Algorithm \ref{partition_windows}.

\begin{algorithm}[H]
\caption{An algorithm for partitioning $P$ into visibility polygons}  
\label{partition_windows}
\begin{algorithmic}[1]
\State Compute $\mathcal{VP}(v_1)$ \label{partition_windows:1} 
\State $V_{1,1} \leftarrow \mathcal{VP}(v_1)$, $W_1 \leftarrow \{V_{1,1}\}$ \label{partition_windows:2}
\State $C \leftarrow \bigcup_{s \in W_1}(\mbox{constructed edges of }s)$, $c \leftarrow |C|$ \label{partition_windows:3}
\State $W \leftarrow W_1$, $i \leftarrow 1$ \label{partition_windows:4}
\While{$c > 0$} \label{partition_windows:5}
\State $i \leftarrow i+1$, $W_i \leftarrow \emptyset$ \label{partition_windows:6}
\For{$j=1$ to $c$} \label{partition_windows:7}
\State $V_{i,j} \leftarrow \mathcal{VP}(v_{k(j)}u_{k(j)})\cap P(v_{k(j)}u_{k(j)})$ \label{partition_windows:8}
\State $W_i \leftarrow W_i \cup \{V_{i,j}\}$ \label{partition_windows:9}   
\EndFor \label{partition_windows:10}
\State $W \leftarrow W \cup W_i$ \label{partition_windows:11}
\State $C \leftarrow \bigcup_{s \in W_i}(\mbox{constructed edges of }s)$, $c \leftarrow |C|$ \label{partition_windows:12}
\EndWhile \label{partition_windows:13}
\State \Return $W$ \label{partition_windows:14} 
\end{algorithmic}
\end{algorithm}

\vspace*{-.5em}
Figure \ref{windows} shows the outcome of running Algorithm \ref{partition_windows} on a simple polygon $P$ having 31 vertices, 
where the maximum link distance of any point of $P$ from $v_1$ is 5. 
The algorithm returns the partition 
$W = \{V_{1,1},V_{2,1},V_{2,2},V_{3,1},V_{3,2},V_{3,3,},V_{4,1},V_{4,2},V_{5,1},V_{5,2},V_{5,3}\}$. \\

\vspace*{-.5em}
It can be seen that Algorithm \ref{partition_windows}, as stated above, requires $\mathcal{O}(n^2)$ time, since the visibility polygons are computed
separately. However, the running time can be improved to $\mathcal{O}(n)$ by using the partitioning method given by Suri \cite{Suri_1986,Suri_1987} 
in the context of computing minimum link paths. Using the algorithm of Hershberger \cite{Hershberger_1989} for computing visibility graphs of $P$, 
Suri's algorithm computes weak visibility polygons from selected constructed edges. The same method can be used to compute weak visibility polygons 
from all constructed edges of visibility polygons in $W$ 
in $\mathcal{O}(n)$ time. 
The \emph{visibility graph} of $P$ is a graph which has a node corresponding to every vertex of $P$ and there is an edge between a pair of nodes if 
and only if the corresponding pair of vertices are visible from each other in $P$. We summarize the result in the following theorem.

\begin{theorem} \label{partition_time}
A simple polygon $P$ can be partitioned into visibility polygons according to their link distance from any vertex in $\mathcal{O}(n)$ time. 
\end{theorem}

\section{Traversing the hierarchy of visibility polygons}
\label{traverse_hierarchy}

Our algorithm for placement of vertex guards uses the hierarchy of visibility polygons $W$, as computed in Section \ref{partitioning_algo}. 
Let $S_{d},S_{d-1},\ldots,S_2,S_1$ be the set of vertex guards chosen for guarding vertices of visibility polygons in $W_{d},W_{d-1},
\ldots,W_2,W_1$ respectively. Since $W_1 = \{V_{1,1}\}$ and $V_{1,1} = \mathcal{VP}(v_1)$, we have $S_1 = \{v_1\}$. So the algorithm 
essentially has to decide guards in $S_{d},S_{d-1},\ldots,S_2$. We have the following observation.

\vspace*{-.33em}
\begin{lemma} \label{S_i}
For every $2 \leq i < d$, every vertex guard in $S_i$ belongs to some visibility polygon in $W_{i+1} \cup W_i \cup W_{i-1}$, 
and every vertex guard in $S_d$ belongs to some visibility polygon in $W_d \cup W_{d-1}$. 
\end{lemma}
\begin{proof} 
 Consider any vertex guard $g \in S_i$, where $2 \leq i < d$. Now $g$ can guard only vertices in $\mathcal{VP}(g)$, and every vertex in $\mathcal{VP}(g)$ 
 must be at a link distance of 1 from $g$. Let $U$ denote the set of vertices in $\mathcal{VP}(g)$ that also belong to any $V_{i,j} \in  W_i$.
 The inclusion of $g$ in $S_i$ guarantees that $U$ is not empty, since there exists at least one vertex $y \in U$ that is guarded by $g$.
 Now, if we consider any such $y \in U$, then the link distance between $g$ and $y$ must be 1, and also the link distance of $y$ from $v_1$ must be $i$. 
 Therefore, the link distance of $g$ from $v_1$ can only be $i-1$, $i$, or $i+1$, and 
 hence $g$ must belong to some visibility polygon in $W_{i+1} \cup W_i \cup W_{i-1}$. 
 Using the same argument, for any vertex guard $g \in S_d$, $g$ must belong to some visibility polygon in $W_d \cup W_{d-1}$ 
 (rather than $W_{d+1} \cup W_d \cup W_{d-1}$), since the level $W_{d+1}$ does not exist in the hierarchy $W$.
\end{proof}

\vspace*{-.33em}
As can be seen from the proof of Lemma \ref{S_i}, the placement of guards is locally restricted to visibility polygons belonging to adjacent 
levels in the partition hierarchy $W$. We formalize this intuition by introducing the notion of the \emph{partition tree} of $P$, which is a 
{\it dual graph} denoted by $T$. 
Each visibility polygon $V_{i,j} \in W$ is represented as a vertex of $T$ (also denoted by $V_{i,j}$), and two vertices of $T$ are connected 
by an edge in $T$ if and only if the corresponding visibility polygons share a constructed edge. 
Treating $V_{1,1}$ as the root of $T$, the standard parent-child-sibling relationships can be imposed between the visibility polygons in $W$. \\

\vspace*{-.5em}
Our algorithm starts off by guarding all vertices belonging to the visibility polygons in $W_d = \{V_{d,1},V_{d,2},\dots\}$, 
which are effectively the nodes of $T$ furthest from the root $V_{1,1}$. 
The algorithm scans $V_{d,1}$,$V_{d,2}$,\dots separately for identifying the respective guards in $S_d$. 
We know from Lemma \ref{S_i} that every vertex guard in $S_d$ 
belongs to some visibility polygon in $W_d \cup W_{d-1}$.
Consider a particular $V_{d,k} \in W_{d}$, and let $V_{d-1,j} \in W_{d-1}$ be the parent of $V_{d,k}$ in $T$.
Consider the constructed edge $v_ku_k$ between  $V_{d,k}$ and $V_{d-1,j}$.
For guarding the vertices of $V_{d,k} = \mathcal{VP}(v_ku_k) \setminus V_{d-1,j}$, it is enough to focus on 
the subpolygon $Q$ consisting of $V_{d,k}$ itself and the portion of $V_{d-1,j}$ that is weakly visible from $v_ku_k$.
So, the subproblem of guarding $V_{d,k}$ (or any other visibility polygon belonging to $W_d$) essentially reduces to placing 
vertex guards in a polygon containing a weak visibility chord $vu$ (corresponding to $v_ku_k$ in the original subproblem) 
in order to guard only the vertices lying on one side of $uv$; however, vertex guards can be chosen freely from either side of the chord $uv$. We discuss the placement of guards in this reduced problem in Section \ref{vertex_algo}. \\

\vspace*{-.5em}
Instead of guarding each weak visibility polygon $Q$ separately, common vertex guards can be placed by traversing the boundary of overlapping weak visibility polygons. 
Let us explain by considering any $V_{d-1,j} \in W_{d-1}$. 
Let us denote the constructed edges that are shared between $V_{d-1,j}$ and the $m$ children of $V_{d-1,j}$ as 
$v_{j(1)}u_{j(1)}, v_{j(2)}u_{j(2)}, \dots, v_{j(m)}u_{j(m)}$ respectively.
Using all these constructed edges, 
let us construct the weak visibility polygons $\mathcal{VP}(v_{j(1)}u_{j(1)})$, $\mathcal{VP}(v_{j(2)}u_{j(2)})$, $\dots$, $\mathcal{VP}(v_{j(m)}u_{j(m)})$.
Observe that each such weak visibility polygon is divided into two portions by the corresponding constructed edge; 
one of the portions forms a child of $V_{d-1,j}$ belonging to $W_{d}$, whereas the other portion is a subregion of $V_{d-1,j}$ itself. 
Moreover, for several of the weak visibility polygons among $\mathcal{VP}(v_{j(1)}u_{j(1)}), \mathcal{VP}(v_{j(2)}u_{j(2)}), \dots, \mathcal{VP}(v_{j(m)}u_{j(m)})$, the second portions may have overlapping subregions in $V_{d-1,j}$. 
Thus, there may exist vertex guards in these overlapping subregions that can see portions of several of the children of $V_{d-1,j}$.
Therefore, for guarding vertices of polygons from $W_{d}$, 
let us extend the definition of $Q$ to be the union of all the overlapping weak visibility polygons defined by the constructed edges corresponding to the children of each $V_{d-1,j}$.
For instance, consider the constructed edges $v_{17}u_{17}$, $v_{21}u_{21}$ and $v_{23}u_{23}$ on the boundary of $V_{4,1}$ in Figure \ref{windows}; for guarding the corresponding children $V_{5,1}$, $V_{5,2}$ and $V_{5,3}$ respectively, we define $Q$ as $\mathcal{VP}(v_{17}u_{17}) \cup \mathcal{VP}(v_{21}u_{21}) \cup \mathcal{VP}(v_{23}u_{23})$ and traverse $Q$. \\ 

\vspace*{-0.5em}
After having successively computed $S_d$ for guarding vertices belonging to visibility polygons in $W_d = \{V_{d,1},V_{d,2},\dots\}$, 
the algorithm next computes $S_{d-1}$ for guarding vertices belonging to visibility polygons in $W_{d-1} = \{V_{d-1,1},V_{d-1,2},\dots\}$. 
Since all vertices belonging to visibility polygons in $W_d$ are already marked by guards chosen belonging to $S_d$, 
all remaining unmarked vertices of $P$ can have link distance at most $d-1$ from $v_1$.
So, any weak visibility polygon $V_{d-1,k} \in W_{d-1}$ can now be treated as a weak visibility polygon that is the farthest link distance from $v_1$.
Therefore, the guards of $S_{d-1}$ are chosen in a similar way as those of $S_d$.   
It can be seen that this same method can be used for computing $S_i$ for every $i<d$. 
Thus, in successive phases, our algorithm computes the guard sets $S_d,S_{d-1},S_{d-2}, \ldots,S_2$ for guarding vertices belonging to visibility polygons in 
$W_d,W_{d-1},W_{d-2},\ldots,W_2$ respectively, until it finally terminates after placing a single guard at $v_1$ for guarding vertices of 
$V_{1,1} \in W_1$. The final guard set $S = S_d \cup S_{d-1} \cup S_{d-2} \cup \dots \cup S_2 \cup S_1$ returned by the algorithm guards all vertices of $P$. 
The pseudocode for the entire algorithmic framework is provided below.

\begin{algorithm}[H]
\caption{Algorithm for computing a guard set $S$ from the partition tree $T$ rooted at $v_1$} 
\label{overall_pcode}        
\begin{algorithmic}[1]
\State Initialize all vertices of $P$ as unmarked \label{overall_pcode:1}
\State $d \leftarrow$ number of levels in the partition tree $T$ \label{overall_pcode:2}  
\For { each $i \in \{d-1,\dots,3,2,1\}$ } \Comment{Traverse starting from the 2nd deepest level of $T$}  \label{overall_pcode:3}
\State $S_{i+1} \leftarrow \emptyset$ \label{overall_pcode:4} 
\State $c_{i} \leftarrow |W_i|$ \Comment{$c_i$ denotes the number of nodes at the $i$th level of $T$}  \label{overall_pcode:5}
\For { each $j \in \{1,2,\dots,c_i\}$ } \label{overall_pcode:6}
\State Place new guards in $S_{i+1}$ for guarding every unmarked vertex of all children of $V_{i,j}$ \label{overall_pcode:7} 
\State Mark all vertices of $P$ that are visible from the new guards added to $S_{i+1}$ \label{overall_pcode:8}
\EndFor \label{overall_pcode:9}
\EndFor \label{overall_pcode:10}
\State $S_1 \leftarrow \{v_1\}$ \label{overall_pcode:11}
\State \Return $S = S_d \cup S_{d-1} \cup S_{d-2} \cup \dots \cup S_2 \cup S_1$ \label{overall_pcode:12} 
\end{algorithmic}
\end{algorithm} 

\vspace*{-.5em}
The procedure for placing new guards in $S_{i+1}$ for guarding all children of a particular $V_{i,j}$, as mentioned in line \ref{overall_pcode:7} of Algorithm \ref{overall_pcode}, 
is presented in detail in Section \ref{vertex_algo}. 

\section{Placement of Vertex Guards in a Weak Visibility Polygon}
\label{vertex_algo}

\vspace*{-.5em}
Let $Q$ be a simple polygon that is weakly visible from an internal chord $uv$, i.e. we have $\mathcal{VP}(uv) = Q$. 
Observe that the chord $uv$ splits $Q$ into two sub-polygons $Q_U$ and $Q_L$ as follows.
The sub-polygon bounded by $bd_c(u,v)$ and $uv$, is denoted as $Q_U$, 
and the sub-polygon bounded by $bd_{cc}(u,v)$ and $uv$, is denoted as $Q_L$.
As a first step, our algorithm (see Algorithm \ref{vg_pcode_gen}) places a set of vertex guards, 
denoted by $S$, for guarding \emph{only} the vertices belonging to $Q_U$,
though $S$ is allowed to contain guards from both $Q_U$ and $Q_L$. \\

\vspace*{-.5em}
Let $G_{opt}$ be a set of optimal vertex guards for guarding all points of $Q_U$, including interior points.
Let $G_{opt}^U$ and $G_{opt}^L$ be the subsets of guards in $G_{opt}$ that belong to $Q_U$ (i.e. lie on $bd_c(u,v)$) and $Q_L$ (i.e. lie on $bd_{cc}(u,v)$) respectively.
Since $G_{opt}^U$ and $G_{opt}^L$ form a partition of $G_{opt}$, 
$|G_{opt}^U| + |G_{opt}^L| = |G_{opt}|$.

\subsection{Concept of Inside and Outside Guards}
\label{inside_and_outside}

\vspace*{-.33em}
Suppose we wish to guard an arbitrary vertex $z$ of $Q_U$. 
Then, a guard must be placed at a vertex of $Q$ belonging to $\mathcal{VP}(z)$. 
Henceforth, let $\mathcal{VVP}(z)$ denote the set of all polygonal vertices of $\mathcal{VP}(z)$. 
Further, let us define the \emph{inward visible vertices} and the \emph{outward visible vertices} of $z$, denoted by $\mathcal{VVP}^{+}(z)$ and $\mathcal{VVP}^{-}(z)$ respectively, as follows. \\

\vspace*{-.7em}
\centerline{ $\mathcal{VVP}^{+}(z) = \{ x \in \mathcal{VVP}(z) : \mbox{the segment $zx$ does not intersect $uv$}\} $ } 
\centerline{ $\mathcal{VVP}^{-}(z) = \{ x \in \mathcal{VVP}(z) : \mbox{the segment $zx$ intersects $uv$}\} $ } 
We shall henceforth refer to the vertex guards belonging to $\mathcal{VVP}^{+}(z)$ and $\mathcal{VVP}^{-}(z)$ as \emph{inside guards}
and \emph{outside guards} for $z$ respectively. \\

\vspace*{-.55em}
Consider the weakly visible polygon in Figure \ref{need_inside}, where the three vertices $z_1$, $z_2$ and $z_3$ of $Q_U$ are such that their respective sets of \emph{outward} visible vertices are pairwise disjoint, i.e. 
$\mathcal{VVP}^{-}(z_1)\cap \mathcal{VVP}^{-}(z_2) = \emptyset$, $\mathcal{VVP}^{-}(z_2)\cap \mathcal{VVP}^{-}(z_3) = \emptyset$, and $\mathcal{VVP}^{-}(z_1)\cap \mathcal{VVP}^{-}(z_3) = \emptyset$. 
If an algorithm chooses only outside guards, then three separate guards are required 
for guarding $z_1$, $z_2$ and $z_3$. 
However, an optimal solution may place a single guard on any one of the five vertices of $Q_U$ 
for guarding $z_1$, $z_2$ and $z_3$. 

\vspace*{-0.55em}
\begin{figure}[H]
\begin{minipage}{.48\textwidth}
\centerline{\includegraphics[width=1.03\textwidth]{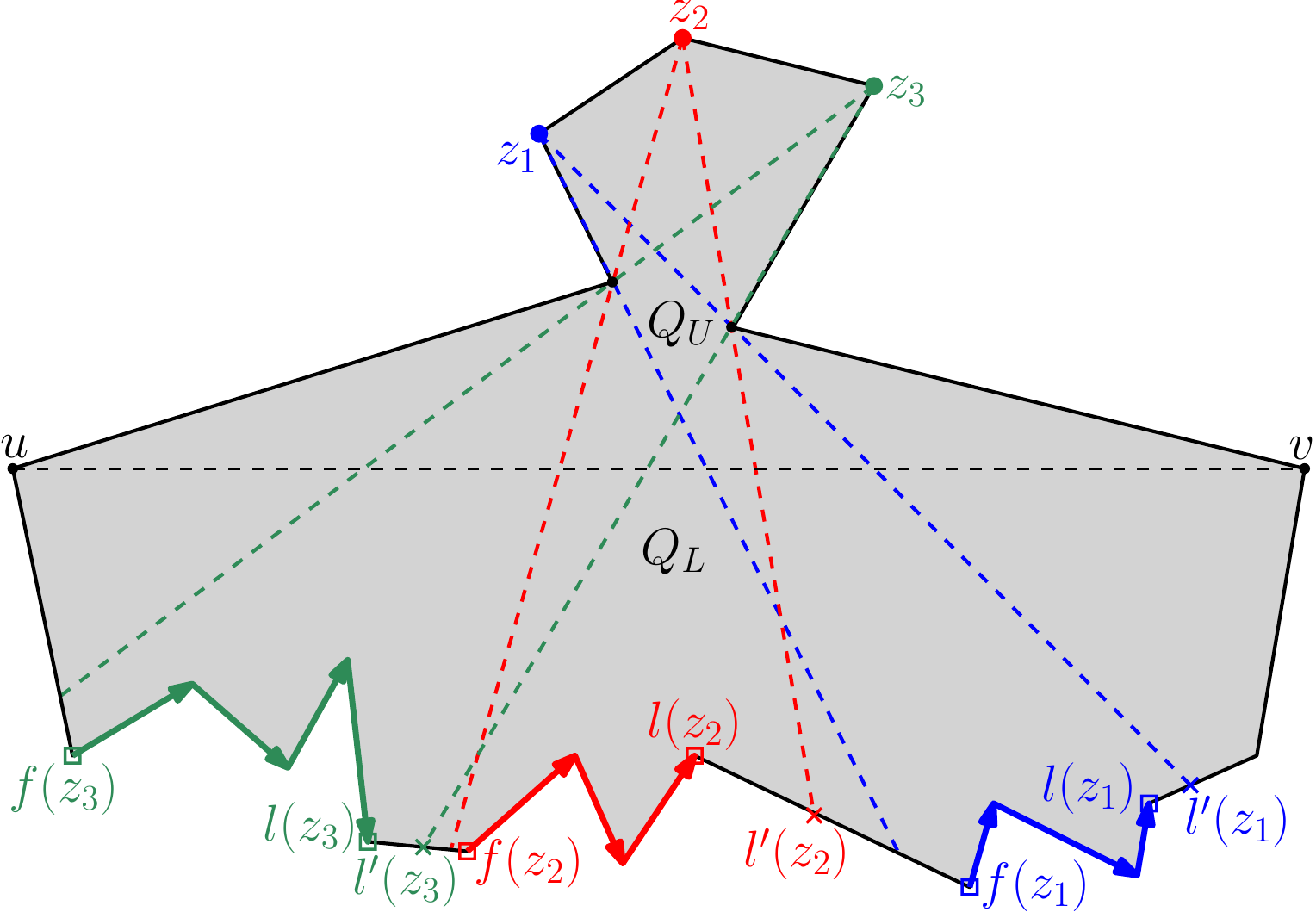}}
\caption{Example showing the need for placement of inside guards. Note that $\mathcal{VVP}^{-}(z_1)\cap \mathcal{VVP}^{-}(z_2) = \emptyset$, $\mathcal{VVP}^{-}(z_2)\cap \mathcal{VVP}^{-}(z_3) = \emptyset$, and $\mathcal{VVP}^{-}(z_1)\cap \mathcal{VVP}^{-}(z_3) = \emptyset$.}
\label{need_inside}
\end{minipage}
\hspace*{.01\textwidth}
\begin{minipage}{.51\textwidth}
\centerline{\includegraphics[width=1.03\textwidth]{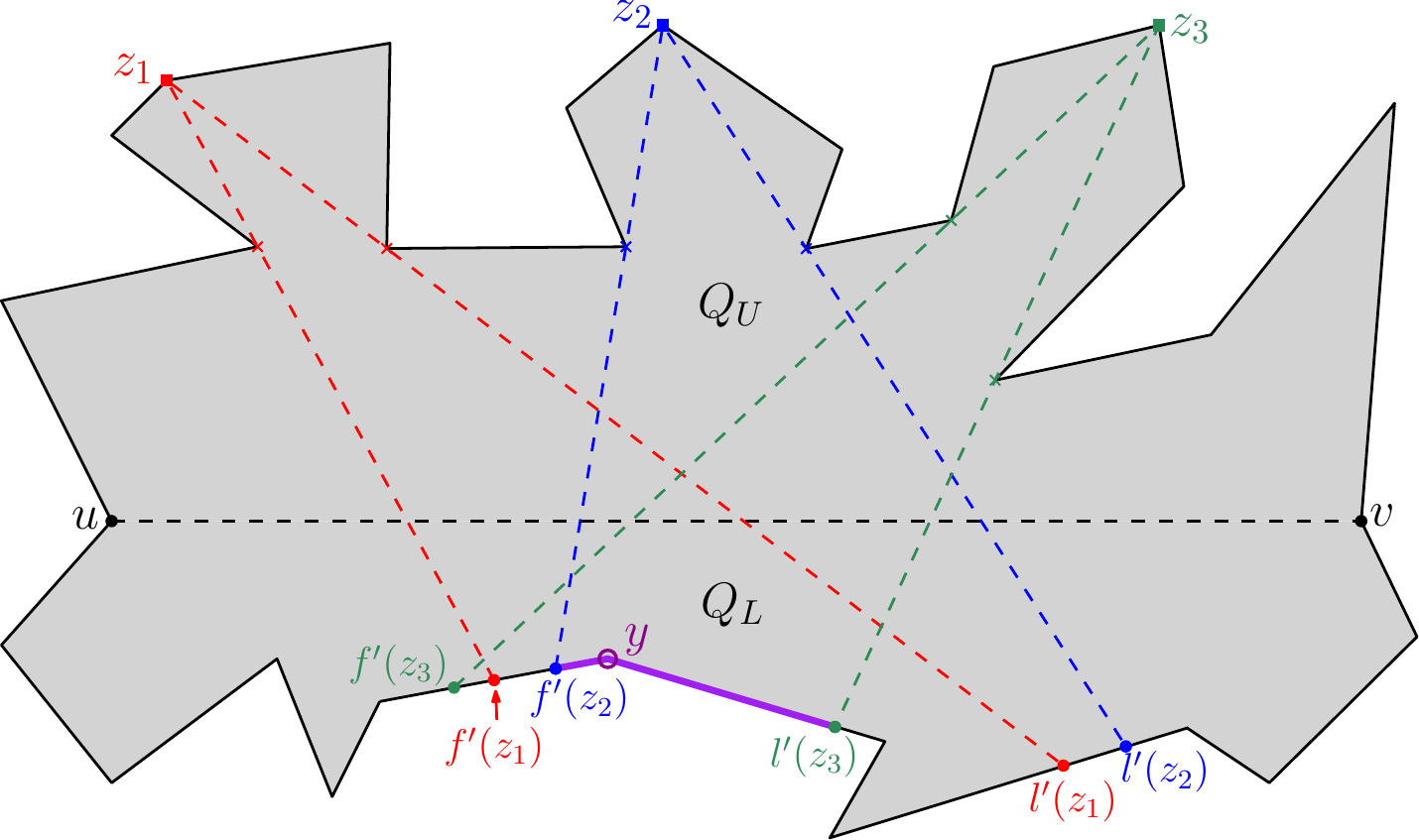}}
\caption{Example showing the need for placement of outside guards. Note that $\mathcal{VVP}^{+}(z_1)\cap \mathcal{VVP}^{+}(z_2) = \emptyset$, $\mathcal{VVP}^{+}(z_2)\cap \mathcal{VVP}^{+}(z_3) = \emptyset$, and $\mathcal{VVP}^{+}(z_1)\cap \mathcal{VVP}^{+}(z_3) = \emptyset$.}
\label{need_outside}
\end{minipage}
\end{figure}

\vspace*{-.55em}
On the other hand, in Figure \ref{need_outside}, the three vertices $z_1$, $z_2$ and $z_3$ of $Q_U$ are such that their respective sets of \emph{inward} visible vertices are pairwise disjoint, i.e. 
$\mathcal{VVP}^{+}(z_1)\cap \mathcal{VVP}^{+}(z_2) = \emptyset$, $\mathcal{VVP}^{+}(z_2)\cap \mathcal{VVP}^{+}(z_3) = \emptyset$, and $\mathcal{VVP}^{+}(z_1)\cap \mathcal{VVP}^{+}(z_3) = \emptyset$. 
If an algorithm chooses only inside guards, then three separate guards are required for guarding $z_1$, $z_2$ and $z_3$. 
However, the outward visible vertices of $z_1$, $z_2$ and $z_3$ overlap; and moreover there exists a vertex $Y$ of $Q_L$ such that $G \in \mathcal{VVP}^{-}(z_1)\cap \mathcal{VVP}^{-}(z_2)\cap \mathcal{VVP}^{-}(z_3)$. Thus, an optimal solution may choose $y$ as an outside guard for guarding $z_1$, $z_2$ and $z_3$ together. \\

\vspace*{-.5em}
The above discussion suggests that it is better to choose guards from both $\mathcal{VVP}^{+}(z)$ and $\mathcal{VVP}^{-}(z)$ for guarding the same vertex $z$ of $Q_U$, 
in order to prevent computing a guard $S$ that is arbitrarily large compared to $G_{opt}$.
Therefore, our algorithm selects a subset $Z$ of vertices of $Q_U$, 
and places a fixed number of both inside and outside guards corresponding to each of them,
so that these guards together see the entire $Q_U$. 
Moreover, the vertices in $Z$ (and also the guards corresponding to them) are selected in such a way that 
they correspond to the guards in $G_{opt}$, though the correspondence is not necessarily one-to-one, 
and this enables us to provide an approximation bound on our chosen set of guards.
We henceforth refer to this subset $Z$ of selected vertices as \emph{primary vertices}. 
Further, let $Z^U \subseteq Z$ be such that each primary vertex in $Z^U$ is visible from at least one guard in $G_{opt}^U$. 
Similarly, let  $Z^L \subseteq Z$ be such that each primary vertex in $Z^L$ is visible from at least one guard in $G_{opt}^L$. 
Since any $z \in Z$ must be visible from at least one guard of $G_{opt}^U$ or $G_{opt}^L$, we have $Z = Z^U \cup Z^L$, 
and so we have:
\vspace*{-.5em}
\begin{equation} \label{Z_split}
|Z| \leq |Z^U| + |Z^L|
\end{equation}

\vspace*{-.5em}
The general strategy for placement of guards by our algorithm for guarding only the vertices of $Q_U$ aims to establish a constant-factor approximation bound on $S$ by separately proving the following three bounds.  
\begin{equation} \label{num_guards}
    |S| \leq c \cdot |Z| 
\end{equation}
\begin{equation} \label{lower_ineq}
    |Z^L| \leq c_1 \cdot |G_{opt}^L| 
\end{equation}
\begin{equation} \label{upper_ineq}
    |Z^U| \leq c_2 \cdot |G_{opt}^U| 
\end{equation}
The above inequalities \eqref{Z_split}, \eqref{num_guards}, \eqref{lower_ineq} and \eqref{upper_ineq}  together imply the following conclusion.
\begin{equation} \label{ineq_final1} 
|S|\leq c.|Z| \leq c.|Z^L| + c.|Z^U| \leq c.c_1 \cdot |G_{opt}^L| + c.c_2 \cdot |G_{opt}^U| \leq c.\max(c_1,c_2)\cdot|G_{opt}| 
\end{equation}

\subsection{Selection of Primary Vertices}

Observe that, for any vertex $z_k \in Z$, both $\mathcal{VVP}^{+}(z_k)$ and $\mathcal{VVP}^{-}(z_k)$ may be considered to be ordered sets by taking into account the natural ordering of the visible vertices of $Q$ in clockwise order along $bd_{c}(u,v)$ and in counter-clockwise order along $bd_{cc}(u,v)$ respectively. 
Let us denote the \emph{first visible vertex} and the \emph{last visible vertex} belonging to the ordered set $\mathcal{VVP}^{-}(z_k)$ 
by $f(z_k)$ and $l(z_k)$ respectively (see Figure \ref{need_inside}).
Also, we denote by $l'(z_k)$ (similarly,  $f'(z_k)$) the \emph{last visible point} (similarly, \emph{first visible point}) 
of $\mathcal{VP}(z_k) \cap bd_{cc}(u,v)$, 
which is obtained by extending the ray $\overrightarrow{z_k p(v,z_k)}$ (respectively, $\overrightarrow{z_k p(u,z_k)}$) till it touches $bd_{cc}(u,v)$, where $p(v,z_k)$ (similarly, $p(u,z_k)$) is the parent of $z_k$ in $SPT(v)$ (respectively, $SPT(u)$). 
Observe that $bd_{cc}(f'(z_k),l'(z_k)) \cap \mathcal{VP}(z_k)$ is the only 
portion of the boundary $bd_{cc}(u,v)$ that has vertices visible from $z_k$.
Note that all vertices of $bd_{cc}(f'(z_k),l'(z_k))$ may not be visible from $z_k$, 
since some of them may lie inside left or right pockets of $\mathcal{VP}(z_k)$. 
Similarly, observe that $bd_{c}(p(u,z_k),p(v,z_k))$ is the only portion of the 
boundary $bd_{c}(u,v)$ that has vertices visible from $z_k$.
Note that all vertices of $bd_{cc}(p(u,z_k),p(v,z_k))$ may not be visible from $z_k$, since some of them may lie inside left or right pockets of $\mathcal{VP}(z_k)$. \\

\vspace*{-.5em}
Let us discuss how primary vertices are selected by our algorithm.
Initially, since all vertices of $Q_U$ are unguarded, they are considered to be \emph{unmarked}. 
As vertex guards are placed over successive iterations, vertices of $Q_U$ are \emph{marked} as soon as they become visible from some guard placed so far.
In the $k$-th iteration, 
the next primary vertex $z_k \in Z$ chosen by our algorithm is an unmarked vertex such that $l'(z_k)$ precedes $l'(x)$ (henceforth denoted notationally as $l'(z_k) \prec l'(x)$) for any other unmarked vertex $x$ of $Q_U$. An immediate consequence of such choice of the primary vertex $z_k$ is that all vertices on $bd_c(z_k,p(v,z_k))$ must already be marked. \\

\vspace*{-.5em}
\begin{lemma} \label{no_unmarked_right}
If $z_k$ is chosen as the next primary vertex by Algorithm \ref{vg_pcode_sp1}, then no unmarked vertices exist on $bd_c(z_k,p(v,z_k))$.
\end{lemma}
\begin{proof}
 Let us assume, on the contrary, that there exists a vertex $q$ on $bd_c(z_k,p(v,z_k))$ that is yet to be marked. 
 Observe that the ray $\overrightarrow{q p(v,q)}$ must intersect the ray $\overrightarrow{z_k p(v,z_k)}$ at a point between $z_k$ and $p(v,z_k)$, 
 which implies that $l'(q) \prec l'(z_k)$ on $bd_{cc}(u,v)$ (see Figure \ref{proof_case1a}).
 So, in the current iteration, $q$ rather than $z_k$ is chosen as the next primary vertex, which is a contradiction.
\end{proof}

In Section \ref{only_lower}, we consider guarding vertices of $Q_U$ in a special scenario where $G_{opt}$ uses vertex guards only from $Q_L$.
In Section \ref{upper_n_lower}, we consider guarding vertices of $Q_U$ in the general scenario where the guards of $G_{opt}$ are not restricted to $Q_L$.
In Section \ref{interior}, we enhance the procedure for guarding vertices to ensure that all interior points of $Q_U$ are guarded as well.

\subsection{Placement of guards in a special scenario}
\label{only_lower}

Let us consider a special scenario where $G_{opt}^L = G_{opt}$ and $G_{opt}^U = \emptyset$, i.e. 
$G_{opt}$ uses only vertex guards from $Q_L$. 
In such a scenario, for every vertex $z_k \in Z$, $\mathcal{VVP}^{-}(z_k)$ must contain a guard from $G_{opt}$.
So, a natural idea is to place outside guards in a greedy manner so that they lie in the common intersection of outward visible vertices of as many vertices of $Q_U$ as possible. 
For any primary vertex $z_k$, let us denote by $\mathcal{OVV}^{-}(z_k)$ the set of unmarked vertices of $Q_U$ whose outward visible vertices overlap with those of $z_k$. 
In other words, 
\vspace{-0.5em}
$$ \mathcal{OVV}^{-}(z_k) = \{ x \in \mathcal{V}(Q_U) : \mbox{ $x$ is unmarked, and } \mathcal{VVP}^{-}(z_k) \cap \mathcal{VVP}^{-}(x) \neq \emptyset \} $$
So, each vertex of $Q_U$ belonging to $\mathcal{OVV}^{-}(z_k)$ is visible from at least one vertex of $\mathcal{VVP}^{-}(z_k)$.
Further, $\mathcal{OVV}^{-}(z_k)$ can be considered to be an ordered set, where for any pair of elements $x_1,x_2 \in \mathcal{OVV}^{-}(z_k)$, 
we define $x_1 \prec x_2$ if and only if $l'(x_1)$ precedes $l'(x_2)$ in counter-clockwise order on $bd_{cc}(u,v)$. 
For the current primary vertex $z_k$, let us assume without loss of generality that $\mathcal{OVV}^{-}(z_k) = \{ x^k_1, x^k_2, x^k_3, \dots\, x^k_{m(k)}\}$ such that $l'(x^k_1) \prec l'(x^k_2) \prec \dots \prec l'(x^k_{m(k)})$ in counter-clockwise order on $bd_{cc}(u,v)$. \\

\vspace{-0.5em}
Let us partition the vertices belonging to $\mathcal{OVV}^{-}(z_k)$ into 3 sets, viz. $A^k$, $B^k$ and $C^k$ in the following manner. 
Every vertex of $\mathcal{OVV}^{-}(z_k)$ visible from $l(z_k)$ is included in $B^k$.
Observe that, by definition $z^k \in B^k$.
Obviously, for each vertex $x^k_i \in \mathcal{OVV}^{-}(z_k) \setminus B^k$, 
$x^k_i$ is not visible from $l(z_k)$ due to the presence of some constructed edge.
The vertices of $\mathcal{OVV}^{-}(z_k) \setminus B^k$ are categorized into $A^k$ and $C^k$ based on whether this constructed edge creates a right pocket or a left pocket.
Suppose $x^k_i \in \mathcal{OVV}^{-}(z_k) \setminus B^k$ is a vertex such that $\mathcal{VP}(x^k_{i})$ creates a constructed edge $t(x^k_i) t'(x^k_i)$, 
where $t(x^k_i) \in \mathcal{V}(Q_L)$ is a polygonal vertex and $t'(x^k_i)$ is the point where $\overrightarrow{x^k_i t(x^k_i)}$ first intersects $bd_{cc}(u,v)$. 
If $t(x^k_i)$ lies on $bd_{cc}(f'(z_k),l'(z_k))$ and $t'(x^k_i)$ lies on $bd_{cc}(l(z_k),v)$,
i.e. if $f(z_k) \prec t(x^k_i) \prec l'(z_k)$ and $l(z_k) \prec t'(x^k_i) \prec v$,
then $x^k_i$ is included in $A^k$. 
For instance, in Figure \ref{fig_case1}, $x_1^2 \in A^2$ due to the constructed edge $t(x_1^2)t'(x_1^2)$.
On the other hand, if $t(x^k_i)$ lies on $bd_{cc}(l'(z_k),v)$ and $t'(x^k_i)$ lies on $bd_{cc}(f(z_k),l'(z_k))$,
i.e. if $l'(z_k) \prec t(x^k_i) \prec v$ and $f(z_k) \prec t'(x^k_i) \prec l'(z_k)$,
then $x^k_i$ is included in $C^k$.
For instance, in Figure \ref{fig_case1}, $x_3^2 \in C^2$ due to the constructed edge $t(x_3^2)t'(x_3^2)$.
Observe that, all vertices of $A^k$ must lie on $bd_c(u,z_k)$, 
whereas all vertices of $C^k$ must lie on $bd_c(z_k,v)$.

\begin{lemma} \label{B_property}
The vertex $l(z_k)$ sees all vertices belonging to $B^k$.
\end{lemma}

\begin{proof}
We know that, due to the choice of the primary vertex $z_k$, $f'(x) \prec l(z_k) \prec l'(x)$ for every vertex $x \in \mathcal{OVV}^{-}(z_k)$.
Therefore, if $x$ is not visible from $l(z_k)$, then there must
exist a constructed edge $t(x)t'(x)$ of $\mathcal{VVP}^{-}(z_k)$
such that (i) either $t(x) \in \mathcal{VVP}^{-}(z_k)$, 
in which case $x$ must belong to $A^k$, 
or (ii) $t(x)$ lies on $bd_{cc}(l'(z_k),v)$, 
in which case $x$ must belong to $C^k$.
So, if $x \in B^k$, $x$ must be visible from $l(z_k)$.
In other words, $l(z_k)$ sees all vertices belonging to $B^k$.
\end{proof}

\begin{corollary} \label{SP_turns}
The shortest paths from $l(z_k)$ to any vertex of $A^k$ makes only left turns, 
whereas the shortest paths from $l(z_k)$ to any vertex of $C^k$ makes only right turns.
\end{corollary}

\begin{corollary} \label{convex_chain}
If $bd_{cc}(f(z_k),l(z_k))$ is convex, then $A^k = C^k = \emptyset$ and $B^k = \mathcal{OVV}^{-}(z_k)$, which implies that all vertices of $\mathcal{OVV}^{-}(z_k)$ are visible from $l(z_k)$.
\end{corollary}

For any $X \subseteq \mathcal{OVV}^{-}(z_k)$, we define $\mathcal{CI}(X) = \bigcap_{x \in X} \mathcal{VVP}^{-}(x)$, 
which implies that, for each vertex $y \in \mathcal{CI}(X)$, 
all the vertices belonging to $X$ are visible from $y$,
and hence can be guarded by placing a vertex guard at $y$.
By definition, for any $X \subseteq \mathcal{OVV}^{-}(z_k)$, the vertices belonging to $\mathcal{CI}(X)$ lie on $bd_{cc}(u,v)$. 
\begin{lemma}
For every $k \in \{1,2,\dots,|Z|\}$, 
$\mathcal{CI}(B^k) \neq \emptyset$.
\end{lemma}
\begin{proof}
It follows directly from Lemma \ref{B_property} that
$l(z_k) \in \mathcal{VVP}^{-}(x)$ for every $x \in B^k$. 
Thus, $l(z_k) \in \bigcap_{x \in B^k} \mathcal{VVP}^{-}(x)$, i.e. $l(z_k) \in \mathcal{CI}(B^k)$.
\end{proof}
Depending on the vertices in $A^k$, $B^k$ and $C^k$, 
we have the following exhaustive cases because 
$\mathcal{CI}(B^k) \neq \emptyset$ by Lemma \ref{B_property}.
\begin{description}

\item[Case 1 -] $\mathcal{CI}(A^k \cup B^k \cup C^k) \neq \emptyset$ (see Figure \ref{fig_case1})

\item[Case 2 -] $\mathcal{CI}(A^k \cup B^k \cup C^k) = \emptyset$ and $\mathcal{CI}(B^k) \neq \emptyset$ 

\item[\hspace{11mm} Case 2a -] $\mathcal{CI}(A^k) \neq \emptyset$ and $\mathcal{CI}(C^k) \neq \emptyset$ (see Figure \ref{fig_case2a})

\item[\hspace{11mm} Case 2b -] $\mathcal{CI}(A^k) = \emptyset$ and $\mathcal{CI}(C^k) \neq \emptyset$ (see Figure \ref{fig_case2b})

\item[\hspace{11mm} Case 2c -] $\mathcal{CI}(A^k) \neq \emptyset$ and $\mathcal{CI}(C^k) = \emptyset$ (see Figure \ref{fig_case2c})

\item[\hspace{11mm} Case 2d -] $\mathcal{CI}(A^k) = \emptyset$ and $\mathcal{CI}(C^k) = \emptyset$ (see Figure \ref{fig_case2d})






\end{description}

\begin{figure}[H]
\begin{minipage}{0.43\textwidth}
  \centerline{\includegraphics[width=\textwidth]{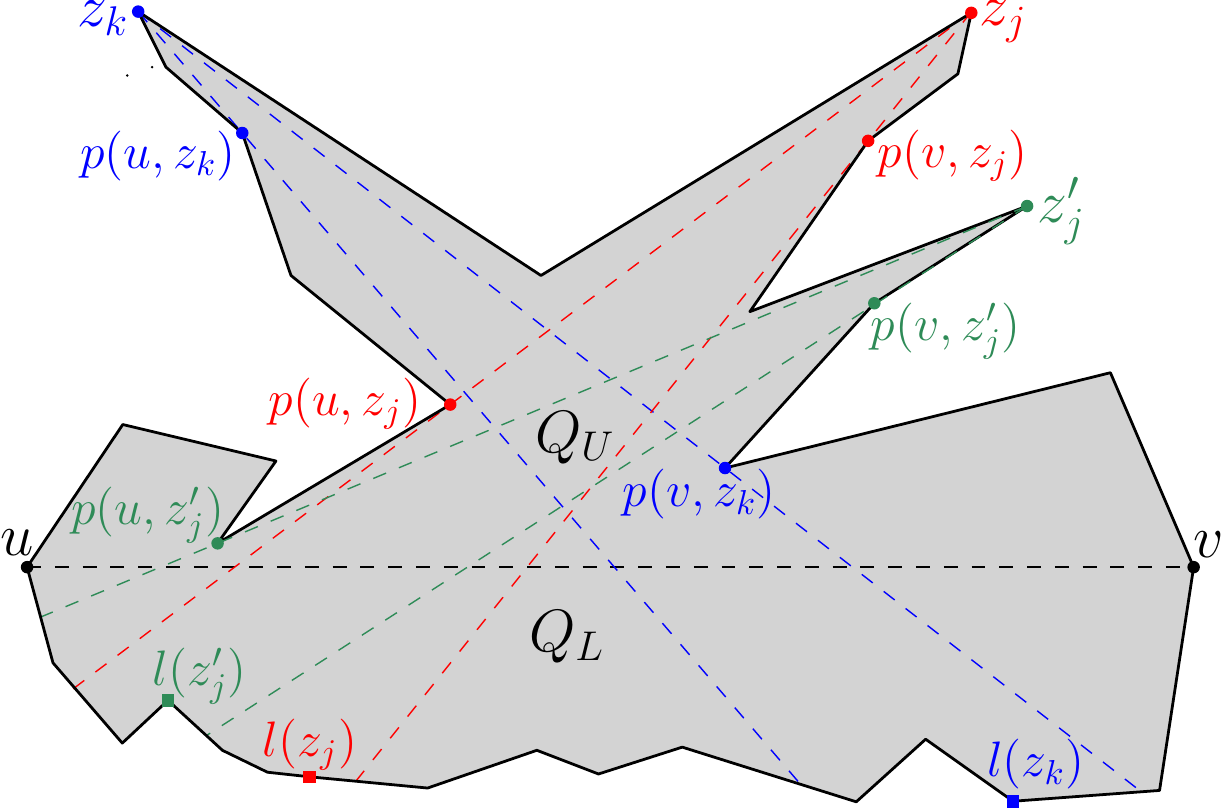}}
  \caption{ The choice of the next primary vertex. }
  \label{proof_case1a}
\end{minipage}
\hspace*{0.01\textwidth}
\begin{minipage}{0.55\textwidth}
  \centerline{\includegraphics[width=\textwidth]{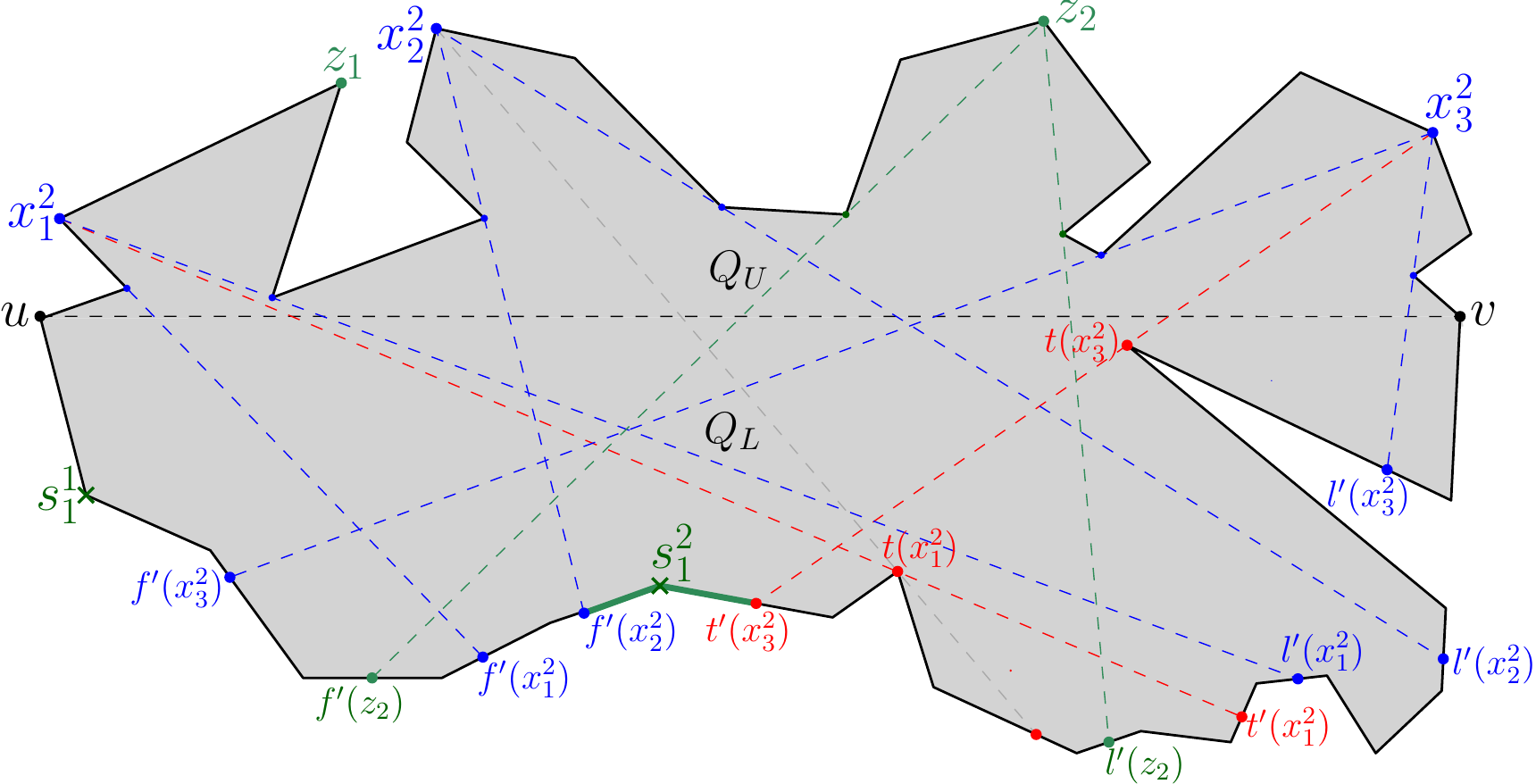}}
  \caption{Figure for Case 1, 
  where $A^2 = \{x_1^2\}$, $B^2 \supseteq \{x_2^2\}$, $C^2 = \{x_3^2\}$, 
  and $s_1^2 \in \mathcal{CI}(A^2 \cup B^2 \cup C^2)$. }
  \label{fig_case1}
\end{minipage}
\end{figure}

\vspace{-0.5em}
Consider Case 1, where $\mathcal{CI}(A^k \cup B^k \cup C^k) \neq \emptyset$ (see Figure \ref{fig_case1}, and line \ref{vg_pcode_sp1:16} of Algorithm \ref{vg_pcode_sp1}).
Here, the algorithm places a vertex guard $s_3^k$ at any vertex belonging to $\mathcal{CI}(A^k \cup B^k \cup C^k)$ (in line \ref{vg_pcode_sp1:19}). 
So, every vertex in $\mathcal{OVV}^{-}(z_k)$ is visible from $s_3^k$ and are hence marked after the placement of the guard at $s_3^k$. 

\begin{lemma} \label{case1}
If $\mathcal{CI}(A^k \cup B^k \cup C^k) \neq \emptyset$, 
then all vertices belonging to $\mathcal{OVV}^{-}(z_k)$ are visible 
from the vertex guard placed at $s_3^k$.
Therefore, no vertex $ x_i^k \in \mathcal{OVV}^{-}(z_k)$ can be a primary vertex $z_j$ 
for any $j > k$.
\end{lemma}

\begin{figure}[H]
  \centerline{\includegraphics[width=0.86\textwidth]{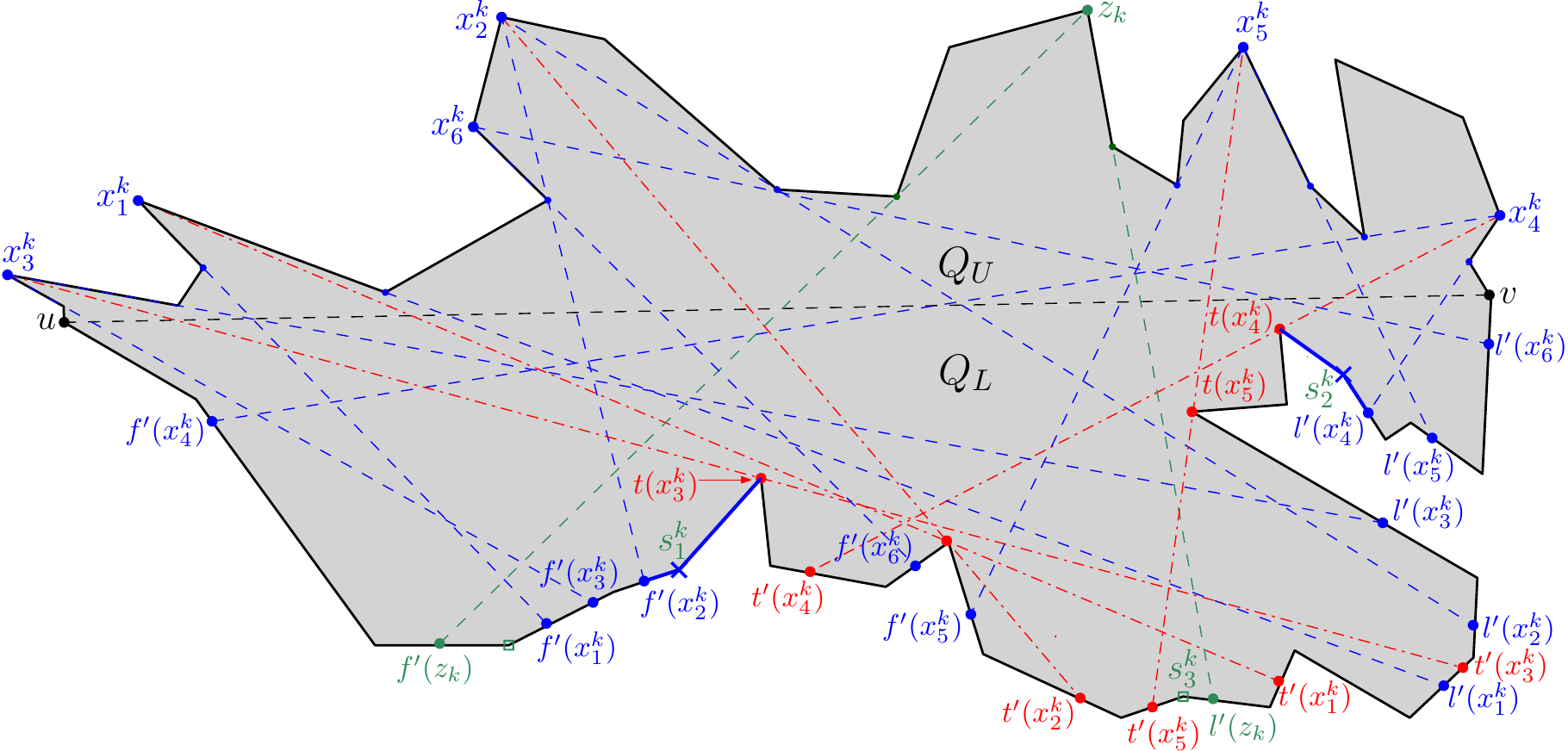}}
  \caption{Figure for Case 2a, 
  where $A^k = \{x_1^k,x_3^k\}$, $B^k = \{x_2^k,x_6^k\}$, $C^k = \{x_4^k,x_5^k\}$, 
  $\mathcal{CI}(A^k \cup B^k \cup C^k) = \emptyset$, 
  $s_1^k \in \mathcal{CI}(A^k)$, $s_2^k \in \mathcal{CI}(C^k)$, 
  and $s_3^k \in \mathcal{CI}(B^k)$. }
  \label{fig_case2a}
\end{figure}

Now consider Case 2a, where $\mathcal{CI}(A^k \cup B^k \cup C^k) = \emptyset$, $\mathcal{CI}(A^k) \neq \emptyset$, and $\mathcal{CI}(C^k) \neq \emptyset$ (see Figure \ref{fig_case2a}).  
Here, Algorithm \ref{vg_pcode_sp1} places three vertex guards, viz. $s_1^k \in \mathcal{CI}(A^k)$,
$s_2^k \in \mathcal{CI}(C^k)$, and
$s_3^k \in \mathcal{CI}(B^k)$ (in lines  \ref{vg_pcode_sp1:21}, \ref{vg_pcode_sp1:33} and \ref{vg_pcode_sp1:19} respectively).
So, the three vertex guards placed by the algorithm together see all the vertices of $\mathcal{OVV}^{-}(z_k)$, and of course $z_k$ itself. 
It is important to note that $s_1^k$ or $s_2^k$ may not belong to $\mathcal{VVP}^{-}(z_k)$, but $s_3^k$ must belong to $\mathcal{VVP}^{-}(z_k)$.

\begin{lemma} \label{case2a}
If $\mathcal{CI}(A^k \cup B^k \cup C^k) = \emptyset$, $\mathcal{CI}(A^k) \neq \emptyset$, and $\mathcal{CI}(C^k) \neq \emptyset$, 
then all vertices belonging to $\mathcal{OVV}^{-}(z_k)$ are visible 
from at least one of the three vertex guards 
placed at $s_1^k$, $s_2^k$, and $s_3^k$.
Therefore, no vertex $ x_i^k \in \mathcal{OVV}^{-}(z_k)$ can be a primary vertex $z_j$ 
for any $j > k$.
\end{lemma}

\begin{figure}[H]
  \centerline{\includegraphics[width=0.91\textwidth]{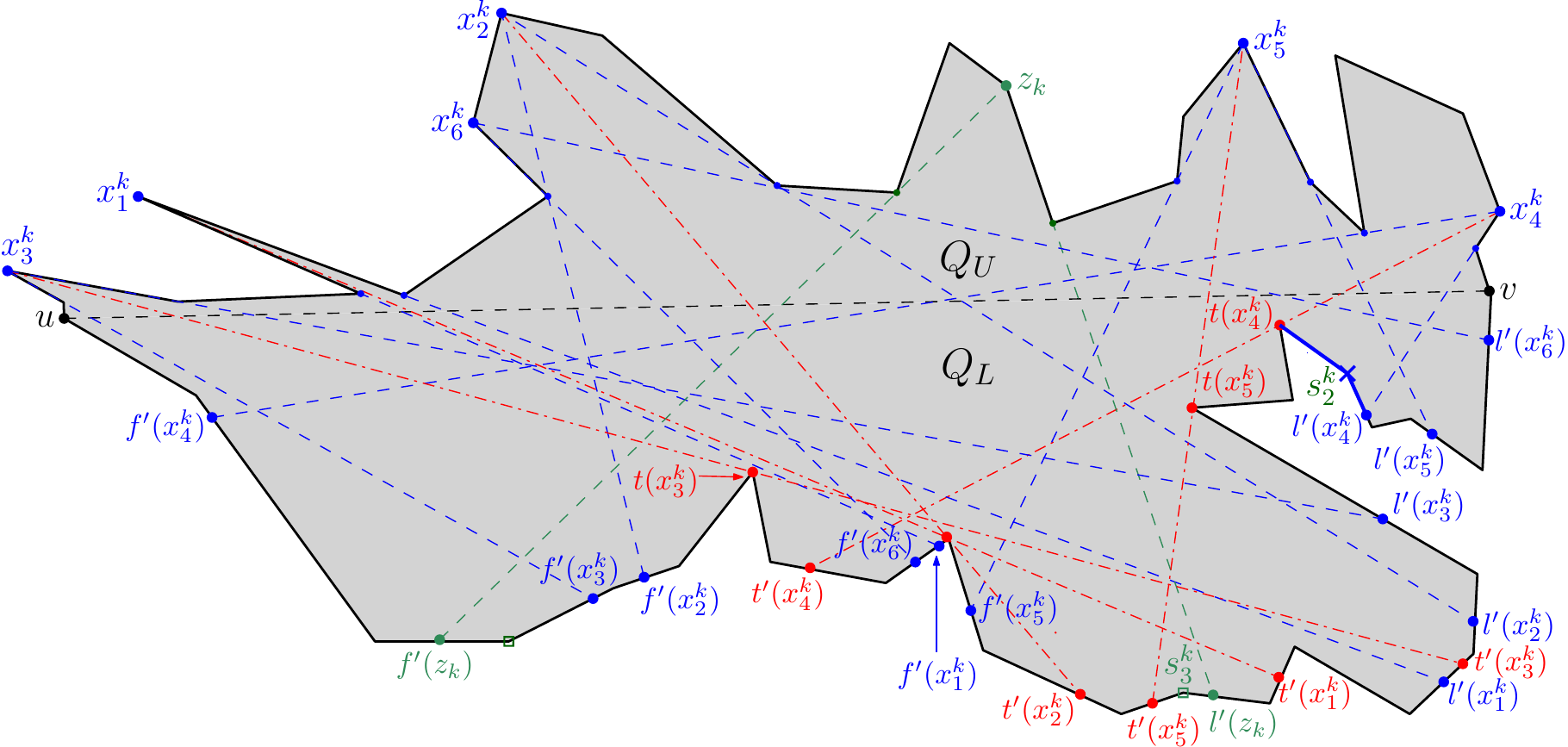}}
  \caption{Figure for Case 2b, 
  where $A^k = \{x_1^k,x_3^k\}$, $B^k = \{x_2^k,x_6^k\}$, $C^k = \{x_4^k,x_5^k\}$, 
  $\mathcal{CI}(A^k \cup B^k \cup C^k) = \emptyset$, 
  $\mathcal{CI}(A^k) = \emptyset$, $s_2^k \in \mathcal{CI}(C^k)$, 
  and $s_3^k \in \mathcal{CI}(B^k)$. }
  \label{fig_case2b}
\end{figure}

Consider Case 2b, where $\mathcal{CI}(A^k \cup B^k \cup C^k) = \emptyset$, $\mathcal{CI}(A^k) = \emptyset$, and $\mathcal{CI}(C^k) \neq \emptyset$ (see Figure \ref{fig_case2b}). 
Observe that a vertex of $\mathcal{CI}(C^k)$ may not lie within $\mathcal{VVP}^{-}(z_k)$, but rather lie on $bd_{cc}(l'(z_k,v))$.
The algorithm places two vertex guards, one at a vertex $s_2^k \in \mathcal{CI}(C^k)$, 
and another one at a vertex $s_3^k \in \mathcal{CI}(B^k)$.
Note that $s_2^k$ or $s_3^k$ may see some of the vertices of $A^k$, which may get marked as a consequence. However, assuming that none of the vertices of $A^k$ are marked due to the placement of $s_2^k$ and $s_3^k$ (in lines \ref{vg_pcode_sp1:33} and \ref{vg_pcode_sp1:19} respectively of Algorithm \ref{vg_pcode_sp1}), 
a third vertex guard $s_1^k$ is also chosen to guard at least a subset of $A^k$, since $\mathcal{CI}(A^k) = \emptyset$. 
In order to choose the vertex guard $s_1^k$ for guarding a subset of $A^k$, $\mathcal{VVP}^{-}(z_k)$ is traversed counter-clockwise starting from $f(z_k)$, till a vertex $y$ is encountered such that there exists a vertex $x_i^k \in A^k$ which is visible from $y$ but not from any subsequent vertex of $\mathcal{VVP}^{-}(z_k)$. 
So, this vertex $y = t(x_i^k)$ is chosen (in line \ref{vg_pcode_sp1:27} for Algorithm \ref{vg_pcode_sp1}) as the vertex guard $s_1^k$.
It can be seen that once such a guard is placed at $s_1^k = t(x_i^k)$, some of the vertices in $A^k$ are visible from $x_i^k$, and are therefore marked.
Let us denote the remaining vertices of $A^k$ that are still unmarked as $A^k_1$.
We have the following lemma.

\begin{lemma} \label{case2b}
If $\mathcal{CI}(A^k \cup B^k \cup C^k) = \emptyset$, $\mathcal{CI}(A^k) = \emptyset$, and $\mathcal{CI}(C^k) \neq \emptyset$, 
then all vertices belonging to $\mathcal{OVV}^{-}(z_k) \setminus A^k_1$ are visible from one of the vertex guards 
placed at $s_1^k$, $s_2^k$, and $s_3^k$.
Therefore, no vertex $x_i^k \in \mathcal{OVV}^{-}(z_k)$ can be a primary vertex $z_j$ 
for any $j > k$ unless $x_i^k \in A^k_1$.
\end{lemma}

\begin{figure}[H]
  \centerline{\includegraphics[width=0.88\textwidth]{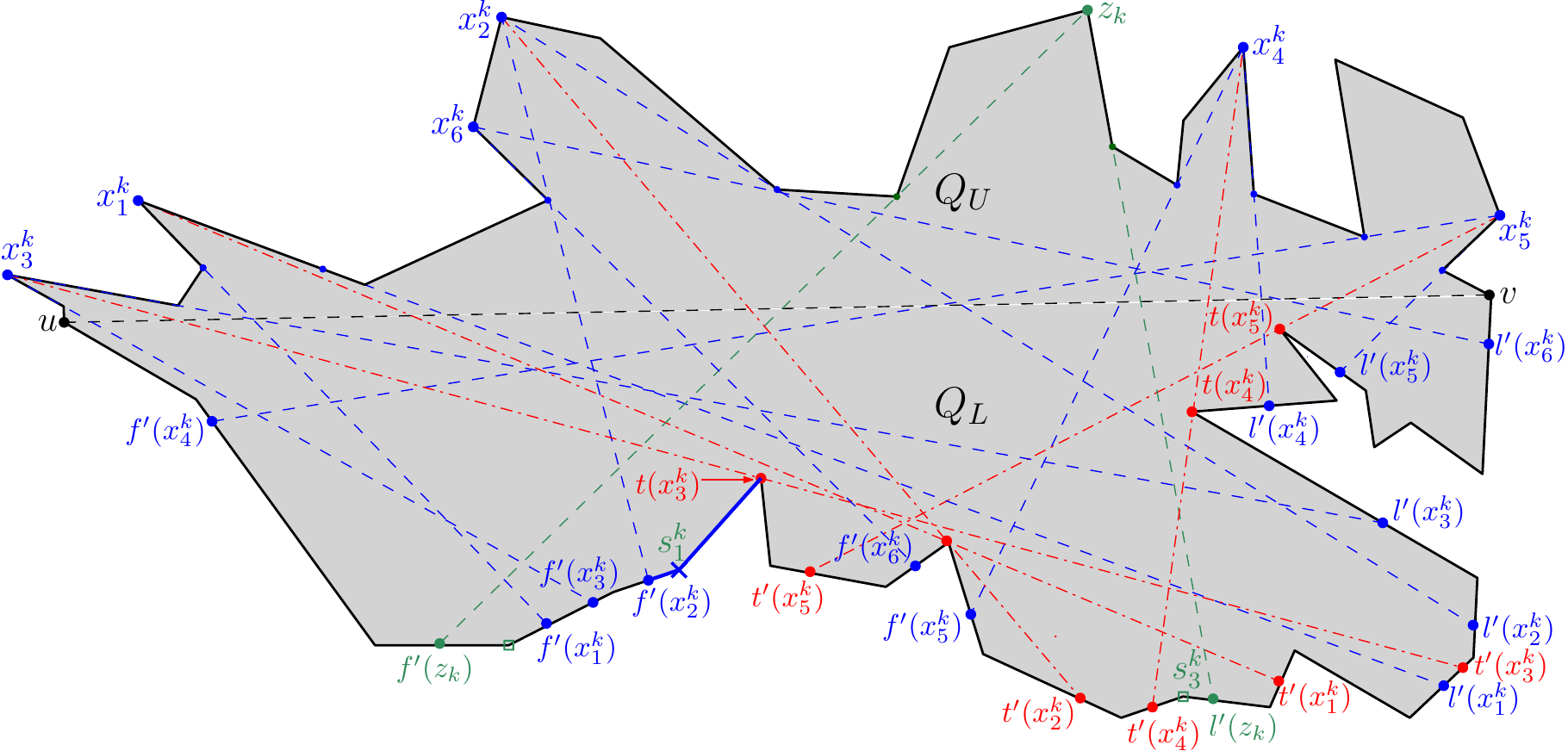}}
  \caption{Figure for Case 2c, 
  where $A^k = \{x_1^k,x_3^k\}$, $B^k = \{x_2^k,x_6^k\}$, $C^k = \{x_4^k,x_5^k\}$, 
  $\mathcal{CI}(A^k \cup B^k \cup C^k) = \emptyset$, 
  $s_1^k \in \mathcal{CI}(A^k)$, $\mathcal{CI}(C^k) = \emptyset$, 
  and $s_3^k \in \mathcal{CI}(B^k)$. }
  \label{fig_case2c}
\end{figure}

\vspace{-0.6em}
Consider Case 2c, where $\mathcal{CI}(A^k \cup B^k \cup C_k) = \emptyset$, $\mathcal{CI}(A^k) \neq \emptyset$, and $\mathcal{CI}(C^k) = \emptyset$ (see Figure \ref{fig_case2c}). 
Observe that a vertex of $\mathcal{CI}(A^k)$ may not lie within $\mathcal{VVP}^{-}(z_k)$, but rather lie on $bd_{cc}(l'(z_k,v)$.
The algorithm places two vertex guards, one at a vertex $s_1^k \in \mathcal{CI}(A^k)$, 
and another one at a vertex $s_3^k \in \mathcal{CI}(B^k)$
(in lines \ref{vg_pcode_sp1:33} and \ref{vg_pcode_sp1:21} respectively of Algorithm \ref{vg_pcode_sp1}).
Note that $s_1^k$ or $s_3^k$ may see some of the vertices of $C^k$, which may get marked as a consequence. However, as a worst case, we assume that none of the vertices of $C^k$ are marked due to the placement of $s_1^k$ and $s_3^k$. In such a scenario, a third vertex guard $s_2^k$ is chosen to guard a subset of $C^k$, since $\mathcal{CI}(C^k) = \emptyset$. 
In order to choose $s_2^k$, $\mathcal{VVP}^{-}(z_k)$ is traversed counter-clockwise, starting from $f(z_k)$, till a vertex $y$ is encountered such that there exists a vertex $x_i^k \in C^k$ which is visible from $y$ but not from any subsequent vertex of $\mathcal{VVP}^{-}(z_k)$. 
So, this vertex $y$, which is effectively the vertex of $\mathcal{VVP}^{-}(z_k)$ immediately preceding $t'(x_i^k)$, 
and hence denoted by $prev(t'(x_i^k))$, is chosen (in line \ref{vg_pcode_sp1:39} of Algorithm \ref{vg_pcode_sp1}) as the vertex guard $s_2^k$.
It can be seen that once a guard is placed at $s_2^k$, a subset of the vertices in $C^k$ are visible from $x_i^k$, and hence marked. 
Let us denote the remaining vertices of $C^k$ that are still unmarked as $C^k_1$. The following lemma summarizes Case 2c.

\begin{lemma} \label{case2c}
If $\mathcal{CI}(A^k \cup B^k \cup C^k) = \emptyset$, $\mathcal{CI}(A^k) \neq \emptyset$, and $\mathcal{CI}(C^k) = \emptyset$, 
then all vertices belonging to $\mathcal{OVV}^{-}(z_k) \setminus C^k_1$ are visible from one of the vertex guards 
placed at $s_1^k$, $s_2^k$, and $s_3^k$.
Therefore, a vertex $x_i^k \in \mathcal{OVV}^{-}(z_k)$ cannot be a primary vertex $z_j$ 
for any $j > k$ unless $x_i^k \in C^k_1$.
\end{lemma}

\begin{figure}[H]
  \centerline{\includegraphics[width=0.82\textwidth]{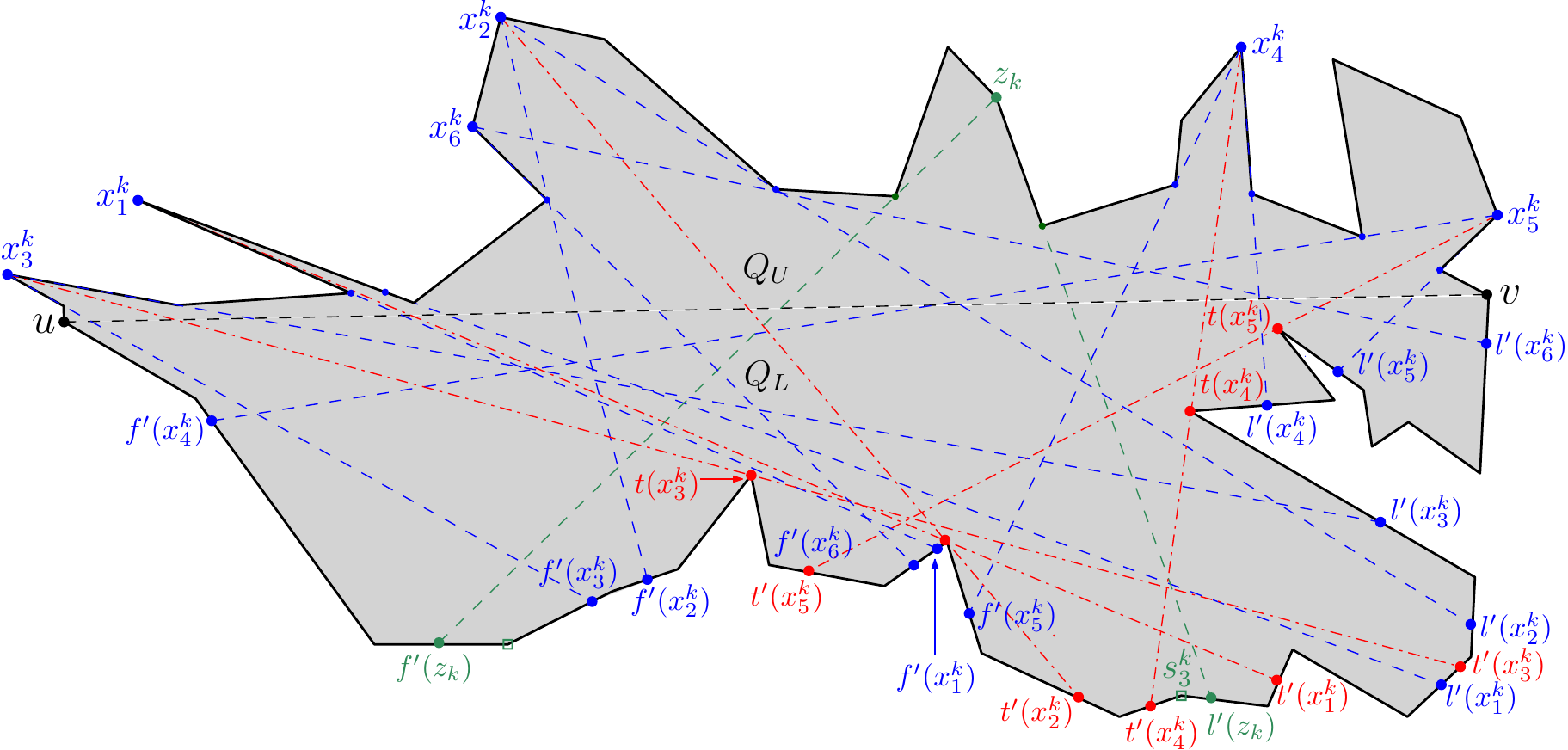}}
  \caption{Figure for Case 2d, 
  where $A^k = \{x_1^k,x_3^k\}$, $B^k = \{x_2^k,x_6^k\}$, $C^k = \{x_4^k,x_5^k\}$, 
  $\mathcal{CI}(A^k \cup B^k \cup C^k) = \emptyset$, 
  $\mathcal{CI}(A^k) = \emptyset$ and $\mathcal{CI}(C^k) = \emptyset$,
  and $s_3^k \in \mathcal{CI}(B^k)$. }
  \label{fig_case2d}
\end{figure}
             
Consider Case 2d, where $\mathcal{CI}(A^k \cup B^k \cup C^k) = \emptyset$, $\mathcal{CI}(A^k) = \emptyset$, and $\mathcal{CI}(C^k) = \emptyset$ (see Figure \ref{fig_case2d}). 
The algorithm first places a guard at a vertex $s_3^k \in \mathcal{CI}(B^k)$.
Note that $s_3^k$ may see some of the vertices of $A^k$ and $C^k$, which may get marked as a consequence. However, assuming that none of the vertices of $A^k$ or $C^k$ are marked due to the placement of $s_3^k$, 
another vertex guard $s_2^k$ is chosen from $\mathcal{CI}(C^k)$,
following a procedure similar to that in Case 2c. 
Similarly, another vertex guard $s_1^k$ is chosen from $\mathcal{CI}(A^k)$ 
following a procedure similar to that in Case 2b.
It can be seen that once guards are placed at $s_1^k$ and $s_2^k$, 
some subsets of $A^k$ and $C^k$ are visible from them, and hence marked. 
So, let us denote the yet unmarked vertices of $A^k$ and $C^k$ 
as $A^k_1$ and $C^k_1$, respectively. 

\begin{lemma} \label{case2d}
If $\mathcal{CI}(A^k \cup B^k \cup C^k) = \emptyset$, $\mathcal{CI}(A^k) = \emptyset$, and $\mathcal{CI}(C^k) = \emptyset$, 
then all vertices belonging to $\mathcal{OVV}^{-}(z_k) \setminus (A^k_1 \cup C^k_1$) are visible from one of the three vertex guards 
placed at $s_1^k$, $s_2^k$, and $s_3^k$.
Therefore, no vertex $x_i^k \in \mathcal{OVV}^{-}(z_k)$ can be a primary vertex $z_j$ 
for any $j > k$ unless $x_i^k \in (A^k_1 \cup C^k_1)$.
\end{lemma}

\begin{lemma} \label{not_in_OVV}
Let $z_k$ and $z_j$ be any two primary vertices, where $j > k$.
If $z_j \notin \mathcal{OVV}^{-}(z_k)$, then $G_{opt}$ ($ = G^L_{opt}$) must require two distinct vertex guards for guarding both $z_k$ and $z_j$. 
\end{lemma}
\begin{proof}
 For the sake of contradiction, assume there exists a single guard $g \in G_{opt} ( = G^L_{opt})$ that can see both $z_k$ and $z_j$. This is only possible if $g \in (\mathcal{VVP}^{-}(z_k) \cap \mathcal{VVP}^{-}(z_j))$, which in turn means that $(\mathcal{VVP}^{-}(z_k) \cap \mathcal{VVP}^{-}(z_j)) \neq \emptyset$. Therefore, $z_j \in \mathcal{OVV}^{-}(z_k)$ by the definition of $\mathcal{OVV}^{-}(z_k)$, a contradiction.
\end{proof}

\begin{lemma} \label{not_in_AorC}
Let $z_k$ and $z_j$ be any two primary vertices, where $j > k$.
Assume that $A^k_1 \neq \emptyset$ and $C^k_1 \neq \emptyset$.
If $z_j \notin (A^k_1 \cup C^k_1)$, 
then $G_{opt}$ (=$G^L_{opt}$) must require two distinct vertex guards for guarding both $z_k$ and $z_j$. 
\end{lemma}
\begin{proof}
After the placement of guards for $z_k$, the only vertices of $\mathcal{OVV}^{-}(z_k)$ that are still unmarked belong to $A^k_1$ or $C^k_1$.
Since $z_j$ is unmarked when it is chosen as a primary vertex, 
$z_j \notin (A^k_1 \cup C^k_1)$ implies that $z_j \notin \mathcal{OVV}^{-}(z_k)$. 
So, it follows directly from Lemma \ref{not_in_OVV} that $G_{opt}$ (=$G^L_{opt}$) must require two distinct vertex guards for guarding both $z_k$ and $z_j$.
\end{proof}

\begin{lemma} \label{cond_bound}
For every $k \in \{1,2,\dots,|Z|-1\}$, assume that $z_j \notin (A^k_1 \cup C^k_1)$ for any $j$, $k < j \leq |Z|$. Then, $|S| \leq 3\cdot|G^L_{opt}|=3\cdot|G_{opt}|$.
\end{lemma}
\begin{proof}
We know from Lemma \ref{not_in_AorC}, that for any arbitrary pair $z_k$ and $z_j$, where $k,j \in \{1,2,\dots,|Z|\}$ and $j > k$, $G_{opt}$ (=$G^L_{opt}$) requires two distinct vertex guards for guarding both $z_k$ and $z_j$.
Thus, applying Lemma \ref{not_in_AorC} repeatedly over all such possible pairs shows that $G_{opt}$ requires as many guards as the number of primary vertices, i.e. $|Z| \leq |G_{opt}|$.
Since at most three vertex guards $s_1^k$, $s_2^k$, and $s_3^k$ are placed corresponding to each primary vertex $z_k \in Z$, $|S| \leq 3\cdot|Z|$.
So, combining the above two inequalities, 
we obtain $|S| \leq 3\cdot|Z| \leq 3\cdot|G_{opt}|$.
\end{proof}

So far we have assumed the special case where, for every $k$, there is no $j>k$ such that a 
primary vertex $z_j$ belongs to $A^k_1\cup C^k_1$. 
Therefore, we now focus on the general case where, for some $j > k$, 
we have the primary vertex $z_j \in A^k_1 \cup C^k_1$. 
Let $A^k = \{a_1,a_2,a_3,\dots,a_m\}$ where $t(a_1) \prec t(a_2) \prec t(a_3) \prec \dots \prec t(a_m)$.
Observe that, if we consider any two arbitrary vertices $a_i, a_j \in A^k$, then these vertices
may be \emph{overlapping}, i.e. $\mathcal{VVP}^{-}(a_i) \cap \mathcal{VVP}^{-}(a_j) \neq \emptyset$,
or \emph{disjoint}, i.e. $\mathcal{VVP}^{-}(a_i) \cap \mathcal{VVP}^{-}(a_j) = \emptyset$.
So, it is possible to create a partition of $A^k$, such that each set in the partition 
consists of a particular vertex $a_i$, called the \emph{leading vertex}, and all the other vertices of $A^k$ that are overlapping with $a_i$.
Obviously, two leading vertices belonging to different sets of such a partition are always disjoint.
Therefore, if a subset of $A^k$ is formed by choosing the leading vertex from each set of the partition, 
then we obtain a \emph{maximal disjoint subset} of $A^k$, i.e. a maximal subset of $A^k$ whose elements are all pairwise disjoint. \\

\vspace*{-0.5em}
In the following, we first formally define \emph{maximum disjoint subsets} of $A^k$ and $C^k$, 
and establish various properties of these subsets in Lemmas \ref{disjoint->nestedA} to \ref{ZR^k_i_bound},
Using these properties, we establish a lower bound on $|G_{opt}|$, which finally leads towards obtaining a constant approximation ratio. 

\begin{lemma} \label{disjoint->nestedA}
Consider any two arbitrary vertices $a_i,a_j \in A^k$, where $t(a_i) \prec t(a_j)$. 
If $a_i$ and $a_j$ are disjoint, i.e. if $\mathcal{VVP}^{-}(a_i) \cap \mathcal{VVP}^{-}(a_j) = \emptyset$,
then $\mathcal{VVP}^{-}(a_j)$ is \emph{geometrically nested} inside $\mathcal{VVP}^{-}(a_i)$,
i.e. $t(a_i) \prec f(a_j) \prec l(a_j) \prec t'(a_i)$.
\end{lemma}

\begin{proof}
Since $a_i,a_j \in A^k$, we must have $t(a_i) \prec l(z_k) \prec t'(a_i)$ and also $t(a_j) \prec l(z_k) \prec t'(a_j)$.
Therefore, the only possibility is to have $t(a_i) \prec t(a_j) \prec l(z_k) \prec t'(a_j) \prec t'(a_i)$.
Moreover, if $f(a_j) \prec t(a_i)$ or $t'(a_i) \prec l(a_j)$, 
then $\mathcal{VVP}^{-}(a_i) \cap \mathcal{VVP}^{-}(a_j) \neq \emptyset$,
which contradicts that $a_i$ and $a_j$ are disjoint.
Hence, we must have $t(a_i) \prec f(a_j) \prec t(a_j) \prec l(z_k) \prec t'(a_j) \prec l(a_j) \prec t'(a_i)$.
So, $\mathcal{VVP}^{-}(a_j)$ is \emph{geometrically nested} inside $\mathcal{VVP}^{-}(a_i)$.
\end{proof}

\begin{lemma} \label{disjoint->nestedC}
Consider any two arbitrary vertices $c_i,c_j \in C^k$, where $t'(c_i) \prec t'(c_j)$. 
If $c_i$ and $c_j$ are disjoint, i.e. if $\mathcal{VVP}^{-}(c_i) \cap \mathcal{VVP}^{-}(c_j) = \emptyset$,
then $\mathcal{VVP}^{-}(c_j)$ is \emph{geometrically nested} inside $\mathcal{VVP}^{-}(c_i)$,
i.e. $t'(c_i) \prec f(c_j) \prec l(c_j) \prec t'(c_i)$.
\end{lemma}

\begin{proof}
Since $c_i,c_j \in C^k$, we must have $t'(c_i) \prec l(z_k) \prec t(c_i)$ and also $t'(c_j) \prec l(z_k) \prec t(c_j)$.
Therefore, the only possibility is to have $t'(c_i) \prec t'(c_j) \prec l(z_k) \prec t(c_j) \prec t(c_i)$.
Moreover, if $f(c_j) \prec t'(c_i)$ or $t(c_i) \prec l(c_j)$, 
then $\mathcal{VVP}^{-}(c_i) \cap \mathcal{VVP}^{-}(c_j) \neq \emptyset$,
which contradicts that $c_i$ and $c_j$ are disjoint.
Hence, we must have $t'(c_i) \prec f(c_j) \prec t'(c_j) \prec l(z_k) \prec t(c_j) \prec l(c_j) \prec t(c_i)$.
So, $\mathcal{VVP}^{-}(c_j)$ is \emph{geometrically nested} inside $\mathcal{VVP}^{-}(c_i)$.
\end{proof}

Observe that the size of any maximal disjoint subset of $A^k$ depends on the choice of the leading element for each set of the partition. 
We are interested in choosing the leading elements in such a way so as to construct a
\emph{canonical partitioning} of $A^k$ corresponding to a particular {\it maximum disjoint subset} $L^k \subseteq A^k$, 
defined as follows.
First include $a_1$ in $L^k$ and construct the set $P_1(L^k)$ consisting of all other vertices $a_j \in A^k$
such that $\mathcal{VVP}^{-}(a_1) \cap \mathcal{VVP}^{-}(a_j) \neq \emptyset$.
Note that $a_1 \in P_1(L^k)$.
Also note that, if $A^k \setminus P_1(L^k) = \emptyset$, then $P(L^k) = \{A^k\}$
(i.e. $CI(A^k) \neq \emptyset$, so this corresponds to cases 2a and 2c in Algorithm \ref{vg_pcode_sp1}).
Otherwise, for each $i \in \{2,3,\dots\}$, 
pick the vertex $a_{\sigma(i)} \in A^k$ where $\sigma(i)$ is the least index such that $a_{\sigma(i)} \notin \bigcup_{1\leq j < i}P_j(L^k)$, 
and include $a_{\sigma(i)}$ in $L^k$. 
Construct the set $P_i(L^k) =\{ a_j \in A^k : \mathcal{VVP}^{-}(a_{\sigma(i)}) \cap \mathcal{VVP}^{-}(a_j) \neq \emptyset \}$.
The process is repeated till $ \bigcup_{j:a_j \in L^k} P_j(L^k) = A^k $. 
Thus, a canonical partition $P(L^k) = \{ P_i(L^k) \subseteq A_k : 1 \leq i \leq |P(L^k)| \leq m\}$ is constructed. \\

\vspace*{-0.5em}
Analogously, a similar {\it canonical partitioning of $C^k$} corresponding to a {\it maximum disjoint subset} $R^k \subseteq C^k$ 
is defined as follows.
Let $C^k = \{c_1,c_2,c_3,\dots,c_m\}$ 
where $prev(t'(c_1)) \prec prev(t'(a_2)) \prec prev(t'(a_3)) \prec \dots \prec prev(t'(a_m))$.
First include $c_1$ in $R^k$ and construct the set $P_1(R^k)$ consisting of all other vertices $c_j \in C^k$
such that $\mathcal{VVP}^{-}(c_1) \cap \mathcal{VVP}^{-}(c_j) \neq \emptyset$.
Note that $c_1 \in P_1(R^k)$.
Also note that, if $C^k \setminus P_1(R^k) = \emptyset$, then $PR^k = \{C^k\}$
(i.e. $CI(C^k) \neq \emptyset$, so this corresponds to cases 2a and 2b in Algorithm \ref{vg_pcode_sp1}).
Otherwise, for each $i \in \{2,3,\dots\}$, 
pick the vertex $c_{\sigma(i)} \in C^k$ where $\sigma(i)$ is the least index such that $c_{\sigma(i)} \notin \bigcup_{1\leq j < i}P_j(R^k)$, 
and include $c_{\sigma(i)}$ in $R^k$. 
Construct the set $P_i(R^k) =\{ c_j \in C^k : \mathcal{VVP}^{-}(c_{\sigma(i)}) \cap \mathcal{VVP}^{-}(c_j) \neq \emptyset \}$.
The process is repeated till $ \bigcup_{j:c_j \in R^k} P_j(R^k) = C^k $. 
Thus, a canonical partition $P(R^k) = \{ P_i(R^k) \subseteq C_k : 1 \leq i \leq |P(R^k)| \leq m\}$ is constructed. \\

We now study the properties of $L^k$ and $R^k$ as constructed above.
Firstly, it is easy to see that, by their very construction,
$L^k$ is a \emph{maximal} disjoint subset of $A^k$, and $R^k$ is a \emph{maximal} disjoint subset of $C^k$.
We show that $L^k$ is also a \emph{maximum} disjoint subset of $A^k$, 
using an interesting pairwise intersection property established in the following lemma. 

\begin{lemma} \label{L^k_i_pairwise}
For every $P_i(L^k) \in P(L^k)$ and for any two vertices $a,a' \in P_i(L^k)$, 
we have $\mathcal{VVP}^{-}(a) \cap \mathcal{VVP}^{-}(a') \neq \emptyset$.
\end{lemma}

\begin{proof}
Without loss of generality let us assume that $t(a) \prec t(a')$.
For the sake of contradiction, suppose $\mathcal{VVP}^{-}(a) \cap \mathcal{VVP}^{-}(a') = \emptyset$.
This means $t(a) \prec f(a')$ and $l(a') \prec t'(a)$, i.e. $\mathcal{VVP}^{-}(a')$ is geometrically nested within $\mathcal{VVP}^{-}(a)$.
On the other hand, we know that both $a$ and $a'$ are overlapping with the leading vertex $a_{\sigma(i)}$ of $P_i(L^k)$.
By the construction of $P_i(L^k)$, we have $t(a_{\sigma(i)}) \prec t(a) \prec t'(a) \prec t'(a_{\sigma(i)})$.
Therefore, $t(a) \prec f(a')$ implies $t(a_{\sigma(i)}) \prec f(a')$, and $l(a') \prec t'(a)$ implies $l(a') \prec t'(a_{\sigma(i)})$.
However, if both $t(a_{\sigma(i)}) \prec f(a')$ and $l(a') \prec t'(a_{\sigma(i)})$ are true, 
then $\mathcal{VVP}^{-}(a_{\sigma(i)}) \cap \mathcal{VVP}^{-}(a') = \emptyset$,
contradicting the initial assumption that $a' \in P_i(L^k)$.
\end{proof}

\begin{lemma} \label{L^k_maximum}
$L^k$ is a {\it maximum disjoint subset} of $A^k$.
\end{lemma}

\begin{proof}
On the contrary, assume that there exists a larger maximal disjoint subset of $A^k$, say $L'^k$.
Then, by the pigeonhole principle there exists at least two vertices $a,a' \in L'^k$ 
such that both belong to the same set $P_i(L^k) \in P(L^k)$.
Thus, by Lemma \ref{L^k_i_pairwise}, we have $\mathcal{VVP}^{-}(a) \cap \mathcal{VVP}^{-}(a') \neq \emptyset$,
which contradicts the fact that $L'^k$ is a disjoint subset of $A^k$.
Hence, $L^k$ is a {\it maximum disjoint subset} of $A^k$.
\end{proof}

In the following lemmas, we show in an analogous manner that $R^k$ is also a \emph{maximum} disjoint subset of $C^k$.

\begin{lemma} \label{R^k_i_pairwise}
For every $P_i(R^k) \in P(R^k)$ and for any two vertices $c,c' \in P_i(R^k)$, 
we have $\mathcal{VVP}^{-}(c) \cap \mathcal{VVP}^{-}(c') \neq \emptyset$.
\end{lemma}

\begin{proof}
Without loss of generality let us assume that $prev(t'(c)) \prec prev(t'(c'))$.
For the sake of contradiction, suppose $\mathcal{VVP}^{-}(c) \cap \mathcal{VVP}^{-}(c') = \emptyset$.
This means $t'(c) \prec f(c')$ and $l(c') \prec t(c)$, i.e. $\mathcal{VVP}^{-}(c')$ is geometrically nested within $\mathcal{VVP}^{-}(c)$.
On the other hand, we know that both $c$ and $c'$ are overlapping with the leading vertex $c_{\sigma(i)}$ of $P_i(R^k)$.
By the construction of $P_i(R^k)$, we have $prev(t'(c_{\sigma(i)})) \prec prev(t'(c)) \prec t(c) \prec t(c_{\sigma(i)})$.
Therefore, $t'(c) \prec f(c')$ implies $t'(c_{\sigma(i)}) \prec f(c')$, 
and $l(c') \prec t(c)$ implies $l(c') \prec t(c_{\sigma(i)})$.
However, if both $t'(c_{\sigma(i)}) \prec f(c')$ and $l(c') \prec t(c_{\sigma(i)})$ are true, 
then $\mathcal{VVP}^{-}(c_{\sigma(i)}) \cap \mathcal{VVP}^{-}(c') = \emptyset$,
contradicting the initial assumption that $c' \in P_i(R^k)$.
\end{proof}

\begin{lemma} \label{R^k_maximum}
$R^k$ is a {\it maximum disjoint subset} of $C^k$.
\end{lemma}

\begin{proof}
On the contrary, assume that there exists a larger maximal disjoint subset of $C^k$, say $R'^k$.
Then, by the pigeonhole principle there exists at least two vertices $c,c' \in R'^k$ 
such that both belong to the same set $P_i(R^k) \in P(R^k)$.
Thus, by Lemma \ref{R^k_i_pairwise}, we have $\mathcal{VVP}^{-}(c) \cap \mathcal{VVP}^{-}(c') \neq \emptyset$,
which contradicts the fact that $R'^k$ is a disjoint subset of $C^k$.
Hence, $R^k$ is a {\it maximum disjoint subset} of $C^k$.
\end{proof}

\begin{figure}[H]
\begin{minipage}{0.49\textwidth}
  \centerline{\includegraphics[width=\textwidth]{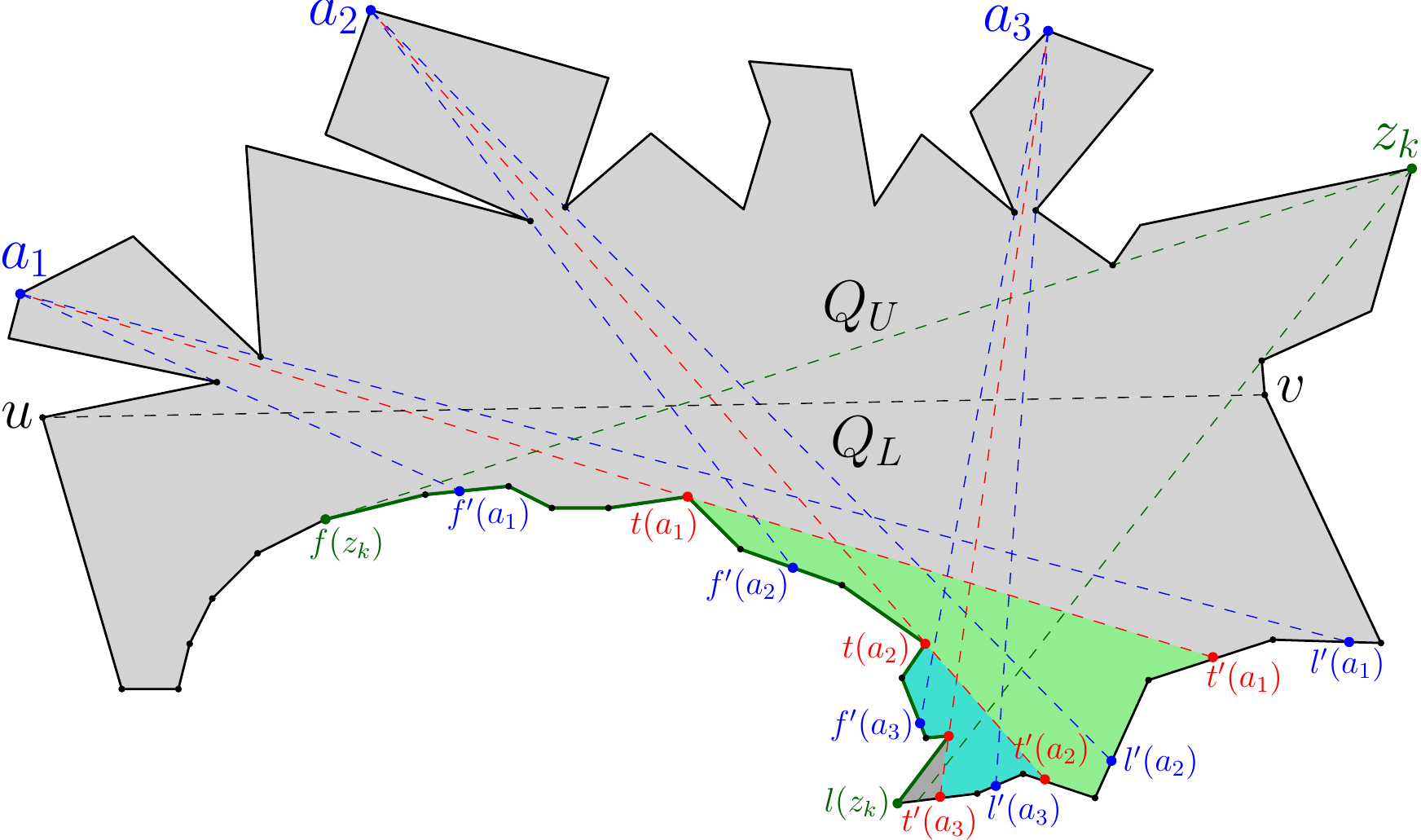}}
  \caption{Figure showing a maximum disjoint subset $ \{a_1, a_2, a_3 \} = L^k \subseteq A^k$ 
         such that $f(a_1) \prec t(a_1) \prec f(a_2) \prec t(a_2) \prec f(a_3) \prec t(a_3) 
         \prec t'(a_3) \prec l(a_3) \prec t'(a_2) \prec l(a_2) \prec t'(a_1) \prec l(a_1) $.}
  \label{L^k_example}
\end{minipage}
\hspace*{0.01\textwidth}
\begin{minipage}{0.49\textwidth}
  \centerline{\includegraphics[width=\textwidth]{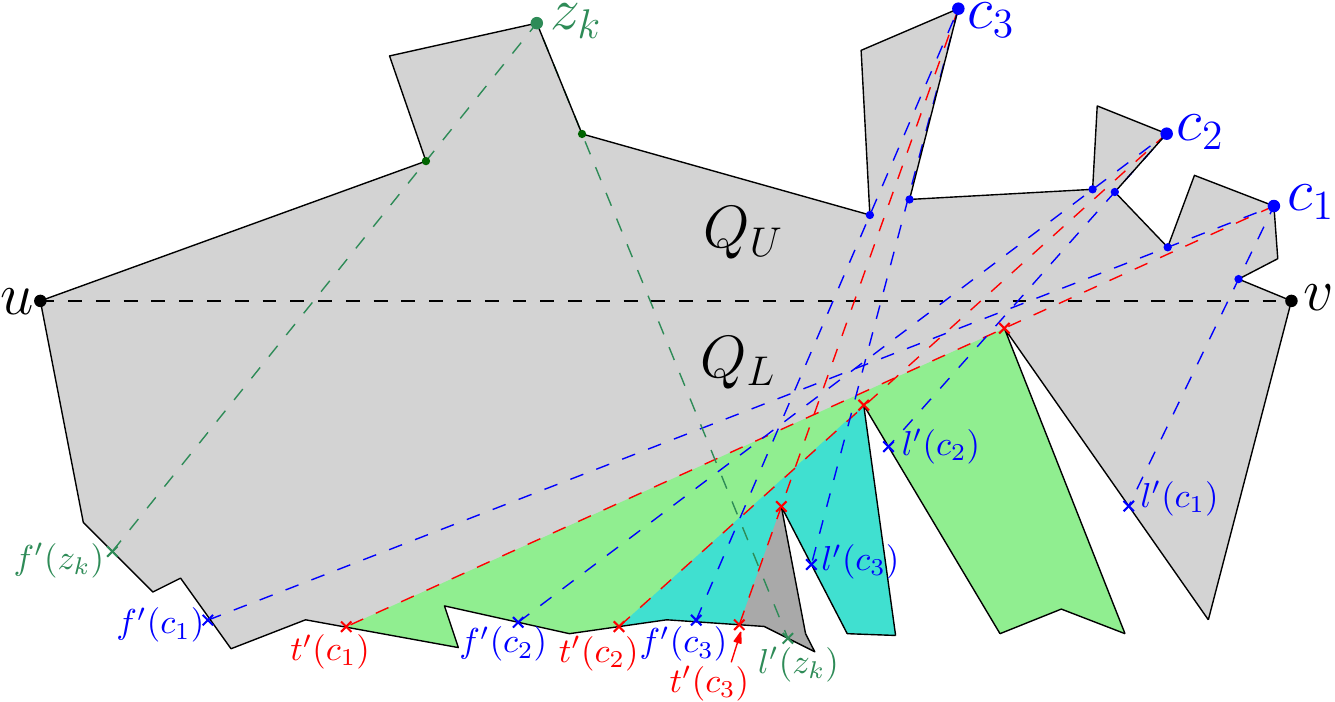}}
  \caption{Figure showing a maximum disjoint subset $\{c_1, c_2, c_3 \} = R^k \subseteq C^k$ 
         such that $f(c_1) \prec t'(c_1) \prec f(c_2) \prec t'(c_2) \prec f(c_3) \prec t'(c_3)
         \prec t(c_3) \prec l(c_3) \prec t(c_2) \prec l(c_2) \prec t(c_1) \prec l(c_1)$. }
  \label{R^k_example}
\end{minipage}
\end{figure}

\vspace*{-0.55em}
The following lemmas are consequences of Lemmas \ref{disjoint->nestedA} and \ref{disjoint->nestedC}.

\begin{lemma} \label{L^k_upper}
For any $k \in \{1,2,\dots,|Z|\}$, any outward vertex guard placed on $bd_{cc}(u,v)$ can see at most one vertex of $L^k$.
\end{lemma}

\begin{lemma} \label{L^k_opt_lb}
For any $k \in \{1,2,\dots,|Z|\}$, any outward guard set 
requires at least $|L^k|$ distinct vertex guards to guard all vertices of $L^k$. 
\end{lemma}

\begin{lemma} \label{R^k_upper}
For any $k \in \{1,2,\dots,|Z|\}$, any outward vertex guard placed on $bd_{cc}(u,v)$ can see at most one vertex of $R^k$.
\end{lemma}

\begin{lemma} \label{R^k_opt_lb}
For any $k \in \{1,2,\dots,|Z|\}$, any outward guard set requires at least $|R^k|$ distinct vertex guards to guard all vertices of $R^k$. 
\end{lemma}

\begin{lemma} \label{opt_both}
For any $k \in \{1,2,\dots,|Z|\}$, any outward optimal guards $g \in G_{opt}^L$ placed on $bd_{cc}(u,v)$ can see at most one vertex of $L^k$ and at most one vertex of $R^k$.
\end{lemma}

\begin{proof}
 Follows directly from Lemmas \ref{L^k_upper} and \ref{R^k_upper}.
\end{proof}

For every set $P_i(L^k)$ belonging to the partition $P(L^k)$ of $L^k$, 
let us define the corresponding two sets $G_i(L^k)$ and $Z_i(L^k)$ as follows.
We define $G_i(L^k) \subseteq G^L_{opt}(=G_{opt})$ to be a minimal subset of the optimal set $G_{opt}$ 
of guards required to guard all the vertices belonging to $P_i(L^k)$.
Similarly, we define $Z_i(L^k) \subseteq Z^L(=Z)$ to be the subset of primary 
vertices chosen by our algorithm from amongst the vertices of $P_i(L^k)$.
Now we are in a position to show that the cardinality of $G_i(L^k)$ 
is lower bounded by $|Z_i(L^k)|$, which never exceeds two. 

\begin{lemma} \label{ZL^k_i_bound}
For every set $P_i(L^k)$ belonging to the partition $P(L^k)$ of $A^k$,  
$|Z_i(L^k)|=|G_i(L^k)|=1$ or $|Z_i(L^k)| = 2 \leq |G_i(L^k)|$, i.e. $|Z_i(L^k)| \leq 2$.
\end{lemma}

\begin{proof}
First, let us consider the case where $|G_i(L^k)|=1$.
This implies that at least one of the conditions below holds:
(i) $f(a) \prec t(a_{\sigma(i)})$ for every $a \in P_i(L^k)$, or
(ii) $t'(a_{\sigma(i)}) \prec l(a)$ for every $a \in P_i(L^k)$.
Observe that, for any vertex $a' \in P_i(L^k)$, if $a'$ is chosen by our algorithm as a later primary vertex $z_j$, 
where $j > k$, then we have $s^j_1 = t(a_{\sigma(i)})$ and $s^j_3 = l(a')$.
If condition (i) is true, then the vertex guard $s^j_1 = t(a_{\sigma(i)})$ sees all vertices belonging to $P_i(L^k)$.
If condition (ii) is true, then the vertex guard $s^j_3 = l(a')$ sees all vertices belonging to $P_i(L^k)$.
Thus, in either case, no further primary vertices will be chosen by our algorithm from $P_i(L^k)$.
So, when $|G_i(L^k)| = 1$, we also have $|Z_i(L^k)| = 1$. \\

\vspace*{-0.33em}
Let us now consider the other case where $|G_i(L^k)| > 1$. 
It is possible that all the vertices belonging to $P_i(L^k)$ are marked by vertex guards 
corresponding to primary vertices chosen from $Z \setminus P_i(L^k)$, in which case we have $|Z_i(L^k)| = 0$.
Otherwise, let $a' \in P_i(L^k)$ be the first primary vertex $z_j = a'$ chosen from $P_i(L^k)$ (see Figure \ref{2inL^k_i}).
Then, as before, we have $s^j_1 = t(a_{\sigma(i)})$ and $s^j_3 = l(a')$.
Observe that, for any vertex $a \in P_i(L^k)$, if $f(a) \prec t(a_{\sigma(i)})$, 
then the vertex guard $s^j_1 = t(a_{\sigma(i)})$ sees $a$, 
and if $t'(a) \prec l(a')$, then the vertex guard $s^j_3 = l(a')$ sees $a$.
Therefore, if any vertex $a \in P_i(L^k)$ is left unmarked even after the placement of vertex guards 
corresponding to $z_j = a'$, then $t(a_{\sigma(i)}) \prec f(a)$.
Since $a \in P_i(L^k)$, by the definition of $P_i(L^k)$ we have 
$\mathcal{VVP}^{-}(a_{\sigma(i)}) \cap \mathcal{VVP}^{-}(a) \neq \emptyset$.
So, the condition $t(a_{\sigma(i)}) \prec f(a)$ would force 
the condition $t'(a_{\sigma(i)}) \prec l(a)$ for any vertex $a \in P_i(L^k)$ 
that is left unmarked even after the placement of vertex guards 
corresponding to $z_j = a'$.
Now, if $z_{j'} = a''$ be another primary vertex chosen from among the yet unmarked vertices of $P_i(L^k)$, then 
$t'(a) \prec t'(a_{\sigma(i)}) \prec l(a'') \prec l(a)$ for any other unmarked vertex $a \in P_i(L^k)$, 
and hence $s^{j'}_3 = l(a'')$ sees all the remaining unmarked vertices of $P_i(L^k)$.  
So, when $|G_i(L^k)| > 1$, we have $|Z_i(L^k)| \leq 2$. 
Therefore, in general, we have $|Z_i(L^k)| \leq |G_i(L^k)|$ and $|Z_i(L^k)| \leq 2$. 
\end{proof}

Analogously, for every set $P_i(R^k)$ belonging to the partition $P(R^k)$ of $R^k$, 
let us define the corresponding two sets $G_i(R^k)$ and $Z_i(R^k)$ as follows.
We define $G_i(R^k) \subseteq G^L_{opt}(=G_{opt})$ to be a minimal subset of the optimal set $G_{opt}$ 
of guards required to guard all the vertices belonging to $P_i(R^k)$.
Similarly, we define $Z_i(R^k) \subseteq Z^L(=Z)$ to be the subset of primary 
vertices chosen by our algorithm from amongst the vertices of $P_i(R^k)$.
Now we are in a position to show that the cardinality of $G_i(R^k)$ 
is lower bounded by $|Z_i(R^k)|$, which never exceeds two. 

\begin{lemma} \label{ZR^k_i_bound}
For every set $P_i(R^k)$ belonging to the partition $P(R^k)$ of $C^k$,  
we have $|Z_i(R^k)| =|G_i(R^k)|=1$ or $|Z_i(R^k)|=2\leq |G_i(R^k)|$, i.e. $|Z_i(R^k)| \leq 2$.
\end{lemma}

\begin{proof}
First, let us consider the case where $|G_i(R^k)| = 1$.
This implies that at least one of the conditions below holds:
(i) $f(c) \prec t'(c_{\sigma(i)})$ for every $c \in P_i(R^k)$, or
(ii) $t(c_{\sigma(i)}) \prec l(c)$ for every $c \in P_i(R^k)$.
Observe that, for any vertex $c' \in P_i(R^k)$, if $c'$ is chosen by our algorithm as a later primary vertex $z_j$, 
where $j > k$, then we have $s^j_2 = prev(t'(c_{\sigma(i)}))$ and $s^j_3 = l(c')$.
If condition (i) is true, then the vertex guard $s^j_2 = prev(t'(c_{\sigma(i)}))$ sees all vertices belonging to $P_i(R^k)$.
If condition (ii) is true, then the vertex guard $s^j_3 = l(c')$ sees all vertices belonging to $P_i(R^k)$.
Thus, in either case, no further primary vertices will be chosen by our algorithm from $P_i(R^k)$.
So, when $|G_i(R^k)| = 1$, we also have $|Z_(R^k)| = 1$. \\

\vspace*{-0.33em}
Let us now consider the other case where $|G_i(R^k)| > 1$. 
It is possible that all the vertices belonging to $P_i(R^k)$ are marked by vertex guards 
corresponding to primary vertices chosen from $Z \setminus P_i(R^k)$, in which case we have $|Z_(R^k)| = 0$.
Otherwise, let $c' \in P_i(R^k)$ be the first primary vertex $z_j = c'$ chosen from $P_i(R^k)$. 
Then, as before, we have $s^j_2 = prev(t'(c_{\sigma(i)}))$ and $s^j_3 = l(c')$.
Observe that, for any vertex $c \in P_i(R^k)$, if $f(c) \prec prev(t'(c_{\sigma(i)}))$, 
then the vertex guard $s^j_2 = prev(t'(c_{\sigma(i)}))$ sees $c$, 
and if $t(c) \prec l(c')$, then the vertex guard $s^j_3 = l(c')$ sees $c$.
Therefore, if any vertex $c \in P_i(R^k)$ is left unmarked even after the placement of vertex guards 
corresponding to $z_j = c'$, then $t'(c_{\sigma(i)}) \prec f(c)$.
Since $c \in P_i(R^k)$, by the definition of $P_i(R^k)$ we have 
$\mathcal{VVP}^{-}(c_{\sigma(i)}) \cap \mathcal{VVP}^{-}(c) \neq \emptyset$.
So, the condition $t'(c_{\sigma(i)}) \prec f(c)$ would force 
the condition $t(c_{\sigma(i)}) \prec l(c)$ for any vertex $c \in P_i(R^k)$ 
that is left unmarked even after the placement of vertex guards corresponding to $z_j = c'$.
Now, if $z_{j'} = c''$ be another primary vertex chosen from among the yet 
unmarked vertices of $P_i(R^k)$, then 
$t(c) \prec t(c_{\sigma(i)}) \prec l(c'') \prec l(c)$ for any other unmarked vertex $c \in P_i(R^k)$, 
and hence $s^{j'}_3 = l(c'')$ sees all the remaining unmarked vertices of $P_i(R^k)$.  
So, when $|G_i(R^k)| > 1$, we have $|Z(R^k)| \leq 2$. 
Therefore, in general, we have $|Z_i(R^k)| \leq |G_i(R^k)|$ and $|Z_i(R^k)| \leq 2$. 
\end{proof}

Two special subsets of vertices of $Q_U$ are constructed, where each subset consists of primary vertices having certain properties 
as well as some special vertices belonging to $A = \bigcup_{k \in \{1,2,\dots,|Z|\}} A^k$ and $C = \bigcup_{k \in \{1,2,\dots,|Z|\}} C^k$.
Then, lower bounds are established on the number of optimal guards required to guard just the vertices belonging to these subsets 
(see Lemmas \ref{GA_bound} and \ref{GC_bound}), which ultimately leads us to a lower bound on $|G_{opt}^L|$,
which is equal to $|G_{opt}|$ in the current scenario. \\

\vspace*{-0.5em}
Let $Z_A$ denote the subset of primary vertices such that $A^k$ is non-empty for each $z_k \in Z_A$.
Also, let $Z_A^{prev}$ denote the subset of primary vertices where each $z_k \in Z_A^{prev}$ belongs to $A^j$ 
for some previously chosen primary vertex $z_j \in Z_A$, where $j<k$.
In other words, we have $Z_A = \{ z_k \in Z : A^k\neq \emptyset \}$, and
$Z_A^{prev} = \{ z_k \in Z : \exists z_j \in Z_A \mbox{ such that } j < k \mbox{ and } z_k \in A^j \}$.
Finally, for every $z_k \in Z_A$, if the corresponding guard $s^k_1$ is placed at $t(a)$,
then $a \in A^k$ is denoted by $t1(k)$ and included in the set $Y_A$,
i.e $Y_A = \bigcup_{z_k \in Z_A} \{t1(k) = a \in A^k : (\forall a'\in A^k) \hspace*{1mm} t(a) \prec t(a') \}$.
Observe that, 
for every vertex $a \in Y_A$, $a$ is marked by $s^k_1$ placed due to the corresponding $z_k$, where $a \in A^k$. 
Also, $Y_A$ may be partitioned into the two sets $Y_A^{prev} = \bigcup_{z_k \in (Z_A \cap Z_A^{prev}) } t1(k)$ and $Y_A \setminus Y_A^{prev}$. \\ 

\vspace*{-0.5em}
Similarly, let $Z_C$ denote the subset of primary vertices such that $C^k$ is non-empty for each $z_k \in Z_C$.
Also, let $Z_C^{prev}$ denote the subset of primary vertices where each $z_k \in Z_C^{prev}$ belongs to $C^j$ 
for some previously chosen primary vertex $z_j \in Z_C$, where $j<k$.
In other words, we have $Z_C = \{ z_k \in Z : C^k\neq \emptyset \}$, and
$Z_C^{prev} = \{ z_k \in Z : \exists z_j \in Z_C \mbox{ such that } j < k \mbox{ and } z_k \in C^j \}$.
Finally, 
for every $z_k \in Z_C$, if the corresponding guard $s^k_2$ is placed at $prev(t'(c))$,
then $c \in C^k$ is denoted by $t'1(k)$ and included in the set $Y_C$,
i.e. $Y_C = \bigcup_{z_k \in Z_C} \{t'1(k) = c \in C^k : (\forall c'\in C^k) \hspace*{1mm} t'(c) \prec t'(c') \}$. 
Observe that, 
for every vertex $c \in Y_C$, $c$ is marked by the vertex guard $s^k_2$ placed due to the corresponding $z_k$, where $c \in C^k$.
Also, $Y_C$ may be partitioned into the two sets $Y_C^{prev} = \bigcup_{z_k \in (Z_C \cap Z_C^{prev}) } t'1(k)$ and $Y_C \setminus Y_C^{prev}$. \\ 

\vspace*{-0.5em}
Let $Z_B = Z \setminus (Z_A \cup Z_A^{prev} \cup Z_C \cup Z_C^{prev})$,
and let $G_B$ denote the minimal subset of $G_{opt}^L$ required to see all vertices of $Z_B$. 

\begin{lemma} \label{GB_bound}
$|G_B| = |Z_B|$.
\end{lemma}

\begin{proof}
Observe that, for each vertex $ z_k \in Z_B$,
$A^k = \emptyset$ and $C^k = \emptyset$, i.e. $\mathcal{OVV}^{-}(z_k) = B^k$.
Moreover, $z_k$ does not belong to $\mathcal{OVV}^{-}(z_j)$ for any other $z_j \in Z$ where $j<k$.
Therefore, $\mathcal{OVV}^{-}(z_k) \cap \mathcal{OVV}^{-}(z_j) = \emptyset$ for any other $z_j \in Z$,
and by Lemma \ref{not_in_OVV}, we know that a single distinct optimal guard is required 
for guarding each $z_k \in Z_B$. Hence, $|G_B| = |Z_B|$.
\end{proof}

\vspace*{-0.5em}
Let us define the minimal subsets $G_A^{prev} \subseteq G_{opt}^L$ and $G_C^{prev} \subseteq G_{opt}^L$ of 
optimal guards required for guarding all vertices belonging to 
$Z_A^{prev} \cup Y_A^{prev}$ and $Z_C^{prev} \cup Y_C^{prev}$ respectively.
In the lemmas below, lower bounds are established on the sizes of the sets $G_A^{prev}$ and $G_B^{prev}$ respectively. 

\begin{lemma} \label{ZA_prev}
$|G_A^{prev}| \geq |Z_A^{prev}|$.
\end{lemma}

\begin{proof}
Assume that $\mathcal{VVP}^{-}(z_j) \cap \mathcal{VVP}^{-}(z_{j'}) = \emptyset$ 
for every pair $z_j, z_{j'} \in Z_A^{prev}$.
Then clearly a distinct optimal guard is required in $G_A^{prev}$ for guarding each primary vertex in $Z_A^{prev}$,
and therefore $|G_A^{prev}| \geq |Z_A^{prev}|$. 
However, as per Lemma \ref{L^k_i_pairwise}, there may exist some pair $z_j, z_{j'} \in Z_A^{prev}$ 
for which $\mathcal{VVP}^{-}(z_j) \cap \mathcal{VVP}^{-}(z_{j'}) = \emptyset$ does not hold,
and a single optimal guard in $g \in G_A^{prev}$ is sufficient to see both $z_j$ and $z_{j'}$.
Then, both $z_j$ and $z_{j'}$ belong to the same set $P_i(L^k)$ in the partition $P(L^k)$ of $A^k$ 
corresponding to some $z_k \in Z_A$ (see Figure \ref{2inL^k_i}).
We prove below that, in such a situation, an additional guard for $t1(j) \in Y_A^{prev}$  is required 
in $G_A^{prev}$ to compensate for the single optimal that sees two primary vertices in $Z_A^{prev}$.
Recall that two primary vertices in $Z_A^{prev}$ can be overlapping 
only if they both belong to the same set in $P(L^k)$.
Moreover, since we know from Lemma \ref{ZL^k_i_bound} that at most only two primary vertices can be chosen from any $P_i(L^k)$,
such compensation can be applied to every $P_i(L^k)$ for which $|P_i(L^k) \cap Z_A^{prev}| = 2$, 
and hence, $|G_A^{prev}| \geq |Z_A^{prev}|$ still holds. \\

\vspace*{-0.33em}
Let $a_{\sigma(i)} \in A^k$ be the leading vertex belonging to $P_i(L^k)$. 
Let us assume without loss of generality that $j < j'$, i.e. $z_j$ was chosen as a primary vertex earlier than $z_{j'}$ by Algorithm \ref{vg_pcode_sp1}.
Then, observe that $t1(j) = a_{\sigma(i)}$, and so $a_{\sigma(i)} \in Y_A^{prev}$.
As $z_{j'}$ is unmarked even after placement of $s^j_1 = t(a_{\sigma(i)})$ and $s^j_3 = l(z_j)$,
we have $t(a_{\sigma(i)}) \prec f(z_{j'})$ and $l(z_j) \prec t'(z_{j'})$, 
and they can both be possible only if $t(a_{\sigma(i)}) \prec f(z_{j'}) \prec t(z_{j'}) \prec t(z_j) \prec t'(z_j) \prec l(z_j) \prec t'(z_{j'})$.
Thus, the common optimal guard $g \in G_A^{prev}$ that sees both $z_j$ and $z_{j'}$ must lie on $bd_{cc}(f(z_{j'}),t(z_{j'}))$.
But in that case, $t(a_{\sigma(i)}) \prec g \prec t'(a_{\sigma(i)})$, 
which means that $g$ does not see $a_{\sigma(i)} = t1(j)$,
and so a separate optimal guard is required in $G_A^{prev}$ for guarding $a_{\sigma(i)} = t1(j)$.
Thus, two distinct optimal guards are required to guard the three vertices $z_j, z_{j'} \in Z_A^{prev}$ 
and $a_{\sigma(i)} = t1(j)$, whenever both $z_j$ and $z_{j'}$ belong to the same set $P_i(L^k)$.
\end{proof}

\vspace*{-1.01em}
\begin{figure}[H]
  \centerline{\includegraphics[width=0.84\textwidth]{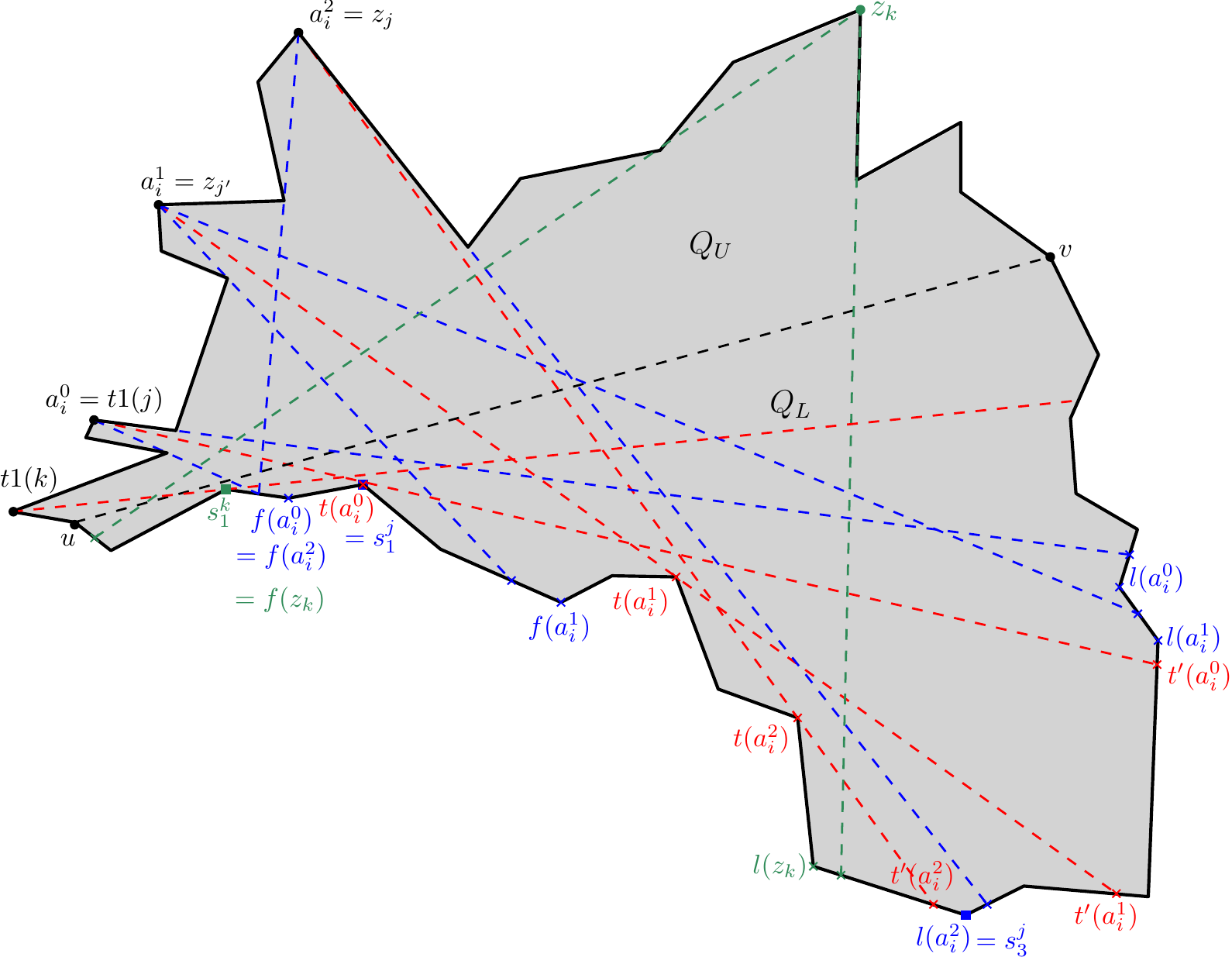}}
  \caption{ Since both $z_j, z_{j'} \in A^k$ belong to $P_i(L^k)$ in the partition $P(L^k)$ of $A^k$, 
            and $k < j < j'$, $z_{j} \in Z_A^{prev}(k)$ and $z_{j'} \in Z_A^{prev}(j)$. } 
  \label{2inL^k_i}
\end{figure}

\vspace*{-0.88em}
\begin{lemma} \label{ZC_prev}
$|G_C^{prev}| \geq |Z_C^{prev}|$.
\end{lemma}

\begin{proof}
Assume that $\mathcal{VVP}^{-}(z_j) \cap \mathcal{VVP}^{-}(z_{j'}) = \emptyset$ 
for every pair $z_j, z_{j'} \in Z_C^{prev}$.
Then clearly a distinct optimal guard is required in $G_C^{prev}$ for guarding each primary vertex in $Z_C^{prev}$,
and therefore $|G_C^{prev}| \geq |Z_C^{prev}|$. 
However, as per Lemma \ref{R^k_i_pairwise}, there may exist some pair $z_j, z_{j'} \in Z_C^{prev}$ 
for which $\mathcal{VVP}^{-}(z_j) \cap \mathcal{VVP}^{-}(z_{j'}) = \emptyset$ does not hold,
and a single optimal guard in $g \in G_C^{prev}$ is sufficient to see both $z_j$ and $z_{j'}$.
Then, both $z_j$ and $z_{j'}$ belong to the same set $P_i(R^k)$ in the partition $P(R^k)$ of $C^k$ 
corresponding to some $z_k \in Z_C$. 
We prove below that, in such a situation, an additional guard for $t1(j) \in Y_C^{prev}$ is required 
in $G_C^{prev}$ to compensate for the single optimal that sees two primary vertices in $Z_C^{prev}$.
Recall that two primary vertices in $Z_C^{prev}$ can be overlapping 
only if they both belong to the same set in $P(R^k)$.
Moreover, since we know from Lemma \ref{ZR^k_i_bound} that at most only two primary vertices can be chosen from any $P_i(R^k)$,
such compensation can be applied to every $P_i(R^k)$ for which $|P_i(R^k) \cap Z_C^{prev}| = 2$, 
and hence, $|G_C^{prev}| \geq |Z_C^{prev}|$ still holds. \\

\vspace*{-0.33em}
Let $c_{\sigma(i)} \in C^k$ be the leading vertex belonging to $P_i(R^k)$. 
Let us assume without loss of generality that $j < j'$, i.e. $z_j$ was chosen as a primary vertex earlier than $z_{j'}$ by Algorithm \ref{vg_pcode_sp1}.
Then, observe that $t1(j) = c_{\sigma(i)}$, and so $c_{\sigma(i)} \in Y_C^{prev}$.
As $z_{j'}$ is unmarked even after placement of $s^j_2 = prev(t'(c_{\sigma(i)}))$ and $s^j_3 = l(z_j)$,
we have $t'(c_{\sigma(i)}) \prec f(z_{j'})$ and $l(z_j) \prec t(z_{j'})$, 
and they can both be possible only if $t'(c_{\sigma(i)}) \prec f(z_{j'}) \prec t'(z_{j'}) \prec t'(z_j) \prec t(z_j) \prec l(z_j) \prec t(z_{j'})$.
Thus, the common optimal guard $g \in G_C^{prev}$ that sees both $z_j$ and $z_{j'}$ must lie on $bd_{cc}(f(z_{j'}),t'(z_{j'}))$.
But in that case, $t'(c_{\sigma(i)}) \prec g \prec t(c_{\sigma(i)})$, 
which means that $g$ does not see $c_{\sigma(i)} = t1(j)$,
and so a separate optimal guard is required in $G_C^{prev}$ for guarding $c_{\sigma(i)} = t1(j)$.
Thus, two distinct optimal guards are required to guard the three vertices $z_j, z_{j'} \in Z_C^{prev}$ 
and $c_{\sigma(i)} = t1(j)$, whenever both $z_j$ and $z_{j'}$ belong to the same set $P_i(R^k)$.
\end{proof}

Consider the two sets $Z_A \cup Y_A \cup Z_A^{prev}$ and $Z_C \cup Y_C \cup Z_C^{prev}$ respectively.
Let us denote by $G_A$ the minimal subset of $G_{opt}^L$ that see all vertices of $Z_A \cup Y_A \cup Z_A^{prev}$,
and similarly, let us denote by $G_C$ the minimal subset of $G_{opt}^L$ that see all vertices of $Z_A \cup Y_C \cup Z_C^{prev}$.
In order to obtain a lower bound on the number of optimal guards, we establish the following two lemmas.

\begin{lemma} \label{GA_bound}
$|G_A| \geq |Z_A \cup Z_A^{prev}|$.
\end{lemma}

\begin{proof}
Observe that $G_A \supseteq G_A^{prev}$, since $G_A$ is also required to guard the vertices belonging 
to $(Z_A \setminus Z_A^{prev}) \cup (Y_A \setminus Y_A^{prev})$ in addition to those guarded by $G_A^{prev}$.
We claim that, for every $z_k \in Z_A \setminus Z_A^{prev}$, 
there exists an optimal guard in $G_A$ for guarding $z_k$ or $t1(k)$
which is distinct from all the optimal guards already counted in $G_A^{prev}$. 
This is enough to prove the lemma, as it implies in conjunction with Lemma \ref{ZA_prev} that 
$|G_A| \geq |G_A^{prev}| + |Z_A \setminus Z_A^{prev}| 
       \geq |Z_A^{prev}| + |Z_A \setminus Z_A^{prev}| = |Z_A \cup Z_A^{prev}|$. 
The above claim is proven below. \\ 

\vspace*{-0.33em}
Let $g_k \in G_A$ be an optimal guard that sees $z_k$.
If $g_k$ does not coincide with any of the optimal guards already counted in $G_A^{prev}$, 
then clearly our claim holds. 
Otherwise, let us consider the situation where $g_k$ coincides with $g_j \in G_A^{prev}$ that guards 
a primary vertex $z_j \in (Z_A^{prev} \cap A^k)$ or some vertex $y \in (Y_A^{prev} \cap A^k)$, or may be both.
Let us first consider the subcase where $g_j$ sees some primary vertex $z_j \in (Z_A^{prev} \cap A^k)$.
If $z_j \notin P_1(L^k)$, then clearly $g_j$ does not see $t1(k) \in P_1(L^k)$. 
Thus, a separate optimal guard is required in $G_A$ for guarding $t1(k) \in Y_A \setminus Y_A^{prev}$, and our claim still holds.
So, let us consider the other case where $z_j \in P_1(L^k)$ and $g_j \in G_A^{prev}$ also sees $t1(k)$.
Observe that, as $g_j$ sees $t1(k)$, we have either $g_j \prec t(t1(k))$ or $t'(t1(k)) \prec g_j$.
However, since $g_k$ and $g_j$ coincide, we must have $g_j \prec l(z_k) \prec t'(t1(k))$, 
which rules out the latter possibility. So, we must have $g_j \prec t(t1(k))$, 
which in turn implies that $f(z_j) \prec g_j \prec t(t1(k)) \prec t(z_j)$.
Note that, if $s^k_1 = t(t1(k))$ is not visible from $z_j$ due to a right pocket, 
then either it means that $t1(k)$ belongs to $B^k$ rather than $A^k$,
or it contradicts the fact that $t1(k)$ is the leading vertex of $P_1(L^k)$.
On the other hand, if $s^k_1 = t(t1(k))$ is not visible from $z_j$ due to a left pocket,
then again it means that $z_j$ belongs to $B^k$ or $C^k$ rather than $A^k$.
So, $z_j$ must be visible from $t(t1(k))$, and is thus marked by the placement of $s^k_1 = t(t1(k))$, 
which contradicts the assumption that $z_j \in Z_A^{prev}$. \\

\vspace*{-0.33em}
Let us now consider the other subcase where $g_j$ sees only a vertex $y \in (Y_A^{prev} \cap A^k)$, 
but $g_j$ sees no primary vertex $z_j \in (Z_A^{prev} \cap A^k)$. 
Once again, if $y \notin P_1(L^k)$, then clearly $g_j$ does not see $t1(k) \in P_1(L^k)$.
So, let us consider the other case where $y \in P_1(L^k)$ and $g_j \in G_A^{prev}$ also sees $t1(k)$.
Observe that, as $g_j$ sees $t1(k)$, we have either $g_j \prec t(t1(k))$ or $t'(t1(k)) \prec g_j$.
However, since $g_k$ and $g_j$ coincide, we must have $g_j \prec l(z_k) \prec t'(t1(k))$, 
which rules out the latter possibility. So, we must have $g_j \prec t(t1(k))$, 
which in turn implies that $f(y) \prec g_j \prec t(t1(k)) \prec t(y)$.
Note that, if $s^k_1 = t(t1(k))$ is not visible from $y$ due to a right pocket, 
then either it means that $t1(k)$ belongs to $B^k$ rather than $A^k$,
or it contradicts the fact that $t1(k)$ is the leading vertex of $P_1(L^k)$.
On the other hand, if $s^k_1 = t(t1(k))$ is not visible from $y$ due to a left pocket,
then again it means that $y$ belongs to $B^k$ or $C^k$ rather than $A^k$.
So, $y$ must be visible from $t(t1(k))$, and is thus marked by the placement of $s^k_1 = t(t1(k))$, 
which contradicts the assumption that $y \in Y_A^{prev}$. 
Hence, it is established that for every $z_k \in Z_A$, there exists an optimal guard in $G_A$ 
for guarding $z_k$ or $t1(k)$, which is distinct from all the optimal guards already counted in $G_A^{prev}$, 
and this completes our proof. 
\end{proof}

\begin{lemma} \label{GC_bound}
$|G_C| \geq |Z_C \cup Z_C^{prev}|$.
\end{lemma}

\begin{proof}
Observe that $G_C \supseteq G_C^{prev}$, since $G_C$ is also required to guard the vertices belonging 
to $(Z_C \setminus Z_C^{prev}) \cup (Y_C \setminus Y_C^{prev})$ in addition to those guarded by $G_C^{prev}$.
We claim that, for every $z_k \in Z_C \setminus Z_C^{prev}$, 
there exists an optimal guard in $G_C$ for guarding $z_k$ or $t'1(k)$
which is distinct from all the optimal guards already counted in $G_C^{prev}$. 
This is enough to prove the lemma, as it implies in conjunction with Lemma \ref{ZC_prev} that 
$|G_C| \geq |G_C^{prev}| + |Z_C \setminus Z_C^{prev}| 
       \geq |Z_C^{prev}| + |Z_C \setminus Z_C^{prev}| = |Z_C \cup Z_C^{prev}|$. 
The above claim is proven below. \\ 

\vspace*{-0.33em}
Let $g_k \in G_C$ be an optimal guard that sees $z_k$.
If $g_k$ does not coincide with any of the optimal guards already counted in $G_C^{prev}$, 
then clearly our claim holds. 
Otherwise, let us consider the situation where $g_k$ coincides with $g_j \in G_C^{prev}$ that guards 
a primary vertex $z_j \in (Z_C^{prev} \cap C^k)$ or some vertex $y \in (Y_C^{prev} \cap C^k)$, or may be both.
Let us first consider the subcase where $g_j$ sees some primary vertex $z_j \in (Z_C^{prev} \cap C^k)$.
If $z_j \notin P_1(R^k)$, then clearly $g_j$ does not see $t'1(k) \in P_1(R^k)$. 
Thus, a separate optimal guard is required in $G_C$ for guarding $t'1(k) \in Y_C \setminus Y_C^{prev}$, 
and our claim still holds.
So, let us consider the other case where $z_j \in P_1(R^k)$ and $g_j \in G_C^{prev}$ also sees $t'1(k)$.
Observe that, as $g_j$ sees $t'1(k)$, we have either $g_j \prec t'(t'1(k))$ or $t(t'1(k)) \prec g_j$.
However, since $g_k$ and $g_j$ coincide, we must have $g_j \prec l(z_k) \prec t(t'1(k))$, 
which rules out the latter possibility. So, we must have $g_j \prec t'(t'1(k))$, 
which in turn implies that $f(z_j) \prec g_j \prec t'(t'1(k)) \prec t'(z_j)$.
Note that, if $s^k_2 = prev(t'(t'1(k)))$ is not visible from $z_j$ due to a right pocket, 
then either it means that $t'1(k)$ belongs to $B^k$ rather than $C^k$,
or it contradicts the fact that $t'1(k)$ is the leading vertex of $P_1(R^k)$.
On the other hand, if $s^k_2 = prev(t'(t'1(k)))$ is not visible from $z_j$ due to a left pocket,
then again it means that $z_j$ belongs to $B^k$ or $A^k$ rather than $C^k$.
So, $z_j$ must be visible from $prev(t'(t'1(k)))$, and is thus marked by the placement 
of $s^k_2 = prev(t'(t'1(k)))$, which contradicts the assumption that $z_j \in Z_C^{prev}$. \\

\vspace*{-0.33em}
Let us now consider the other subcase where $g_j$ sees only a vertex $y \in (Y_C^{prev} \cap C^k)$, 
but $g_j$ sees no primary vertex $z_j \in (Z_C^{prev} \cap C^k)$. 
Once again, if $y \notin P_1(R^k)$, then clearly $g_j$ does not see $t'1(k) \in P_1(R^k)$.
So, let us consider the other case where $y \in P_1(R^k)$ and $g_j \in G_C^{prev}$ also sees $t'1(k)$.
Observe that, as $g_j$ sees $t'1(k)$, we have either $g_j \prec t'(t'1(k))$ or $t(t'1(k)) \prec g_j$.
However, since $g_k$ and $g_j$ coincide, we must have $g_j \prec l(z_k) \prec t(t'1(k))$, 
which rules out the latter possibility. So, we must have $g_j \prec t'(t'1(k))$, 
which in turn implies that $f(y) \prec g_j \prec t'(t'1(k)) \prec t'(y)$.
Note that, if $s^k_2 = prev(t'(t'1(k)))$ is not visible from $y$ due to a right pocket, 
then either it means that $t'1(k)$ belongs to $B^k$ rather than $C^k$,
or it contradicts the fact that $t'1(k)$ is the leading vertex of $P_1(R^k)$.
On the other hand, if $s^k_2 = prev(t'(t'1(k)))$ is not visible from $z_j$ due to a left pocket,
then again it means that $z_j$ belongs to $B^k$ or $A^k$ rather than $C^k$.
So, $z_j$ must be visible from $prev(t'(t'1(k)))$, and is thus marked by the placement 
of $s^k_2 = prev(t'(t'1(k)))$, which contradicts the assumption that $y \in Y_A^{prev}$. 
Hence, it is established that for every $z_k \in Z_A$, there exists an optimal guard in $G_A$ 
for guarding $z_k$ or $t1(k)$, which is distinct from all the optimal guards already counted in $G_A^{prev}$, 
and this completes our proof.
\end{proof}

\begin{lemma} \label{Gopt_lb}
$ |G_{opt}^L| \geq |Z| / 2 $.
\end{lemma}

\begin{proof}
Since $G_A$ and $G_C$ are both subsets of $G_{opt}^L \setminus G_B$, we know that 
$|G_{opt}^L \setminus G_B| \geq |G_A|$ and $|G_{opt}^L \setminus G_B| \geq |G_C|$.
By Lemma \ref{opt_both}, we know that $G_A$ and $G_C$ need not necessarily be disjoint. Thus, 
$ |G_{opt}^L \setminus G_B| \geq \max(|G_A|,|G_C|) \geq (|G_A| + |G_C|) / 2 $.
By using Lemmas \ref{GA_bound} and \ref{GC_bound}, we have
$ (|G_A| + |G_C|) / 2 
\geq (|Z_A \cup Z_A^{prev}| + |Z_C \cup Z_C^{prev}|) / 2 
\geq (|Z_A \cup Z_A^{prev} \cup Z_C \cup Z_C^{prev}|) / 2
= |Z \setminus Z_B|/2 $.
Therefore, by using Lemma \ref{GB_bound},
$ |G_{opt}^L| = |G_{opt}^L \setminus G_B| + |G_B| 
              \geq (|G_A| + |G_C|) / 2 + |Z_B|               \geq |Z \setminus Z_B|/2 + |Z_B| 
              \geq (|Z \setminus Z_B| + |Z_B|) / 2 
              = |Z| / 2 $.
\end{proof}

\begin{theorem} \label{Z_ub}
Let $Z$ be the set of primary vertices chosen by Algorithm \ref{vg_pcode_sp1}. 
Then, the set $S$ of vertex guards computed by Algorithm \ref{vg_pcode_sp1} satisfies
$ |S| \leq 6 \cdot |G_{opt}^L| $.
\end{theorem}

\begin{proof}
Since Algorithm \ref{vg_pcode_sp1} places at most three vertex guards $s^k_1$, $s^k_2$ and $s^k_3$ corresponding to each primary vertex $z_k \in Z$, 
we know that $|S| \leq 3\cdot |Z|$.
Also, from Lemma \ref{Gopt_lb}, we know that $|Z| \leq 2\cdot |G_{opt}^L|$.
Therefore, $|S| \leq 3\cdot |Z| \leq 6\cdot |G_{opt}^L|$. 
\end{proof}

\begin{algorithm}[H]
\caption{An $\mathcal{O}(n^4)$-algorithm for computing a guard set $S$ for all vertices of $Q_U$} 
\label{vg_pcode_sp1}
\begin{algorithmic}[1]

\State Initialize $k \leftarrow 0$ and $S \leftarrow \emptyset$  \label{vg_pcode_sp1:1}
\State Initialize all vertices of $Q_U$ as unmarked  \label{vg_pcode_sp1:2}

\While { there exists an unmarked vertex in $Q_U$} 
\label{vg_pcode_sp1:3}
\State Set $k \leftarrow k + 1$  \Comment{ Variable $k$ keeps count of the current iteration}  \label{vg_pcode_sp1:4}

\State $z_k \leftarrow$ the first unmarked vertex of $bd_c(u,v)$ in clockwise order  \label{vg_pcode_sp1:5}
\State $q \leftarrow z_k$ \label{vg_pcode_sp1:6}
\While{ $q \neq v$ }  \label{vg_pcode_sp1:7}
\State $q \leftarrow$ vertex next to $q$ in clockwise order on $bd_c(u,v)$   \label{vg_pcode_sp1:8}
\If{ $q$ is unmarked and $l'(q) \prec l'(z_k)$ }  \label{vg_pcode_sp1:9}
\State $z_k \leftarrow q$
\Comment{Update $z_k$ to $q$ whenever $q$ is unmarked and $l'(q) \prec l'(z_k)$}
\label{vg_pcode_sp1:10}
\EndIf  \label{vg_pcode_sp1:11}
\EndWhile  \Comment{ Variable $z_k$ is now the primary vertex for the current iteration} \label{vg_pcode_sp1:12}

\State Compute the ordered set $\mathcal{OVV}^{-}(z_k) = \{x^k_1,x^k_2,\ldots,x^k_{m(k)}\}$ 
\label{vg_pcode_sp1:13} 
\State Partition $\mathcal{OVV}^{-}(z_k)$ into the sets $A^k$, $B^k$ and $C^k$
\label{vg_pcode_sp1:14} 

\If{ $\mathcal{CI}(A^k \cup B^k \cup C^k) \neq \emptyset$ }
\label{vg_pcode_sp1:15} 
\State $s_3^k \leftarrow$ any vertex belonging to $\mathcal{CI}(A^k \cup B^k \cup C^k)$ \label{vg_pcode_sp1:16}
\State $S^k \leftarrow \{ s_3^k \}$ 
\Comment{ See Figure \ref{case1} }
\label{vg_pcode_sp1:17}

\Else   \label{vg_pcode_sp1:18} 
\State $s_3^k \leftarrow l(z_k)$
\Comment{ See Figures \ref{case2a}, \ref{case2b}, \ref{case2c} and \ref{case2d} }
\label{vg_pcode_sp1:19}

\If{ $\mathcal{CI}(A^k) \neq \emptyset$ }
\label{vg_pcode_sp1:20} 
\State $s_1^k \leftarrow$ any vertex belonging to $\mathcal{CI}(A^k)$
\Comment{ See Figures \ref{case2a} and \ref{case2c} }
\label{vg_pcode_sp1:21}
\Else \label{vg_pcode_sp1:22}
\State $q' \leftarrow f(z_k)$ \label{vg_pcode_sp1:23}
\While{ $q' \neq l(z_k)$ }  \label{vg_pcode_sp1:24}
\State $q' \leftarrow$ vertex next to $q'$ in counter-clockwise order on $\mathcal{VVP}^{-}(z_k)$ 
\label{vg_pcode_sp1:25}
\If{ $q' = t(x_i^k)$ for some $x_i^k \in A^k$  }
\label{vg_pcode_sp1:26}
\State $s_1^k \leftarrow q'$  
\Comment{ See Figures \ref{case2b} and \ref{case2d} }
\label{vg_pcode_sp1:27}
\State \textbf{break}  \label{vg_pcode_sp1:28}
\EndIf  \label{vg_pcode_sp1:29}
\EndWhile  \label{vg_pcode_sp1:30}
\EndIf \label{vg_pcode_sp1:31}

\If{ $\mathcal{CI}(C^k) \neq \emptyset$ }
\label{vg_pcode_sp1:32} 
\State $s_2^k \leftarrow$ any vertex belonging to $\mathcal{CI}(C^k)$
\Comment{ See Figures \ref{case2a} and \ref{case2b} }
\label{vg_pcode_sp1:33}
\Else \label{vg_pcode_sp1:34}
\State $q' \leftarrow f(z_k)$ \label{vg_pcode_sp1:35}
\While{ $q' \neq l(z_k)$ }  \label{vg_pcode_sp1:36}
\State $q' \leftarrow$ vertex next to $q'$ in counter-clockwise order 
on $\mathcal{VVP}^{-}(z_k)$ 
\If{ $q'$ immediately precedes $t'(x_j^k)$ for some $x_j^k \in C^k$  }
\label{vg_pcode_sp1:38}
\State $s_2^k \leftarrow q'$
\Comment{ See Figures \ref{case2c} and \ref{case2d} }
\label{vg_pcode_sp1:39}
\State \textbf{break}  \label{vg_pcode_sp1:40}
\EndIf  \label{vg_pcode_sp1:41}
\EndWhile  \label{vg_pcode_sp1:42}

\EndIf  \label{vg_pcode_sp1:43}

\State $S^k \leftarrow \{ s_1^k, s_2^k, s_3^k \}$ \label{vg_pcode_sp1:44} 
\EndIf \label{vg_pcode_sp1:45} 

\State $S \leftarrow S \cup S^k$  \label{vg_pcode_sp1:46} 
\State Mark all vertices of $Q_U$ visible from new guards \label{vg_pcode_sp1:47} 
\EndWhile  \label{vg_pcode_sp1:48}

\State \Return the guard set $S$ \label{vg_pcode_sp1:49}
\end{algorithmic}
\end{algorithm}

\begin{theorem}
Algorithm \ref{vg_pcode_sp1} has a worst-case time complexity of $\mathcal{O}(n^4)$.
\end{theorem}

\subsection{Placement of guards in the general scenario}
\label{upper_n_lower}

\vspace{-0.33em}
Let us consider the general scenario where $G^L_{opt} \subset G_{opt}$ and $G^U_{opt} \neq \emptyset$. 
If Algorithm \ref{vg_pcode_sp1} is executed in this scenario, 
there may exist a subset $Z' \subseteq Z$ of primary vertices that are visible from an optimal guard belonging to $G^U_{opt}$ (see Figure \ref{inside_guard_role}). 
Therefore, it is necessary to choose both inside and outside vertex guards corresponding to each primary vertex $z_k \in Z$, 
so that it is ensured that distinct optimal guards are required for guarding every two primary vertices. 
So, keeping this in mind, let us modify Algorithm \ref{vg_pcode_sp1} so that, 
in addition to the three outside guards $s^k_1$, $s^k_2$ and $s^k_3$ 
it also places at most three inside guards $s^k_4$, $s^k_5$ and $s^k_6$ for every $z_k \in Z$. \\

\vspace*{-0.88em}
\begin{figure}[H]
  \centerline{\includegraphics[width=0.66\textwidth]{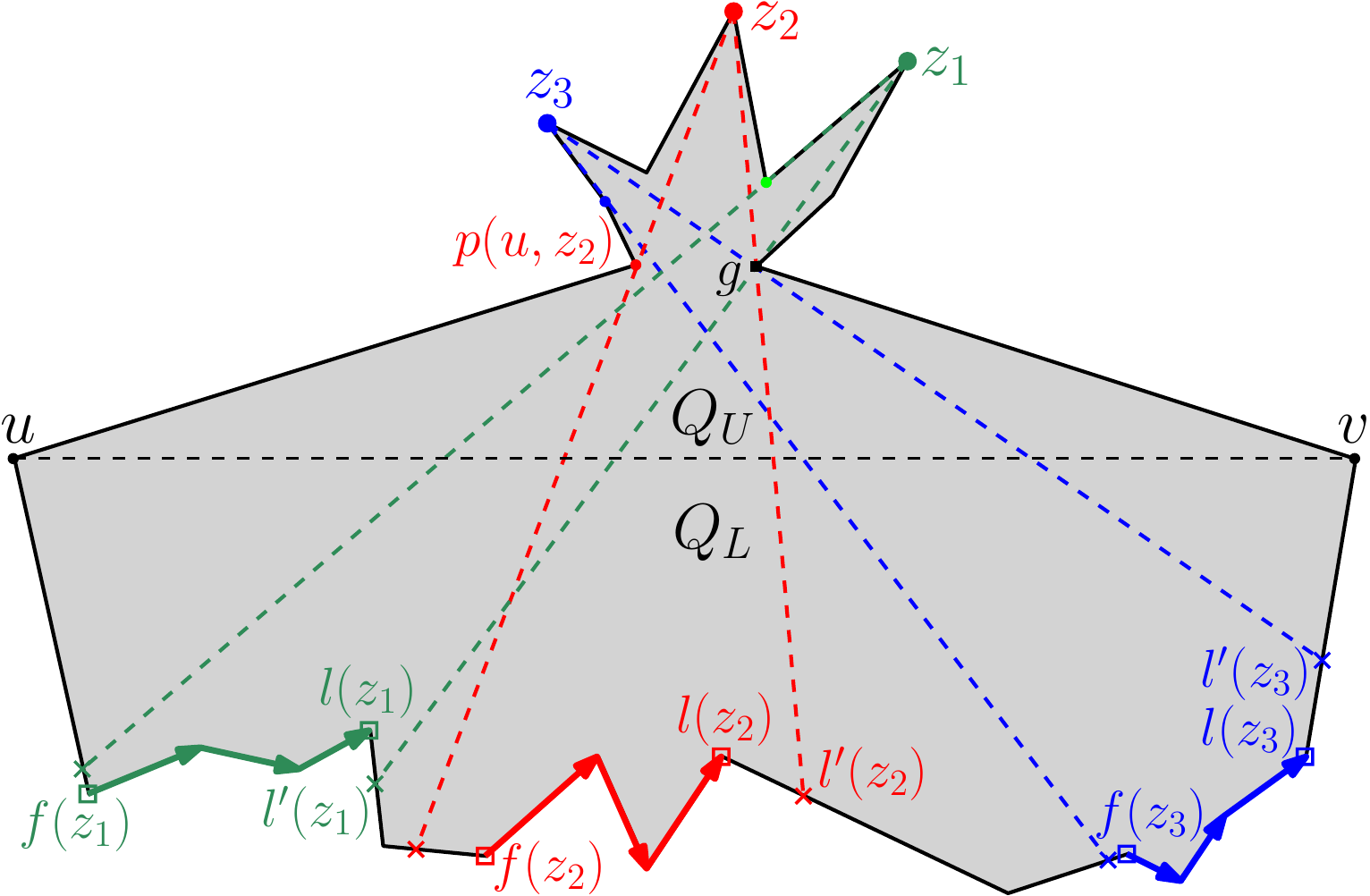}}
  \caption{ A guard $g$ on $bd_c(u,v)$ sees $z_1$, $z_2$ and $z_3$, but three guards are necessary on 
  $bd_{cc}(f(z_1),l(z_1))$, $bd_{cc}(f(z_2),l(z_2))$ and $bd_{cc}(f(z_3),l(z_3))$ respectively to see them. }
  \label{inside_guard_role}
\end{figure}

\vspace{-0.77em}
For any primary vertex $z_k$, let us denote by $\mathcal{OVV}^{+}(z_k)$ the set of unmarked vertices of $Q_U$ whose inward visible vertices overlap with those of $z_k$. 
In other words, 
\vspace{-0.3em}
$$ \mathcal{OVV}^{+}(z_k) = \{ x \in \mathcal{V}(Q_U) : \mbox{ $x$ is unmarked, and } \mathcal{VVP}^{+}(z_k) \cap \mathcal{VVP}^{+}(x) \neq \emptyset \} $$
Note that all vertices of $bd_c(p(u,z_k),z_k)$ may not belong to $\mathcal{OVV}^{+}(z_k)$. 
Further, every vertex of $\mathcal{OVV}^{+}(z_k)$ is at a link distance of 1 from some vertex of $\mathcal{VVP}^{+}(z_k)$ and at a link distance of 2 from $z_k$. 
In the modified algorithm, two inside guards $s_4^k$ and $s_5^k$ are placed at $p(u,z_k)$ and $p(v,z_k)$ respectively, for every primary vertex $z_k \in Z$.
For the placement of an additional inside guard $s^k_6$, 
consider the following cases.
\begin{description}
\item [Case 1:] 
All vertices of $\mathcal{OVV}^{+}(z_k)$ 
are visible from $s^k_4$ or $s^k_5$. 
So, placement of $s_6^k$ is not required.

\vspace{-0.3em}
\item [Case 2:] 
If all vertices of $\mathcal{OVV}^{+}(z_k)$ 
are not visible from $s^k_4$ or $s^k_5$, 
and there exists one common vertex that sees all vertices of $\mathcal{OVV}^{+}(z_k)$, 
then $s_6^k$ is placed on that common vertex.

\vspace{-0.3em}
\item [Case 3:] 
If all vertices of $\mathcal{OVV}^{+}(z_k)$ are 
not visible from $s^k_4$ or $s^k_5$, 
and there does not exist any common vertex that sees all vertices of $\mathcal{OVV}^{+}(z_k)$ (see Figure \ref{3.4_case3}), 
then $s_6^k$ is placed as follows.
The vertex guard $s_6^k$ is placed  at the farthest vertex $w_k$ of $\mathcal{VVP}^{+}(z_k)$ in clockwise order, starting from $p(u,z_k)$, 
so that it can see all vertices of $\mathcal{OVV}^{+}(z_k)$ that are visible from any vertex of $\mathcal{VVP}^{+}(z_k)$ lying on $bd_c(p(u,z_k),w_k)$ (see Figure \ref{3.4_case3_shift}).
\end{description}

Let us discuss these cases in the presence of guards of $G^U_{opt}$.
Assume that the current primary vertex $z_k$ is visible from an optimal guard $g \in G^U_{opt}$. 
So, $z_k \in Z^U$ and $g \in \mathcal{VVP}^{+}(z_k)$.
Thus, if Case 1 true for $z_k$, then 
all visible vertices of $\mathcal{OVV}^{+}(z_k)$ from $g$ are also visible from $s_4^k$ and $s_5^k$.
Similarly, if Case 2 holds for $z_k$, then 
all visible vertices of $\mathcal{OVV}^{+}(z_k)$ from $g$ are also visible from $s_6^k$, $s_4^k$ or $s_5^k$.
However, if Case 3 holds for $z_k$, then there exists a non-empty subset of vertices  $U^k \subset \mathcal{OVV}^{+}(z_k)$ that are visible from $g$ but not visible from $s_6^k$, $s_4^k$ or $s_5^k$. 
If $y_i^k \in U^k$ does not become a primary vertex later, 
then $y_i^k$ is guarded by guards placed due to some other primary vertex.
Moreover, if this happens for every $y_i^k \in U^k$, then no additional inside guard on $\mathcal{VVP}^{+}(z_k)$ is required. 

\vspace*{-0.66em}
\begin{figure}[H]
\begin{minipage}{0.49\textwidth}
  \centerline{\includegraphics[width=1.11\textwidth]{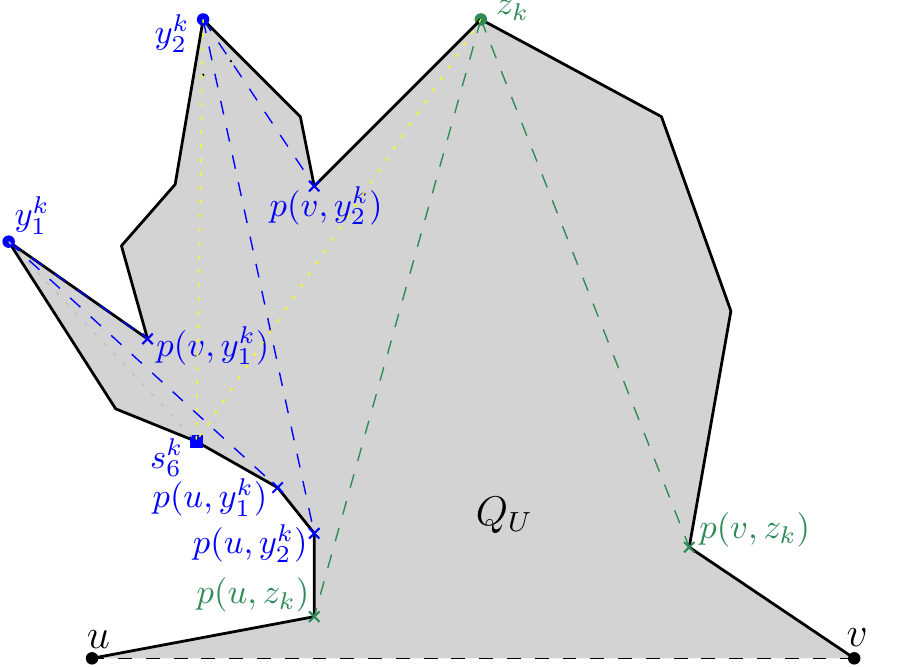}}
  \caption{Vertices $y_1^k$ and $y_2^k$ can be guarded by an inside guard $s_6^k$, but they are not visible from guards at $s_4^k = p(u,z_k)$ and $s_5^k = p(v,z_k)$.}
  \label{3.4_case2}
\end{minipage}
\hspace*{0.01\textwidth}
\begin{minipage}{0.49\textwidth}
  \centerline{\includegraphics[width=1.11\textwidth]{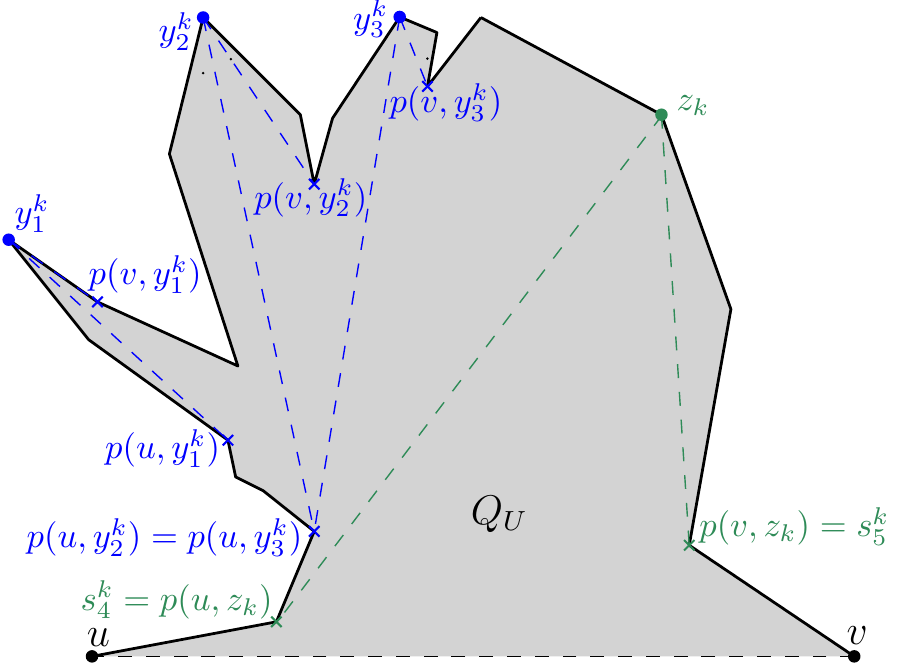}}
  \caption{Vertices $y_1^k$ and $y_2^k$ cannot be guarded by inside guards at $s_4^k = p(u,z_k)$ and $s_5^k = p(v,z_k)$. Moreover, no single additional guard $s_k^6$ can see both of them.}
  \label{3.4_case3}
\end{minipage}
\end{figure}

\vspace*{-0.66em}
Consider the other case where there exists two primary vertices $z_j$ and $z_m$, where $m > j > k$, such that $z_j, z_m \in U^k$.
Consider the other primary vertex $z_m$, where $m > j$, such that $z_m \in \mathcal{OVV}^{+}(z_k)$. 
Since $z_j$ and $z_m$ are later primary vertices, we know that neither $z_j$ nor $z_m$ is visible from $s_6^k$, $s_4^k$ or $s_5^k$.
If $z_j$ is visible from $g\in G^U_{opt}$, then $g \in bd_c(p(u,z_j),p(v,z_j))$ and $g \in \mathcal{VVP}^{+}(z_j)$.
Similarly, if $z_m$ is visible from $g$, then $g \in bd_c(p(u,z_m),p(v,z_m))$ and $g \in \mathcal{VVP}^{+}(z_m)$. 
Now, if $bd_c(p(u,z_j),p(v,z_j))$ and $bd_c(p(u,z_m),p(v,z_m))$ are disjoint parts of $bd_c(p(u,z_k),z_k)$ (see Figure \ref{inner_disjoint}), then $g$ cannot simultaneously belong to $bd_c(p(u,z_j),p(v,z_j))$ and $bd_c(p(u,z_m),p(v,z_m))$, and therefore needs another optimal 
$g'$ lying on $bd_c(p(u,z_k),z_k)$ in order to guard both.
Consider the special case where $p(u,z_j) = p(v,z_m)$, and so $bd_c(p(u,z_j),p(v,z_j))$ and $bd_c(p(u,z_m),p(v,z_m))$ are not totally disjoint (see Figure \ref{inner_shared}). In this case, if $g$ has to simultaneously belong to $bd_c(p(u,z_j),p(v,z_j))$ and $bd_c(p(u,z_j),p(v,z_j))$, then the only possibility for $g$ is to lie on $p(u,z_j) = p(v,z_m)$. But in this case, $z_m$ cannot be a primary vertex later, since it becomes visible from $s_4^j = p(u,z_j)$, and hence marked. Finally, consider the case where $bd_c(p(u,z_m),p(v,z_m))$ is a part of $bd_c(p(u,z_j),z_j)$ (see Figure \ref{inner_nested}). 
Even in this case, there exists no vertex on $bd_c(p(u,z_m),p(v,z_m))$ which can see both $z_m$ and $z_j$. Therefore, $g$ cannot simultaneously see both $z_j$ and $z_m$. We summarize these observations in the following lemma.

\begin{lemma} \label{disjoint}
If three primary vertices $z_k$, $z_j$ and $z_m$, where $k < j < m$, are such that $bd_c(p(u,z_j),p(v,z_j))$ and $bd_c(p(u,z_m),p(v,z_m))$ are both part of $bd_c(p(u,z_k),z_k)$, then an optimal guard $g \in G^U_{opt}$ that sees $z_k$ cannot also see both $z_j$ and $z_m$.
\end{lemma}

\begin{corollary} \label{atmost2}
Any optimal guard $g \in G^U_{opt}$ can see at most two primary vertices. 
\end{corollary}

The above corollary leads to the following theorem.

\begin{theorem} \label{Z^U_bound}  
For Algorithm \ref{vg_pcode_gen}, $|S| \leq 6\cdot|Z^U| \leq 12\cdot|G^U_{opt}|$.
\end{theorem}

\begin{algorithm}[H]
\caption{An $\mathcal{O}(n^4)$-algorithm for computing a guard set $S$ for all vertices of $Q_U$} 
\label{vg_pcode_gen}
\begin{algorithmic}[1]

\State Initialize $k \leftarrow 0$ and $S \leftarrow \emptyset$  \label{vg_pcode_gen:1}
\State Initialize all vertices of $Q_U$ as unmarked  \label{vg_pcode_gen:2}
\State Compute $SPT(u)$ and $SPT(v)$ \label{vg_pcode_gen:3}

\While { there exists an unmarked vertex in $Q_U$} 
\label{vg_pcode_gen:4}
\State Set $k \leftarrow k + 1$  \Comment{ Variable $k$ keeps count of the current iteration}  \label{vg_pcode_gen:5}
\State $z_k \leftarrow$ the first unmarked vertex of $bd_c(u,v)$ in clockwise order  \label{vg_pcode_gen:6}

\State $q \leftarrow z_k$ \label{vg_pcode_gen:7}
\While{ $q \neq v$ }  \label{vg_pcode_gen:8}
\State $q \leftarrow$ vertex next to $q$ in clockwise order on $bd_c(u,v)$   \label{vg_pcode_gen:9}
\If{ $q$ is unmarked and $l'(q) \prec l'(z_k)$ }  \label{vg_pcode_gen:10}
\State $z_k \leftarrow q$
\Comment{Update $z_k$ to $q$ whenever $q$ is unmarked and $l'(q) \prec l'(z_k)$}
\label{vg_pcode_gen:11}
\EndIf  \label{vg_pcode_gen:12}
\EndWhile  \Comment{ Variable $z_k$ is now the primary vertex for the current iteration} \label{vg_pcode_gen:13}

\State $s_4^k \leftarrow p(u,z_k)$, $s_5^k \leftarrow p(v,z_k)$, 
$S^k \leftarrow \{ s_4^k, s_5^k \} $  
\Comment{ See Figures \ref{3.4_case2} and \ref{3.4_case3} }
\label{vg_pcode_gen:14}
\State Mark all vertices of $Q_U$ visible from guards currently in $S^k$ \label{vg_pcode_gen:15}  

\State Compute the ordered set $\mathcal{OVV}^{+}(z_k) = \{y^k_1,y^k_2,\ldots,y^k_{n(k)}\}$ 
\label{vg_pcode_gen:16}
\If { $\mathcal{OVV}^{+}(z_k) \neq \emptyset$ } \label{vg_pcode_gen:17}

\If { $\mathcal{CI}(\mathcal{OVV}^{+}(z_k)) \neq \emptyset$ } \label{vg_pcode_gen:18}
\State $s_6^k \leftarrow$ any vertex of $\mathcal{CI}(\mathcal{OVV}^{+}(z_k))$ 
\Comment{See Figure \ref{3.4_case2} } 
\label{vg_pcode_gen:19}

\Else \Comment{where $\mathcal{CI}(\mathcal{OVV}^{+}(z_k)) = \emptyset$} 
\Comment{See Figure \ref{3.4_case3} } \label{vg_pcode_gen:20}
\State $q \leftarrow p(u,z_k)$, $VSF = \emptyset$ \label{vg_pcode_gen:21}
\While{ ($q \neq z_k$) and every vertex of $VSF$ is visible from $q$ }  \label{vg_pcode_gen:22}
\State $s_6^k \leftarrow q$ \label{vg_pcode_gen:23}
\State $q \leftarrow$ vertex of $\mathcal{VVP}^{+}(z_k)$ next to $q$ in clockwise order on $bd_c(u,v)$   \label{vg_pcode_gen:24}
\State $VSF \leftarrow VSF \cup (\mathcal{VP}(s_6^k) \cap \mathcal{OVV}^{+}(z_k))$  
\Comment{See Figures \ref{3.4_case3_shift}, \ref{inner_disjoint} and \ref{inner_nested} } \label{vg_pcode_sp1:25}
\EndWhile  \label{vg_pcode_gen:26}

\EndIf \label{vg_pcode_gen:27}

\State $S^k \leftarrow S^k \cup \{ s_6^k \} $  \label{vg_pcode_gen:28}
\EndIf \label{vg_pcode_gen:29}
\State Mark all vertices of $Q_U$ visible from guards currently in $S^k$ \label{vg_pcode_gen:30}  

\State Compute the ordered set $\mathcal{OVV}^{-}(z_k) = \{x^k_1,x^k_2,\ldots,x^k_{m(k)}\}$ 
\label{vg_pcode_gen:31} 
\State Partition $\mathcal{OVV}^{-}(z_k)$ into the sets $A^k$, $B^k$ and $C^k$ 
\label{vg_pcode_gen:32} 
\State Compute $s_1^k$, $s_2^k$, and $s_3^k$ as per Algorithm \ref{vg_pcode_sp1}  \label{vg_pcode_gen:33} 
\State $S^k \leftarrow S^k \cup \{ s_1^k, s_2^k, s_3^k \}$ \label{vg_pcode_gen:34}
\State Mark all vertices of $Q_U$ visible from guards currently in $S^k$ \label{vg_pcode_gen:35}  

\State $S \leftarrow S \cup S^k$  \label{vg_pcode_gen:36} 
\EndWhile  \label{vg_pcode_gen:37}

\State \Return the guard set $S$ \label{vg_pcode_gen:38}
\end{algorithmic}
\end{algorithm}

\begin{theorem} \label{gen_bound}
For Algorithm \ref{vg_pcode_gen}, $|S| \leq 6\cdot|Z| \leq 12\cdot|G_{opt}|$.
\end{theorem}

\begin{proof}
From Theorem \ref{Z_ub}, we know that $|Z^L| \leq 2\cdot|G^L_{opt}|$. Similarly, from Theorem \ref{Z^U_bound}, we know that $|Z^U| \leq 2\cdot|G^U_{opt}|$. 
Further, we know that a maximum of 6 vertex guards, viz. $s_1^k$, $s_2^k$, $s_3^k$, $s_4^k$, $s_5^k$ and $s_6^k$, are chosen for each primary vertex $z_k \in Z$, and so $|S| \leq 6\cdot|Z|$. 
Combining all the above inequalities, we obtain:
$ |S| \leq 6\cdot|Z| \leq 6\cdot(|Z^U|+|Z^L|) \leq 12\cdot(|G^U_{opt}|+|G^L_{opt}|) \leq 12\cdot|G_{opt}| $
\end{proof}

\begin{figure}[H]
\begin{minipage}{0.56\textwidth}
  \centerline{\includegraphics[width=1.11\textwidth]{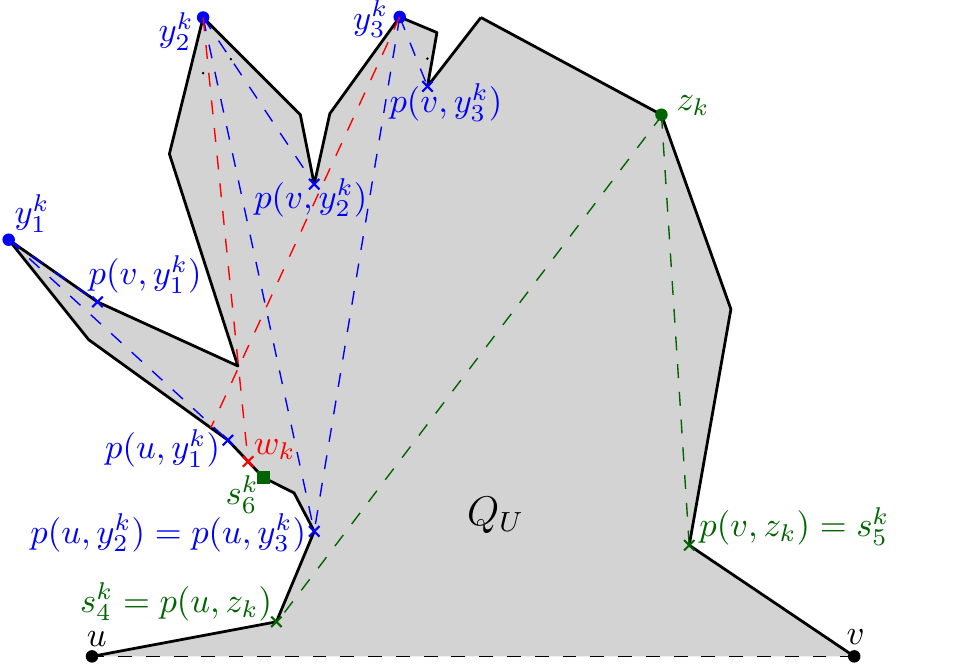}}
  \caption{ The guard $s_6^k$ is placed at the farthest vertex from $p(u,z_k)$ beyond which it cannot move without losing the visibility of $y_2^k \in \mathcal{OVV}^{+}(z_k)$. }
  \label{3.4_case3_shift}
\end{minipage}
\hspace*{0.01\textwidth}
\begin{minipage}{0.43\textwidth}
  \centerline{\includegraphics[width=\textwidth]{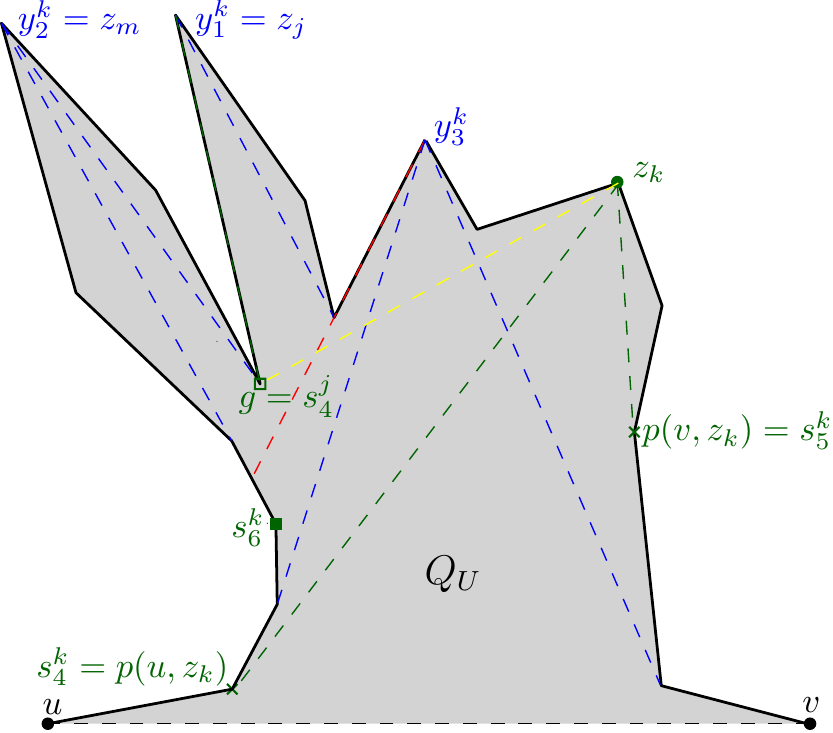}}
  \caption{ The two boundaries $bd_c(p(u,z_m),p(v,z_m))$ and $bd_c(p(u,z_j),p(v,z_j))$, that are both a part of $bd_c(p(u,z_k),z_k)$, share only one vertex $p(u,z_j) = p(v,z_m) = g = s_4^j$. }
  \label{inner_shared}
\end{minipage}
\end{figure}

\begin{figure}[H]
\begin{minipage}{0.49\textwidth}
  \centerline{\includegraphics[width=0.99\textwidth]{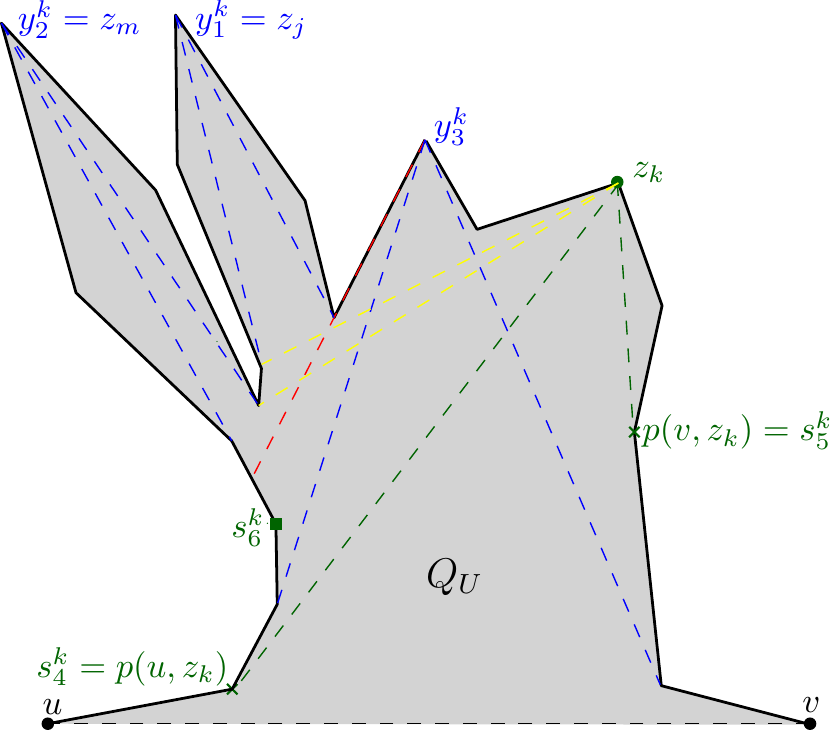}}
  \caption{ Boundaries $bd_c(p(u,z_m),p(v,z_m))$ \& $bd_c(p(u,z_j),p(v,z_j))$ are disjoint parts of $bd_c(p(u,z_k),z_k)$. }
  \label{inner_disjoint}
\end{minipage}
\hspace*{0.01\textwidth}
\begin{minipage}{0.49\textwidth}
  \centerline{\includegraphics[width=0.99\textwidth]{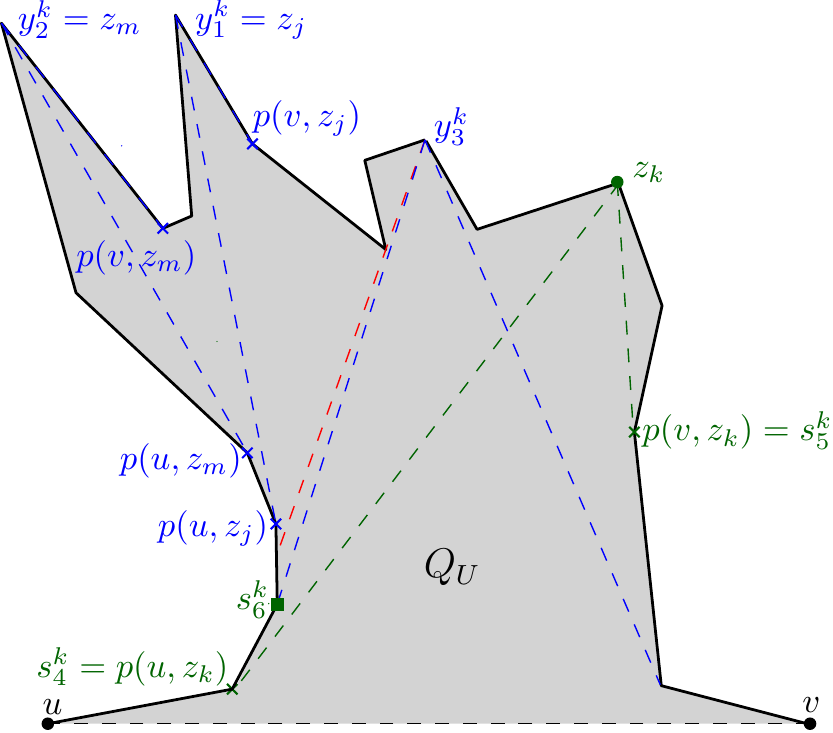}}
  \caption{ Boundary $bd_c(p(u,z_m),p(v,z_m))$ is contained in $bd_c(p(u,z_j),p(v,z_j))$, which in turn is contained in $bd_c(p(u,z_k),z_k)$. }
  \label{inner_nested}
\end{minipage}
\end{figure}

\vspace*{-0.77em}
\begin{theorem} \label{runtime}
The running time for Algorithm \ref{vg_pcode_gen} is $\mathcal{O}(n^4)$.
\end{theorem}

\begin{proof}
While executing Algorithm \ref{vg_pcode_gen} as stated, we need to precompute $SPT(u)$ and $SPT(v)$ respectively. 
Also, for every vertex $z$ belonging to $Q_U$, we need to precompute $\mathcal{VVP}^{+}(z)$, $\mathcal{VVP}^{-}(z)$,
which we can do by constructing the visibility graph of $z$ in $\mathcal{O}(n^2)$ using the algorithm by Ghosh and Mount \cite{GhoshMount_1991}.
Note that $\mathcal{O}(n)$ vertices are chosen in total.
Therefore, in order to get the overall running time for Algorithm \ref{vg_pcode_gen}, let us consider the running times for all the operations performed by Algorithm \ref{vg_pcode_gen} in the outer while-loop 
corresponding to each primary vertex $z_k \in Z$. \\

\vspace{-0.5em}
Since $SPT(u)$ and $SPT(v)$ are precomputed, it takes only $\mathcal{O}(1)$ time to choose the guards $s_4^k = p(u,z_k)$ and $s_5^k = p(v,z_k)$. 
For choosing the guard $s_6^k$, it is required to compute  $\mathcal{OVV}^{+}(z_k)$ and then $\mathcal{CI}(\mathcal{OVV}^{+}(z_k))$. 
The former operation takes $\mathcal{O}(n^2)$ time, since $\mathcal{VVP}^{+}(z)$ is precomputed.
For the latter operation, it is required to compute $|\mathcal{OVV}^{+}(z_k)| = \mathcal{O}(n)$ intersections, and since each intersection takes $\mathcal{O}(n^2)$ time, the total time required for the operation is $\mathcal{O}(n^3)$.
If $\mathcal{CI}(\mathcal{OVV}^{+}(z_k))$ is non-empty, then the choice of $s_6^k$ requires only $\mathcal{O}(1)$ additional time.
Otherwise, the choice of $s_6^k$ requires a linear scan along $bd_c(u,v)$, which takes $\mathcal{O}(n)$ time. 
Since $\mathcal{VVP}^{-}(z)$ is precomputed, it takes only $\mathcal{O}(1)$ time to choose the guard $s_3^k = l(z_k)$. 
However, for choosing the guards $s_1^k$ and $s_2^k$, it is required to compute $\mathcal{OVV}^{-}(z_k)$, which takes $\mathcal{O}(n)$ time, and then partition $\mathcal{OVV}^{-}(z_k)$ into the sets $A^k$, $B^k$ and $C^k$, which takes a further $\mathcal{O}(n^2)$ amount of time. 
Finally, the the choice of $s_1^k$ (and similarly $s_2^k$) requires a linear scan along $bd_{cc}(u,v)$, which takes $\mathcal{O}(n)$ time. \\

\vspace{-0.5em}
From the discussion above, it is clear that all the operations corresponding to a single primary vertex $z_k \in Z$ are completed by Algorithm \ref{vg_pcode_gen} in $\mathcal{O}(n^3)$ time in the worst case. 
Since at most $\mathcal{O}(n)$ primary vertices are chosen, the overall worst case running time of Algorithm \ref{vg_pcode_gen} is $\mathcal{O}(n^4)$.
\end{proof}

\subsection{Guarding all interior points of a polygon}
\label{interior}
In the previous subsection, we have presented an algorithm (see Algorithm \ref{vg_pcode_gen}) that returns a guard set $S$ such that all vertices of $Q_U$ are visible from guards in $S$. 
Recall that the art gallery problem demands that $S$ must see all interior points of $Q_U$ as well. 
However, it may not always be true that the guards in $S$ see all interior points of $Q_U$. 
Consider the polygon shown in Figure \ref{pocket_edge}. 
Assume that Algorithm \ref{vg_pcode_gen} places guards at $p(u,z_k)$ and $p(v,z_k)$, and all vertices of $bd_c(p(u,z_k),p(v,z_k))$ become visible from $p(u,z_k)$ or $p(v,z_k)$. 
However, the triangular region $Q_U \setminus (VP(p(u,z_k))) \cup VP(p(v,z_k))$, bounded by the segments $x_1 x_2$, $x_2 x_3$ and $x_3 x_1$, is not visible from $p(u,z_k)$ or $p(v,z_k)$. Also, one of the sides $x_1 x_2$ of the triangle $x_1 x_2 x_3$ is a part of the polygonal edge $a_1 a_2$. \\

\vspace{-0.5em}
Suppose there exists another guard $g$ lying on $bd_c(p(u,z_k),p(v,z_k))$ (see Figure \ref{pocket_edge}) that sees the part of the triangle $x_1x_2x_3$ containing the side $x_1x_2$,
but does not see the other part containing $x_3$.
In that case, such a vertex $g$ cannot be weakly visible from $uv$, which is a contradiction. 
Hence, for any such region invisible from guards $s_4^k, s_5^k \in S^k$ corresponding to some $z_k \in Z$, henceforth referred to as an \emph{invisible cell}, one of the sides must always be a part of a polygonal edge.
The polygonal edge which contributes as a side to the invisible cell is referred to as its corresponding \emph{partially invisible edge}. \\

\vspace{-0.5em}
Observe that $s_4^k$ and $s_5^k$ can in fact create several invisible cells,
as shown in Figure \ref{pockets}.
Each invisible cell must be wholly contained within the intersection region (which is a triangle) of a left pocket and a right pocket. For example, in Figure \ref{pocket_edge}, the invisible cell $x_1 x_2 x_3$ is actually the entire intersection region of the left pocket of $VP(s_4^k)$ 
and the right pocket of $VP(s_5^k)$. 
In general, where $VP(s_4^k)$ has several left pockets and $VP(s_5^k)$ has several right pockets which intersect pairwise to create multiple invisible cells (as shown in Figure \ref{pockets}), 
every such cell can be seen by placing guards on the common vertices between adjacent pairs of cells. 
Further, if $G_{opt}$ is also constrained to guard these invisible cells using only inward guards from $Q_U$, then the number of such additional guards required can be at most twice of $G_{opt}$,
as shown by Bhattacharya et al. \cite{AGWVP}. 
However, in the absence of any constraint on placing guards, $G_{opt}$ may place an outside guard in $Q_L$ that sees several such invisible cells. 
So, it is natural to explore the possibility of being able to guard all such invisible cells by using additional guards from $Q_L$, in combination with guards from $Q_U$. \\

\vspace{-0.5em}
We present a modified algorithm that ensures that all partially invisible edges are guarded completely, and therefore the entire $bd_c(u,v)$ is guarded.
For every pair of visible vertices in $Q$, extend the visible segment connecting them till they intersect the boundary of $Q_U$. 
These intersection points partition the boundary of $Q_U$ into distinct intervals called \emph{minimal visible intervals}. 
We have the following lemmas.

\begin{lemma} \label{mvi}
Every minimal visible interval on the boundary of $Q_U$ is either entirely visible from a vertex or totally not visible from that vertex.
\end{lemma}

\begin{proof}
 Let $ab$ be a minimal visible interval. 
 If $ab$ is partially visible from a vertex $g$,
 then there must exist another vertex $g'$ such that the extension of the segment $g g'$ intersects $ab$ at some point, which contradicts the fact that $ab$ is a minimal visible interval.
\end{proof}

\begin{corollary} \label{triangle}
If a minimal visible interval $ab$ is entirely visible from a vertex $g$ of $Q$, then the entire triangle $gab$ lies totally inside $Q$.
\end{corollary}

The modified Algorithm \ref{vg_pcode_pseudo} first computes all minimal visible intervals and chooses one internal representative point from each minimal visible interval on the boundary of $Q_U$. 
These representative points are referred to as \emph{pseudo-vertices}. 
Alongside the original polygonal vertices of $Q$, all pseudo-vertices are introduced on the boundary of $Q_U$, and the modified polygon is denoted by $Q'$. Note that the endpoints of minimal visible intervals are not introduced in $Q'$.
In Algorithm \ref{vg_pcode_pseudo}, the pseudo-vertices of $Q'$ are treated in almost the same manner as the original vertices. 
We compute $\mathcal{VVP}^{+}(z)$ and $\mathcal{VVP}^{-}(z)$ for all vertices of $Q'$, irrespective of whether it is a pseudo-vertex, 
and some of the psuedo-vertices may even be chosen as primary vertices. 
However, while computing $\mathcal{VVP}^{+}(z)$ for any vertex $z$ of $Q'$,
no pseudo-vertices are included in $\mathcal{VVP}^{+}(z)$, since they cannot be vertex guards in any case.

\begin{algorithm}[H]
\caption{An $\mathcal{O}(n^5)$-algorithm for computing a guard set $S$ for all vertices of $Q_U$} 
\label{vg_pcode_pseudo}
\begin{algorithmic}[1]

\State Initialize $k \leftarrow 0$ and $S \leftarrow \emptyset$  \label{vg_pcode_pseudo:1}
\State Compute $SPT(u)$ and $SPT(v)$ \label{vg_pcode_pseudo:2}
\State Compute all minimal visible intervals on boundary of $Q_U$ \label{vg_pcode_pseudo:3}
\State Introduce representative points from each minimal visible interval as pseudo-vertices \label{vg_pcode_pseudo:4} 
\State Initialize all vertices and pseudo-vertices of $Q_U$ as unmarked  \label{vg_pcode_pseudo:5}

\While { there exists an unmarked vertex or pseudo-vertex on $Q_U$} 
\label{vg_pcode_pseudo:6}
\State Set $k \leftarrow k + 1$  
\label{vg_pcode_pseudo:7}
\State Choose the current primary vertex $z_k$ as per Algorithm \ref{vg_pcode_gen} \label{vg_pcode_pseudo:8}

\State $s_4^k \leftarrow p(u,z_k)$, $s_5^k \leftarrow p(v,z_k)$, 
$S^k \leftarrow \{ s_4^k, s_5^k \} $  \label{vg_pcode_pseudo:9}
\State Choose the inside guard $s_6^k$ (if required) as per Algorithm \ref{vg_pcode_gen} \label{vg_pcode_pseudo:10}
\State $S^k \leftarrow S^k \cup \{ s_6^k \} $  \label{vg_pcode_pseudo:11}
\State Compute $s_1^k$, $s_2^k$, and $s_3^k$ as per Algorithm \ref{vg_pcode_sp1}  \label{vg_pcode_pseudo:12} 
\State $S^k \leftarrow S^k \cup \{ s_1^k, s_2^k, s_3^k \}$ \label{vg_pcode_pseudo:13} 
\State Mark all vertices of $Q_U$ visible from guards currently in $S^k$ \label{vg_pcode_pseudo:14}  
\State $S \leftarrow S \cup S^k$  \label{vg_pcode_pseudo:15} 

\EndWhile  \label{vg_pcode_pseudo:16}
\State \Return the guard set $S$ \label{vg_pcode_pseudo:17}
\end{algorithmic}
\end{algorithm}

\vspace*{-0.3em}
Suppose an invisible cell is created by the guards placed on the parents of some primary vertex. This implies that there exists a pseudo-vertex on the partially invisible edge contained in the invisible cell which has been left unmarked.
So, either this pseudo-vertex is marked due to the guards chosen for some later primary vertex, or else it is eventually chosen as a primary vertex itself.
If a pseudo-vertex on a partially invisible edge is chosen as a primary vertex $z_k$ by Algorithm \ref{vg_pcode_pseudo}, then the entire invisible cell containing the partially invisible edge must be visible from $s_4^k$ and $s_5^k$. Moreover, if the adjacent invisible cell to the left or right shares a parent, then $s_4^k$ or $s_5^k$ also sees both invisible cells. Therefore, if all invisible cells are guarded by guards in $G^U_{opt}$, then the number of pseudo-vertices that are chosen as primary vertices is at most twice of the number of guards in $G^U_{opt}$ \cite{AGWVP}. \\

\vspace*{-0.5em}
Observe that several invisible cells may be seen by a few outside guards belonging to $G^L_{opt}$, unlike $G^U_{opt}$ which can see only two such cells at the most. 
Assume that a pseudo-vertex is chosen as a primary vertex $z_k$.
Every vertex and pseudo-vertex belonging to $B^k$ is marked due to the placement of the guard at $s_3^k$, and therefore no additional guards are introduced for guarding the vertices or pseudo-vertices in $B^k$.
Consider the outside guards introduced due to $A^k$. 
Assume that $A^k$ has pseudo-vertices. 
If $L^k$ does not have pseudo-vertices, then the pseudo-vertices do not create new disjoint right pockets, and therefore guards placed to guard vertices of $L^k$ are enough to guard pseudo-vertices in $A^k$. 
If $L^k$ contains a pseudo-vertex, then a new disjoint right pocket is created, 
and therefore guards placed for other disjoint right pockets cannot see this pseudo-vertex. 
So, an additional outside guard is required, as well as an additional optimal guard in $G^L_{opt}$. 
The same argument holds for $C^k$ and $R^k$. Thus Lemmas \ref{disjoint->nestedA} to \ref{Gopt_lb}, Lemma \ref{disjoint} and 
Corollary \ref{atmost2}, and Theorems \ref{Z_ub} and \ref{Z^U_bound}  hold even after the introduction of pseudo-vertices, and so the overall bound remains $|S| \leq 12\cdot|G_{opt}|$, as in Theorem \ref{gen_bound}. 
We state this result in Theorem \ref{bound_pseudo}.

\begin{theorem} \label{bound_pseudo}
For the guard set $S$ computed by Algorithm \ref{vg_pcode_pseudo}, $|S| \leq 6\cdot|Z| \leq 12\cdot|G_{opt}|$.
\end{theorem}

\begin{theorem} \label{runtime_pseudo}
The running time of Algorithm \ref{vg_pcode_pseudo} is $\mathcal{O}(n^5)$.
\end{theorem}

\begin{proof}
While executing Algorithm \ref{vg_pcode_pseudo} as stated, the precomputation of $SPT(u)$ and $SPT(v)$
requires $\mathcal{O}(n^2)$ time as the number of pseudo-vertices is $\mathcal{O}(n^2)$.
Also, for every vertex (or pseudo-vertex) $z$ belonging to $Q'$, the precomputation of $\mathcal{VVP}^{+}(z)$ and $\mathcal{VVP}^{-}(z)$ can be done by constructing the visibility graph using the output-sensitive algorithm of Ghosh and Mount \cite{GhoshMount_1991}.
Since the total number of vertices (including pseudo-vertices) in $Q'$ is $\mathcal{O}(n^2)$, and the visibility edges are not computed between pseudo-vertices, the size of the visibility graph for $Q'$ is $\mathcal{O}(n^3)$, and thus $\mathcal{O}(n^3)$ time is required to precompute $\mathcal{VVP}^{+}(z)$ and $\mathcal{VVP}^{-}(z)$.
Note that, while each of the original vertices of $Q_U$ may be chosen as primary vertices, only at most one of the pseudo-vertices belonging to a single edge may be chosen as a primary vertex, which means that there can be at most $2n = \mathcal{O}(n)$ primary vertices chosen by Algorithm \ref{vg_pcode_pseudo}. 
Therefore, in order to get the overall running time for Algorithm \ref{vg_pcode_gen}, let us consider the running times for all the operations performed by Algorithm \ref{vg_pcode_gen} in the outer while-loop (see lines ...) corresponding to each primary vertex $z_k \in Z$. \\

\vspace*{-0.5em}
Since $SPT(u)$ and $SPT(v)$ are precomputed, it takes only $\mathcal{O}(1)$ time to choose the guards $s_4^k = p(u,z_k)$ and $s_5^k = p(v,z_k)$. 
For choosing the guard $s_6^k$, it is necessary to compute  $\mathcal{OVV}^{+}(z_k)$ and then $\mathcal{CI}(\mathcal{OVV}^{+}(z_k))$. 
The former operation takes $\mathcal{O}(n^2)$ time, since $\mathcal{VVP}^{+}(z)$ is precomputed.
For the latter operation, $|\mathcal{OVV}^{+}(z_k)| = \mathcal{O}(n^2)$ intersections are computed; since each intersection takes $\mathcal{O}(n^2)$ time, the total time required for the operation is $\mathcal{O}(n^4)$.
If $\mathcal{CI}(\mathcal{OVV}^{+}(z_k))$ is non-empty, then the choice of $s_6^k$ requires only $\mathcal{O}(1)$ additional time.
Otherwise, the choice of $s_6^k$ requires a linear scan along $bd_c(u,v)$, which takes $\mathcal{O}(n^2)$ time. 
Since $\mathcal{VVP}^{-}(z)$ is precomputed, it takes only $\mathcal{O}(1)$ time to choose the guard $s_3^k = l(z_k)$. 
However, for choosing the guards $s_1^k$ and $s_2^k$, it is required to compute $\mathcal{OVV}^{-}(z_k)$, which takes $\mathcal{O}(n)$ time, and then the partition of $\mathcal{OVV}^{-}(z_k)$ into the sets $A^k$, $B^k$ and $C^k$ requires a further $\mathcal{O}(n^3)$ amount of time. 
Finally, the the choice of $s_1^k$ (and similarly $s_2^k$) requires a linear scan along $bd_{cc}(u,v)$, which takes $\mathcal{O}(n^2)$ time. \\

\vspace*{-0.5em}
From the discussion above, it is clear that all the operations corresponding to a single primary vertex $z_k \in Z$ are completed by Algorithm \ref{vg_pcode_gen} in $\mathcal{O}(n^4)$ time in the worst case, 
and since at most $\mathcal{O}(n)$ primary vertices are chosen, the overall worst case running time of Algorithm \ref{vg_pcode_gen} is $\mathcal{O}(n^5)$.
\end{proof}

\vspace*{-1.1em}
\begin{figure}[H]
\begin{minipage}{0.33\textwidth}
  \centerline{\includegraphics[width=\textwidth]{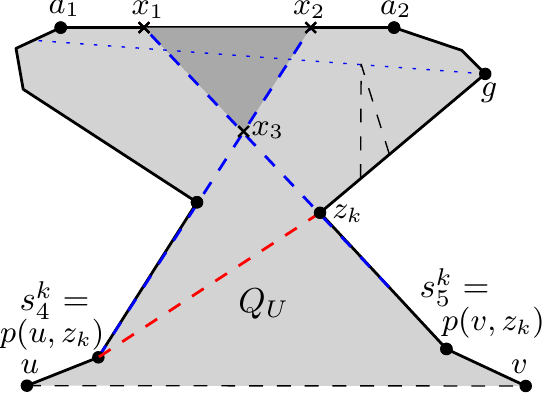}}
  \caption{All vertices are visible from $p(u,z_k)$ or $p(v,z_k)$, but the triangle $x_1 x_2 x_3$ is invisible.}
  \label{pocket_edge}
\end{minipage}
\hspace*{0.01\textwidth}
\begin{minipage}{0.66\textwidth}
  \centerline{\includegraphics[width=\textwidth]{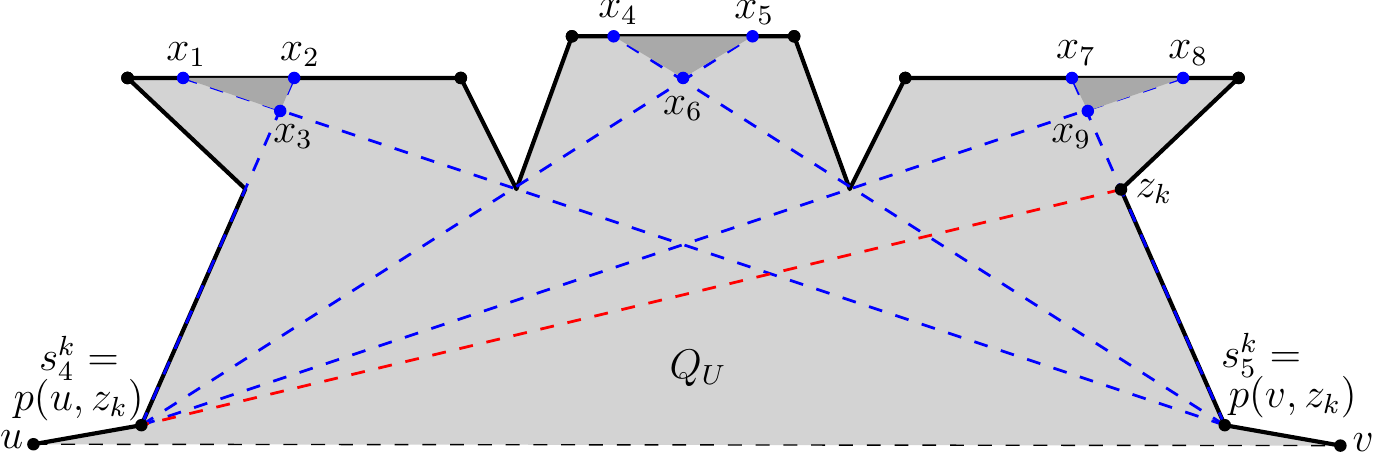}}
  \caption{Multiple invisible cells exist within the polygon that are not visible from the guards placed at $p(u,z_k)$ and $p(v,z_k)$.}
  \label{pockets}
\end{minipage}
\end{figure}

\vspace*{-0.7em}
Algorithm \ref{vg_pcode_pseudo} chooses a guard set $S$ that ensures no partially invisible edge in $Q_U$. 
However, there is no guarantee that $S$ sees the entire interior of $Q_U$, as there may remain residual invisible cells in the interior of $Q_U$ (see Figure \ref{icell}).
Consider a residual invisible cell that is a part of an invisible cell $x_1x_2x_3$, where $x_1x_2$ is contained in a partially invisible edge.
For such a residual invisible cell, there exists a pseudo-vertex on $x_1x_2$ whose parents can see the entire cell $x_1x_2x_3$, as discussed earlier in the context of placing inside guards for guarding entire visibility cell. 
So, placing a guard at an appropriate parent, such as $z_k$ in Figure \ref{icell}, 
guarantees that the residual invisible cell is totally visible.
Since such an additional inside guard on $Q_U$ corresponds to an unique outward guard in $Q_L$, the additional number of inside guards can be at most the number of outside guards. 
This amounts to placing at most (3+3)=6 inside guards and 3 outside guards corresponding to each primary vertex, while the number of primary vertices chosen remains at most $2\cdot|G_{opt}|$. 
We summarize the result in the following theorem. 

\begin{theorem}
Let $Q$ be a polygon of $n$ vertices that is weakly visible from an internal chord $uv$. 
Then, a vertex guard set $S$ can be computed in $\mathcal{O}(n^5)$ time, 
and $|S| \leq 18 \times|G_{opt}|$, 
where $G_{opt}$ is a an optimal vertex guard cover for the entire boundary of $P$.  
\end{theorem}

\vspace*{-0.44em}
\section{Final Results}
\label{final}

\vspace*{-0.22em}
In Section \ref{vertex_algo}, we have presented an approximation algorithm 
for guarding a polygon $Q$ weakly visible from a chord $uv$.
This algorithm chooses primary vertices $z_k$ according to the ordering of their last visible point $l'(z_k)$.
So the algorithm does not depend on a chord of the weak visibility polygon.
Therefore, if this algorithm is executed on the union of overlapping weak visibility polygons $Q$, then it chooses primary vertices in the same way irrespective of chords in $Q$, producing a guard set that sees the entire union $Q$. 
So, if this algorithm is used for every overlapping weak visibility polygon in $P$, then the entire polygon $P$ can be guarded by the union of the guard sets produced for guarding these overlapping weak visibility polygons.
Note that there is no increase in running time for this overall algorithm, since each vertex can appear in at most 2 weak visibility polygons from the hierarchy $W$. \\

\vspace*{-0.5em}
Let $g \in G_{opt}$ be an optimal guard in $V_{i,j}$. 
It can be seen that $g$ is either a guard in $G^U_{opt}$ in the weak visibility polygon $Q$ whose chord $uv$ is the constructed edge between $V_{i,j}$ and its parent (say $V_{i-1,j'}$), 
or a guard in $G^L_{opt}$ for the overlapping weak visibility polygon $Q'$ whose chords are constructed edges that separate $V_{i,j}$ from its children.
Let $g \in Q \cap Q'$. So $g$ is a guard in $G^L_{opt}$ in $Q'$ or a guard in $G^U_{opt}$ in $Q$. 
Consider the case where $g$ belongs to $G^L_{opt}$ in $Q'$. 
Observe that all vertices of $Q_U'$ that are visible from any such $g \in G^L_{opt}$ in $Q'$ are guarded by $s_1^k$, $s_2^k$, $s_3^k$ for all primary vertices $z_k \in Q_U'$. 
Therefore, the approximation bound for Algorithm \ref{vg_pcode_pseudo} run on $Q'$ does not change in this case. 

\vspace*{-0.8em}
\begin{figure}[H]
\begin{minipage}{0.33\textwidth}
\centerline{\includegraphics[width=\textwidth]{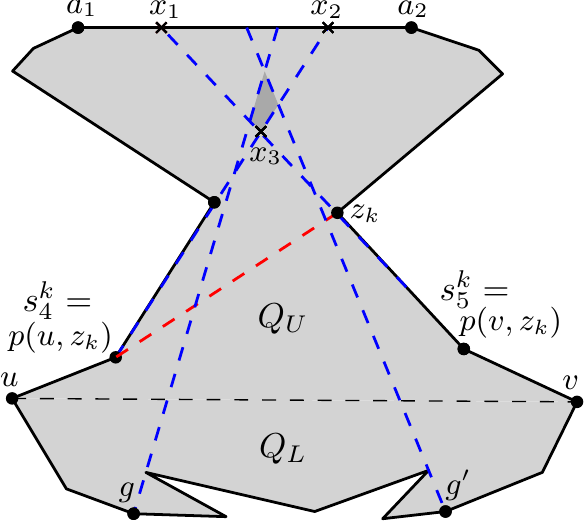}}
\caption{Two outside guards $g$ and $g'$ can create a residual invisible cell that is a part of the triangle $x_1x_2x_3$.}
\label{icell}
\end{minipage}
\hspace*{0.01\textwidth}
\begin{minipage}{0.66\textwidth}
\centerline{\includegraphics[width=\textwidth]{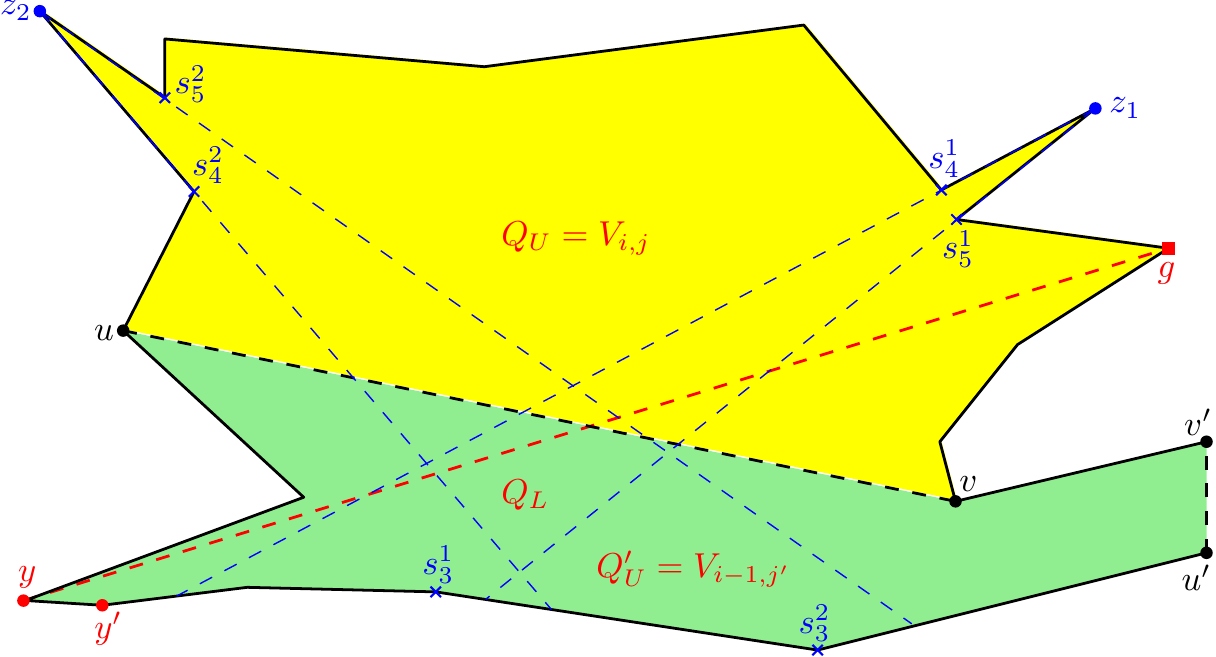}}
\caption{The vertex $y \in Q_L \cap Q_U'$ is visible from $g \in G^U_{opt}$ in $Q_U$, but it is not visible from any of the guards $s_3^1$, $s_4^1$, $s_5^1$, $s_3^2$, $s_4^2$ or $s_5^2$ placed in $Q$. }
\label{upper_opt}
\end{minipage}
\end{figure}

\vspace*{-0.5em}
Consider the other case where $g \in G^U_{opt}$ in $Q$.
Observe that all vertices of $Q_U$ that are visible from any such $g \in G^U_{opt}$ in $Q$ 
are guarded by $s_4^k$, $s_5^k$, $s_6^k$ for all primary vertices $z_k \in Q_U$. 
However, all vertices of $Q_L$ that are visible from any such $g \in G^U_{opt}$ in $Q$ may not necessarily be guarded by $s_1^k$, $s_2^k$, $s_3^k$,$s_4^k$, $s_5^k$, $s_6^k$ for all primary vertices $z_k \in Q_U$ (see Figure \ref{upper_opt}), since the guards chosen by Algorithm \ref{vg_pcode_pseudo} are not meant for guarding vertices in $Q_L$. \\

\vspace*{-0.5em}
Let $y \in Q_L \cap Q_U'$ be a vertex that is visible from $g \in G^U_{opt}$ on $Q$, but not visible from any of the guards placed in $Q$ by Algorithm \ref{vg_pcode_pseudo}.
Since $y$ remains unmarked, $y$ can be chosen as a primary vertex during the execution of \ref{vg_pcode_pseudo} on $Q'$.
Then, the guards placed on the parents see not only $y$, but also see all other such vertices $y'$ visible from $g$ (see Figure \ref{upper_opt}) due to cross-visibility across two adjacent weak visibility polygons in the partition hierarchy $W$. 
So, this amounts to choosing an extra primary vertex in $Q$, and thus the number of primary vertices visible from $g \in Q^U_{opt}$ also increases by at most 1. 
Since there could be at most one extra primary corresponding to every $g \in G^U_{opt}$, for counting purposes, we can attribute these to $Q_U$ as an extra primary vertex. So, $|Z| \leq 2\cdot|G_{opt}|$ is replaced by $|Z| \leq 3\cdot|G_{opt}|$ in our bound for all such $Q_U = V_{i,j}$.
We have the following results.

\begin{theorem}
Let $P$ be a simple polygon of $n$ vertices. 
Then, a vertex guard set $S$ for guarding all vertices of $P$ can be computed in $\mathcal{O}(n^4)$ time, 
and $|S| \leq 18 \times|G_{opt}|$, where $G_{opt}$ is a an optimal vertex guard cover for all vertices of $P$. 
\end{theorem}

\begin{theorem}
Let $P$ be a simple polygon of $n$ vertices. 
Then, a vertex guard set $S$ for guarding the entire boundary of $P$ can be computed in $\mathcal{O}(n^5)$ time, 
and $|S| \leq 18 \times|G_{opt}|$, where $G_{opt}$ is a an optimal vertex guard cover for the entire boundary of $P$.  
\end{theorem}

\begin{theorem}
Let $P$ be a simple polygon of $n$ vertices. 
Then, a vertex guard set $S$ for guarding interior and boundary points of $P$ can be computed in $\mathcal{O}(n^5)$ time, 
and $|S| \leq 27 \times|G_{opt}|$, where $G_{opt}$ is a an optimal vertex guard cover for the entire interior and boundary of $P$.  
\end{theorem}

\section{Algorithms for Polygon Guarding using Edge Guards}
\label{edge_peri}

Let us consider a slightly different version of the art gallery problem, 
where edge guards rather than vertex guards are used for guarding a simple $n$-sided polygon $P$.
The hierarchical partitioning method used is almost exactly the same as 
Algorithm \ref{partition_windows} (presented in Section \ref{partitioning_algo}), 
where we guard the polygon using vertex guards. 
While using edge guards instead, the only modification to Algorithm \ref{partition_windows} is that now we initially choose an arbitrary edge $v_1v_2$ 
instead of a vertex, and compute $V_{1,1} = \mathcal{VP}(v_1v_2)$ 
(compare with Lines \ref{partition_windows:1}-\ref{partition_windows:2} of Algorithm \ref{partition_windows}).
As before, each weak visibility polygon in the hierarchy is used to generate 
one weak visibility polygon in the next level corresponding to each of its constructed edges. \\

\vspace*{-0.55em}
As discussed in Section \ref{vertex_algo}, 
let $Q$ denote a simple polygon that is weakly visible from an internal chord $uv$, i.e. we have $\mathcal{VP}(uv) = Q$, 
and observe that the chord $uv$ similarly splits $Q$ into two sub-polygons $Q_U$ and $Q_L$.
Suppose we wish to guard an arbitrary vertex $z$ of $Q_U$ using an edge guard. 
Then, a guard must be placed at an edge of $Q$ that belongs fully or even partially to $\mathcal{VP}(z)$. 
Henceforth, let $\mathcal{EVP}(z)$ denote the set of all polygonal edges that contain at least one point belonging to $\mathcal{VP}(z)$. 
Further, let us define the \emph{inward visible edges} and the \emph{outward visible edges} of $z$, denoted by $\mathcal{EVP}^{+}(z)$ and $\mathcal{EVP}^{-}(z)$ respectively, as follows. 

$\mathcal{EVP}^{+}(z) = \{ e \in \mathcal{EVP}(z)$: {\it for any point $x \in e$, the segment $zx$ does not intersect $uv$}\} 

$\mathcal{EVP}^{-}(z) = \{ e \in \mathcal{EVP}(z)$ : {\it for any point $x \in e$, the segment $zx$ intersects $uv$}\}

We shall henceforth refer to the vertex guards belonging to $\mathcal{EVP}^{+}(z)$ and $\mathcal{EVP}^{-}(z)$ as \emph{inside guards} and \emph{outside guards} for $z$ respectively. \\

\vspace*{-0.55em}
Observe that, for any vertex $z_k \in Z$, both $\mathcal{EVP}^{+}(z_k)$ and $\mathcal{EVP}^{-}(z_k)$ may be considered to be ordered sets by taking into account the natural ordering of the visible edges of $Q$ in clockwise order along $bd_{c}(u,v)$ and in counter-clockwise order along $bd_{cc}(u,v)$ respectively. 
Let us denote the \emph{first visible edge} and the \emph{last visible edge} belonging to the ordered set $\mathcal{EVP}^{-}(z_k)$ 
by $f(z_k)$ and $l(z_k)$ respectively (see Figure \ref{need_inside}).
Also, we denote by $l'(z_k)$ the \emph{last visible point} from $z_k$, which is obtained by extending the ray $\overrightarrow{z_k p(v,z_k)}$ 
till it touches $bd_{cc}(u,v)$. \\

\vspace*{-0.55em}
As in Section \ref{vertex_algo}, 
our algorithm selects a subset $Z$ of vertices of $Q_U$, and places a fixed number of both inside and outside edge guards corresponding to each of them, 
so that these edge guards together see the entire $Q_U$. As before, we refer to this subset of special vertices as \emph{primary vertices}, and denote it by $Z$. 
Moreover, the choice of primary vertices is also made in a manner identical to that discussed in Section \ref{vertex_algo}. \\

\vspace*{-0.55em}
Once again, a natural idea is to place outside edge guards in a greedy manner so that
they lie in the common intersection of outward visible edges of as many vertices of $Q_U$ as possible. 
For any primary vertex $z_k$, let us denote by $\mathcal{OEV}^{-}(z_k)$ the set of unmarked vertices of $Q_U$ whose outward visible edges overlap with those of $z_k$. 
In other words, 
\vspace{-0.55em}
$$ \mathcal{OEV}^{-}(z_k) = \{ x \in \mathcal{V}(Q_U) : \mbox{ $x$ is unmarked, and } \mathcal{EVP}^{-}(z_k) \cap \mathcal{EVP}^{-}(x) \neq \emptyset \} $$
So, each vertex of $Q_U$ belonging to $\mathcal{OEV}^{-}(z_k)$ is visible from at least one edge of $\mathcal{EVP}^{-}(z_k)$.
Further, $\mathcal{OEV}^{-}(z_k)$ can be considered to be an ordered set, where for any pair of elements $x_1,x_2 \in \mathcal{OEV}^{-}(z_k)$, we define $x_1 \prec x_2$ if and only if $l'(x_1)$ precedes $l'(x_2)$ in counter-clockwise order on $bd_{cc}(u,v)$. 
For the current primary vertex $z_k$, let us assume without loss of generality that $\mathcal{OEV}^{-}(z_k) = \{ x^k_1, x^k_2, x^k_3, \dots\, x^k_{m(k)}\}$ such that $l'(x^k_1) \prec l'(x^k_2) \prec \dots \prec l'(x^k_{m(k)})$ in counter-clockwise order on $bd_{cc}(u,v)$. \\

\vspace*{-0.55em}
Just as we partitioned $\mathcal{OVV}^{-}(z_k)$ while guarding the polygon using vertex guards, 
let us partition the vertices belonging to $\mathcal{OEV}^{-}(z_k)$ into 3 sets, viz. $A^k$, $B^k$ and $C^k$, in the following manner. 
Consider any vertex $x^k_i \in \mathcal{OEV}^{-}(z_k)$,
such that $\mathcal{EVP}^{-}(x^k_{i})$ creates a constructed edge $t(x^k_i) t'(x^k_i)$, where $t(x^k_i) \in \mathcal{V}(Q_L)$ is a polygonal vertex and $t'(x^k_i)$ is the point where $\overrightarrow{x^k_i t(x^k_i)}$ first intersects $bd_{cc}(u,v)$. 
Every vertex of $\mathcal{OEV}^{-}(z_k)$ visible from $l(z_k)$ is included in $B^k$.
Observe that, by definition $z^k \in B^k$.
Obviously, for each vertex $x^k_i \in \mathcal{OEV}^{-}(z_k) \setminus B^k$, 
$x^k_i$ is not visible from $l(z_k)$ due to the presence of some constructed edge.
The vertices of $\mathcal{OEV}^{-}(z_k) \setminus B^k$ are categorized into $A^k$ and $C^k$ based on whether this constructed edge creates a right pocket or a left pocket.
Suppose $x^k_i \in \mathcal{OEV}^{-}(z_k) \setminus B^k$ is a vertex such that $\mathcal{VP}(x^k_{i})$ creates a constructed edge $t(x^k_i) t'(x^k_i)$, 
where $t(x^k_i) \in \mathcal{V}(Q_L)$ is a polygonal vertex and $t'(x^k_i)$ is the point where $\overrightarrow{x^k_i t(x^k_i)}$ first intersects $bd_{cc}(u,v)$. 
If $t(x^k_i)$ lies on $bd_{cc}(f'(z_k),l'(z_k))$ and $t'(x^k_i)$ lies on $bd_{cc}(l(z_k),v)$,
i.e. if $f(z_k) \prec t(x^k_i) \prec l'(z_k)$ and $l(z_k) \prec t'(x^k_i) \prec v$,
then $x^k_i$ is included in $A^k$. 
For instance, in Figure \ref{fig_case1}, $x_1^2 \in A^2$ due to the constructed edge $t(x_1^2)t'(x_1^2)$.
On the other hand, if $t(x^k_i)$ lies on $bd_{cc}(l'(z_k),v)$ and $t'(x^k_i)$ lies on $bd_{cc}(f(z_k),l'(z_k))$,
i.e. if $l'(z_k) \prec t(x^k_i) \prec v$ and $f(z_k) \prec t'(x^k_i) \prec l'(z_k)$,
then $x^k_i$ is included in $C^k$.
For instance, in Figure \ref{fig_case1}, $x_3^2 \in C^2$ due to the constructed edge $t(x_3^2)t'(x_3^2)$.
Observe that, all vertices of $A^k$ must lie on $bd_c(u,z_k)$, 
whereas all vertices of $C^k$ must lie on $bd_c(z_k,v)$.
We have the following lemma, whose proof is similar to the proof of Lemma \ref{B_property}. 

\begin{lemma} \label{B_property_edge}
The edge $l(z_k)$ sees all vertices belonging to $B^k$.
\end{lemma}

Depending on the vertices in $A^k$, $B^k$ and $C^k$, 
we have the following cases, just as in Section \ref{only_lower}.

\begin{description}

\item[Case 1 -] $\mathcal{CI}(A^k \cup B^k \cup C^k) \neq \emptyset$ (see Figure \ref{fig_case1})

\item[Case 2 -] $\mathcal{CI}(A^k \cup B^k \cup C^k) = \emptyset$ and $\mathcal{CI}(B^k) \neq \emptyset$ 

\item[\hspace{11mm} Case 2a -] $\mathcal{CI}(A^k) \neq \emptyset$ and $\mathcal{CI}(C^k) \neq \emptyset$ (see Figure \ref{fig_case2a})

\item[\hspace{11mm} Case 2b -] $\mathcal{CI}(A^k) = \emptyset$ and $\mathcal{CI}(C^k) \neq \emptyset$ (see Figure \ref{fig_case2b})

\item[\hspace{11mm} Case 2c -] $\mathcal{CI}(A^k) \neq \emptyset$ and $\mathcal{CI}(C^k) = \emptyset$ (see Figure \ref{fig_case2c})

\item[\hspace{11mm} Case 2d -] $\mathcal{CI}(A^k) = \emptyset$ and $\mathcal{CI}(C^k) = \emptyset$ (see Figure \ref{fig_case2d})

\end{description}

In each of the above cases, we also proceed to choose the three outside edge guards $s^k_1$, $s^k_2$ and $s^k_3$ in a similar manner as was done in Section \ref{only_lower}, with the only difference being that here we use the new definitions of $\mathcal{EVP}(z_k)$ and $\mathcal{OEV}(z_k)$ instead of $\mathcal{VVP}(z_k)$ and $\mathcal{OVV}(z_k)$ respectively.
To elaborate further, we choose the edge guard $s_3^k = l(z_k)$; if the guards $s^k_1$ and $s^k_2$ are chosen from a common intersection region, then we use  $\mathcal{EVP}^{-}(z_k)$ rather than $\mathcal{VVP}^{-}(z_k)$; if the guards $s^k_1$ and $s^k_2$ are chosen greedily, then we use the last edge along the traversal after which some vertex belonging to $A^k$ and $C^k$ respectively is no longer visible. 
We can establish lemmas similar to Lemmas \ref{case1} to \ref{Z_ub} (in Section \ref{only_lower}) from these choices of edge guards. \\

\vspace*{-0.5em}
As far as the choice of inside edge guards is concerned, we again choose them following the same case analysis (see Figures \ref{inner_shared}, \ref{inner_disjoint} and \ref{inner_nested}) as described in Section \ref{upper_n_lower}. 
To elaborate further, we choose the edge guards $s^k_4$ and $s^k_5$ to be the edges adjacent to $p(u,z_k)$ and $p(v,z_k)$ respectively that lie on $bd_c(p(u,z_k),p(v,z_k))$, 
whereas for the greedily chosen guard $s^k_6$ (if required at all) we use the last edge along the traversal after which some vertex belonging to $\mathcal{OEV}^{+}(z_k)$ is no longer visible. 
Again, we can establish lemmas similar to Lemmas \ref{disjoint} to \ref{gen_bound} (in Section \ref{upper_n_lower}) from these choices of edge guards. 
As a consequence, we obtain the following theorems. 

\begin{theorem} \label{Z_bound_edge}  
Let $Z$ be the set of primary vertices chosen by our modified algorithm,
and let $S$ be the set of all edge guards placed by it.
Then, $|S| \leq 6\cdot|Z| \leq 12\cdot|G_{opt}|$.
\end{theorem}

\begin{theorem} \label{runtime_edge}
It is possible to compute in $\mathcal{O}(n^4)$ time a set of inside and outside edge guards for guarding all vertices of a weak visibility polygon $Q_U$ such that the number of edge guards chosen is at most $12\cdot|G_{opt}|$, where $G_{opt}$ is an optimal edge guard cover for all vertices of $Q_U$.
\end{theorem}

In Theorem \ref{runtime_edge}, we established the existence of an approximation algorithm for guarding a polygon $Q$ weakly visible from a chord $uv$.
If this same algorithm is executed on the union of overlapping weak visibility polygons $Q$, then it chooses primary vertices in the same way irrespective of chords in $Q$, thereby producing a set of edge guards that sees the entire union $Q$. 
So, if this algorithm is used for every overlapping weak visibility polygon in $P$, then the entire polygon $P$ can be guarded by the union of the guard sets produced for guarding these overlapping weak visibility polygons.
Note that there is no increase in running time for this overall algorithm.
Moreover, since each vertex can appear in at most two weak visibility polygons from the hierarchy $W$, 
we obtain the following theorem. 

\begin{theorem}
Let $P$ be a simple polygon having $n$ vertices. 
Then, an edge guard set $S$ for guarding all vertices of $P$ can be computed in $\mathcal{O}(n^4)$ time,
such that $|S| \leq 18 \times|G_{opt}|$, where $G_{opt}$ is an optimal edge guard cover for all vertices of $P$. 
\end{theorem}

However, it may not always be true that the guards in $S$ see all interior points of $Q_U$. 
Consider the polygon shown in Figure \ref{ icell_edge_1 }. 
Assume that our previous algorithms places guards at $s_4^k$ and $s_5^k$, and all vertices of $bd_c(p(u,z_k),p(v,z_k))$ become visible from $s_4^k$ or $s_5^k$. 
However, the triangular region $Q_U \setminus (VP(p(u,z_k))) \cup VP(p(v,z_k))$, bounded by the segments $x_1 x_2$, $x_2 x_3$ and $x_3 x_1$, 
is not visible from $s_4^k$ or $s_5^k$. Also, one of the sides $x_1 x_2$ of the triangle $x_1 x_2 x_3$ is a part of the polygonal edge $a_1 a_2$. 
In fact, for any such region invisible from edge guards $s_4^k, s_5^k \in S^k$ corresponding to some $z_k \in Z$, 
henceforth referred to as an \emph{invisible cell}, one of the sides must always be a part of a polygonal edge.
The polygonal edge which contributes as a side to the invisible cell is referred to as its corresponding \emph{partially invisible edge}. \\

Observe that $s_4^k$ and $s_5^k$ can in fact create several invisible cells,
in a manner very similar to that shown in Figure \ref{pockets}.
Each invisible cell must be wholly contained within the intersection region (which is a triangle) of a left pocket and a right pocket. 
For example, in Figure \ref{ icell_edge_1 }, 
the invisible cell $x_1 x_2 x_3$ is actually the entire intersection region of the left pocket of $VP(s_4^k)$ and the right pocket of $VP(s_5^k)$. 
In general, where $VP(s_4^k)$ has several left pockets and $VP(s_5^k)$ has several right pockets which intersect pairwise to create multiple invisible cells (as shown in Figure \ref{pockets}), 
every such cell can be seen by placing guards on the common vertices between adjacent pairs of cells. 
Further, if $G_{opt}$ is also constrained to guard these invisible cells using only inward guards from $Q_U$, then the number of such additional guards required can be at most twice of $G_{opt}$,
as shown by Bhattacharya et al. \cite{AGWVP}. 
However, in the absence of any constraint on placing guards, $G_{opt}$ may place an outside guard in $Q_L$ that sees several such invisible cells. 
So, it is natural to explore the possibility of being able to guard all such invisible cells by using additional edge guards from $Q_L$, 
in combination with guards from $Q_U$. \\

Just as in Section \ref{interior}, we present a modified algorithm that ensures that all partially invisible edges are guarded completely, and therefore the entire $bd_c(u,v)$ is guarded.
Let us compute the weak visibility polygons corresponding to every edge of $P$. 
Then, the non-vertex endpoints of all the constructed edges belonging to these weak visibility polygons 
partition the boundary of $Q_U$ into distinct intervals called \emph{minimal visible intervals}. 
We have the following lemmas.

\begin{lemma} \label{mvi_edge}
Every minimal visible interval on the boundary of $Q_U$ is either entirely visible from an edge or totally not visible from that edge.
\end{lemma}

The modified algorithm first computes all minimal visible intervals and 
chooses one internal representative point from each minimal visible interval on the boundary of $Q_U$. 
These representative points are referred to as \emph{pseudo-vertices}. 
Alongside the original polygonal vertices of $Q$, all pseudo-vertices are introduced on the boundary of $Q_U$, and the modified polygon is denoted by $Q'$. Note that the endpoints of minimal visible intervals are not introduced in $Q'$.
In the modified algorithm, the pseudo-vertices of $Q'$ are treated in the same manner as the original vertices. 
We compute $\mathcal{EVP}^{+}(z)$ and $\mathcal{EVP}^{-}(z)$ for all vertices of $Q'$, 
irrespective of whether it is a pseudo-vertex, 
and some of the psuedo-vertices may even be chosen as primary vertices. 

\begin{theorem} \label{bound_pseudo_edge}   
For the edge guard set $S$ computed by our modified algorithm with pseudo-vertices, 
$|S| \leq 6\cdot|Z| \leq 12\cdot|G_{opt}|$,
where $G_{opt}$ is an optimal edge guard cover for the entire boundary of $Q_U$.  
\end{theorem}

\begin{theorem}
Let $P$ be a simple polygon having $n$ vertices. 
Then, an edge guard set $S$ for guarding the entire boundary of $P$ can be computed in $\mathcal{O}(n^5)$ time,
such that $|S| \leq 18 \times|G_{opt}|$, where $G_{opt}$ is an optimal edge guard cover for the entire boundary of $P$. 
\end{theorem}

\begin{proof}
 This result follows directly from Theorem \ref{bound_pseudo_edge}.
 The running time bound of $\mathcal{O}(n^5)$ can be established through arguments very similar to those in the proof of Theorem \ref{runtime_pseudo}. 
 The only difference lies in the computation of $\mathcal{EVP}(z_k)$ and $\mathcal{OEV}(z_k)$, in place of $\mathcal{VVP}(z_k)$ and $\mathcal{OVV}(z_k)$ respectively, for every primary vertex $z_k \in Z$, but the time complexity of these computations are also same. 
\end{proof}

\begin{figure}[H]
\begin{minipage}{0.49\textwidth}
    \centering{\includegraphics[width=\textwidth]{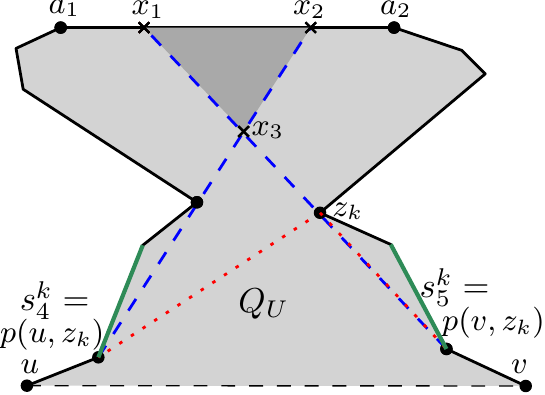}}
    \caption{ All vertices are visible from edge guards $s^k_4$ or $s^k_5$, but the triangle $x_1 x_2 x_3$ is invisible. }
    \label{ icell_edge_1 }
\end{minipage}
\hspace*{0.01\textwidth}
\begin{minipage}{0.49\textwidth}
    \centering{\includegraphics[width=\textwidth]{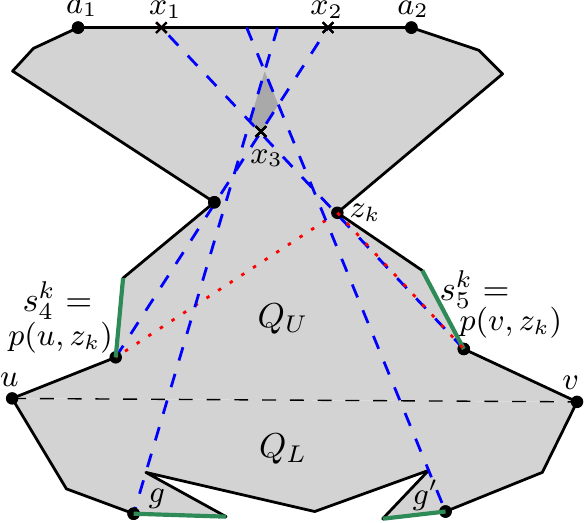}}
    \caption{ Two outside edge guards $g$ and $g'$ can create a residual invisible cell that is a part of $x_1x_2x_3$. }
    \label{ icell_edge_2 }
\end{minipage}
\end{figure}

\vspace*{-0.7em}
However, there is no guarantee that $S$ sees the entire interior of $Q_U$, as there may remain residual invisible cells in the interior of $Q_U$ (see Figure \ref{icell}).
Consider a residual invisible cell that is a part of an invisible cell $x_1x_2x_3$, where $x_1x_2$ is contained in a partially invisible edge.
For such a residual invisible cell, there exists a pseudo-vertex on $x_1x_2$ whose parents can see the entire cell $x_1x_2x_3$, as discussed earlier in the context of placing inside guards for guarding entire visibility cell. 
So, placing a guard at an appropriate parent, such as the parent $z_k$ for the pseudo-vertex that will be placed on the minimal visible interval $x_1x_2$ in Figure \ref{icell}, 
guarantees that the residual invisible cell is totally visible.
Since such an additional inside guard on $Q_U$ corresponds to an unique outward guard in $Q_L$, the additional number of inside guards can be at most the number of outside guards. 
This amounts to placing at most (3+3)=6 inside guards and 3 outside guards corresponding to each primary vertex, while the number of primary vertices chosen remains at most $2\cdot|G_{opt}|$. 
We summarize the result in the following theorem. 

\begin{theorem}
Let $P$ be a simple polygon having $n$ vertices. 
Then, an edge guard set $S$ for guarding the entire interior and boundary points of $P$ can be computed in $\mathcal{O}(n^5)$ time,
such that $|S| \leq 27 \times|G_{opt}|$, 
where $G_{opt}$ is a an optimal edge guard cover for the entire interior and boundary of $P$. 
\end{theorem}

\section{Concluding Remarks}
\label{conclude}

We have presented three approximation algorithms for guarding simple polygons using vertex guards. 
We have also shown how these algorithms can be modified to obtain similar approximation bounds while using edge guards. 
Though the approximation ratios for our algorithms are slightly on the higher side, they do successfully settle 
the long-standing conjecture by Ghosh by providing constant-factor approximation algorithms for this problem.
We feel that, in practice, our algorithms will provide guard sets that are much closer in size to an optimal solution. 
This can be further ensured by introducing a redundancy check after the placement of each new guard, 
which removes each (inside or outside) guard placed previously by the algorithm that does not see 
at least one vertex of $Q_U$ not seen by any other guard placed so far.
By incorporating such a redundancy check, we conjecture that our analysis of the approximation bound can be tightened further, 
and the existence of smaller approximation ratios can be proven for these three variations of the polygon guarding problem. 
Our algorithms exploit several deep visibility structures of simple polygons which are interesting in their own right.

\bibliographystyle{plain}
\bibliography{VGP_ref}

\end{document}